\documentclass[a4paper,english]{article}

\usepackage[a4paper, margin=1in]{geometry}
\usepackage{amsmath}
\usepackage{amsfonts}
\usepackage{amsthm}
\usepackage{amssymb}
\usepackage{mathtools}
\mathtoolsset{centercolon}
\usepackage{thmtools, thm-restate}

\usepackage{enumitem}
\setenumerate{label=(\arabic*), itemsep=2pt, topsep=5pt}

\usepackage{ifthen}
\usepackage{tikz}
\usepackage{url}
\usepackage{stmaryrd}
\usepackage{bbm}
\usepackage{etoolbox}
\usepackage{todonotes}
\usepackage{hyperref}

\usepackage[most]{tcolorbox}

\usetikzlibrary{decorations.pathmorphing}
\usetikzlibrary{decorations.pathreplacing}
\usetikzlibrary{calc}
\usetikzlibrary{positioning}
\usetikzlibrary{patterns}
\usetikzlibrary{patterns.meta}

\usetikzlibrary{backgrounds}

\renewcommand{\phi}{\varphi}
\renewcommand{\epsilon}{\varepsilon}
\newcommand{\nat}{\mathbb{N}}
\newcommand{\ZZ}{\mathbb{Z}}

\newcommand{\QQ}{\mathbb{Q}}

\newcommand{\defining}[1]{\emph{#1}}

\newcommand{\iso}{\cong}

\newcommand{\bigdisunion}{\biguplus}

\DeclarePairedDelimiter\set{\lbrace}{\rbrace}

\DeclarePairedDelimiterX\setcond[2]{\{}{\}}{\mathchoice{\,}{}{}{}#1 \;\delimsize\vert\; #2\mathchoice{\,}{}{}{}}

\newcommand{\tup}[1]{\bar{#1}}

\newcommand{\vertA}{u}
\newcommand{\vertB}{v}

\DeclareMathSymbol{\shortminus}{\mathbin}{AMSa}{"39}
\newcommand{\inv}[1]{#1^{\shortminus 1}}

\newcommand{\CFIsym}{\mathbf{G}}
\newcommand{\CFIA}[2]{\CFIsym_{#1}^{#2}}
\newcommand{\CFIB}[2]{\widetilde{\CFIsym}_{#1}^{#2}}

\newcommand{\sig}{\tau}
\newcommand{\arity}[1]{\operatorname*{ar}(#1)}
\newcommand{\StructA}{\mathbf{A}}
\newcommand{\StructB}{\mathbf{B}}
\newcommand{\StructC}{\mathbf{C}}
\newcommand{\StructL}{\mathbf{L}}

\newcommand{\CSP}[1]{\mathrm{CSP}(#1)}

\newcommand{\kcol}[3]{\kappa_{#1}^{#2}[#3]}
\newcommand{\restrict}[2]{#1|_{#2}}
\newcommand{\Hom}[2]{\mathrm{Hom}(#1,#2)}

\newcommand{\leqs}{\mathsf{L}}
\newcommand{\cspiso}[3]{\leqs^{#1,#2}_{\mathsf{CSP}}(#3)}
\newcommand{\zafkleq}[4]{\leqs^{#1,#2}_{\ZZ\mathsf{\text{-}aff}}(#3,#4)}
\newcommand{\ipk}[3]{\leqs^{#1,#2}_{\mathsf{IP}} (#3)}
\newcommand{\blk}[3]{\leqs^{#1,#2}_{\mathsf{BLP}} (#3)}
\newcommand{\aipk}[3]{\leqs^{#1,#2}_{\mathsf{AIP}} (#3)}
\newcommand{\blp}[2]{\leqs^#1_{\mathsf{BLP}} (#2)}

\newcommand{\bbN}{\mathbb{N}}
\newcommand{\bbZ}{\mathbb{Z}}
\newcommand{\Hh}{\mathcal{H}}

\newcommand{\autgrp}[1]{\operatorname*{Aut}(#1)}
\newcommand{\isos}[2]{\operatorname*{Iso}(#1,#2)}

\newcommand{\bcisosys}[2]{\mathbf{BI}(#1;#2)}
\newcommand{\colors}{\mathfrak{C}}
\newcommand{\Var}[1]{\operatorname*{Var}(#1)}

\newcommand{\CosetGrpTmplt}[2]{#1^{[#2]}}

\newcommand{\Sym}[1]{S_{#1}}
\newcommand{\SymStruct}[2]{\CosetGrpTmplt{\Sym{#1}}{#2}}

\newcommand{\NP}{\mathrm{NP}}
\newcommand{\Ptime}{\mathrm{P}}

\newcommand{\ORparam}[1]{\mathbf{OR}[#1]}
\newcommand{\ORT}[1]{\mathbf{OR}_\text{T}[#1]}
\newcommand{\ORNPC}[1]{\mathbf{OR}_\text{NPC}[#1]}
\newcommand{\OR}[1]{\mathbf{OR}[#1]}

\newcommand{\onestruc}[1]{\mathbf{1}_{#1}}

\newcommand{\ORISO}[2]{\mathbf{OR}^{\text{\upshape ISO}}_{#1}[#2]}

\newcommand{\CLAP}[1]{\mathsf{CLAP}(#1)}
\newcommand{\CLAPw}[1]{\mathsf{CLAP'}(#1)}

\colorlet{boxtitle}{lightgray!50!white}

\NewTColorBox{systembox}{O{} m}
{colback=white, colframe=black, sharp corners, boxrule=0.5pt, colbacktitle=white, colbacktitle=boxtitle, coltitle=black, title={\strut #2}, top=-0.6em, #1}

\NewTColorBox{algobox}{O{} m}
{colback=white, colframe=black, boxrule=0.5pt, colbacktitle=white, colbacktitle=boxtitle, coltitle=black, title={\strut #2}, top=-0.1em, #1, enhanced jigsaw}

\newtheorem{theorem}{Theorem}[section]
\newtheorem{lemma}[theorem]{Lemma}
\newtheorem{corollary}[theorem]{Corollary}
\newtheorem{definition}[theorem]{Definition}
\newtheorem{conjecture}[theorem]{Conjecture}

\theoremstyle{plain}
\newtheorem{claim}{Claim}
\usepackage{chngcntr}
\counterwithin*{claim}{lemma}

\AtBeginEnvironment{proof}{\setcounter{claim}{0}}
\newenvironment{claimproof}[1][\proofname]{\begin{claimprooftemp}[#1]}{\end{claimprooftemp}}

\newcommand\blfootnote[1]{%
	\begin{NoHyper}%
	\begingroup
	\renewcommand\thefootnote{}\footnote{#1}%
	\addtocounter{footnote}{-1}%
	\endgroup
\end{NoHyper}%
}

\newcommand{\orcid}[1]{\href{https://orcid.org/#1}{\includegraphics[height=1.8ex]{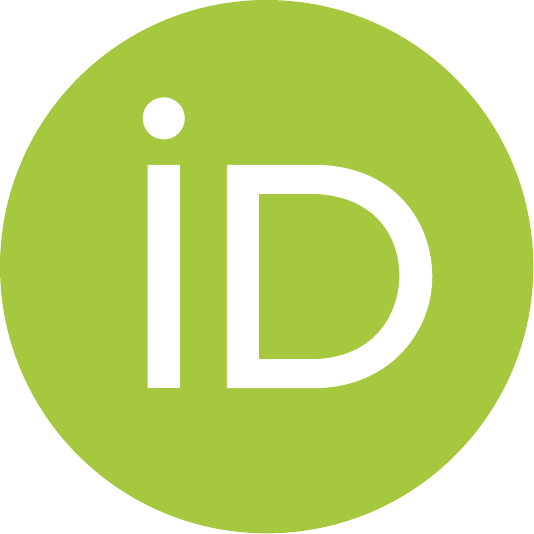}}}

\title{Limitations of Affine Integer Relaxations for Solving Constraint Satisfaction Problems
	}

\author{
	Moritz Lichter \orcid{0000-0001-5437-8074}\\ \small RWTH Aachen University, Germany	\and 
	Benedikt Pago \orcid{0000-0001-6377-1230}\\ \small University of Cambridge, UK
}
\date{}

\begin{document}

	\maketitle

	\begin{abstract}
		\noindent We show that various recent algorithms for finite-domain constraint satisfaction problems (CSP), which are based on solving their affine integer relaxations, do not solve all tractable and not even all Maltsev CSPs. This rules them out as candidates for a universal polynomial-time CSP algorithm.
		The algorithms are $\mathbb{Z}$-affine $k$-consistency, BLP+AIP, BA$^{k}$, and CLAP. We thereby answer a question by Brakensiek, Guruswami, Wrochna, and Živný~\cite{BrakensiekGWZ2020} whether a constant level of BA$^{k}$solves all tractable CSPs in the negative: Indeed, not even a sublinear level $k$ suffices.
		We also refute a conjecture by Dalmau and Opr\v{s}al \cite{DalmauOprsal2024} (LICS 2024) that every CSP is either solved by $\mathbb{Z}$-affine $k$-consistency or admits a Datalog reduction from 3-colorability.
		For the cohomological $k$-consistency algorithm, that is also based on  affine relaxations, we show that it correctly solves our counterexample but fails on an NP-complete template. 
		\blfootnote{The first author received funding from the European Research Council (ERC)
			under the European Union’s Horizon 2020 research and innovation programme (SymSim: grant
			agreement No. 101054974). Views and opinions expressed are however those of the author(s) only and
			do not necessarily reflect those of the European Union or the European Research Council. Neither
			the European Union nor the granting authority can be held responsible for them.\\
			The second author was funded by UK Research and Innovation (UKRI) under the UK government’s Horizon
			Europe funding guarantee: grant number EP/X028259/1}
	\end{abstract}

	\section{Introduction}
	Constraint satisfaction problems (CSPs) provide a general framework that encompasses a huge variety of different problems, from solving systems of linear equations over Boolean satisfiability to variants of the graph isomorphism problem.
	We view CSPs as homomorphism problems. A~CSP is defined by a relational structure~$\StructA$ called the \emph{template} of the CSP.
	An \emph{instance} is a structure~$\StructB$ of matching vocabulary and the question is whether there is a homomorphism from~$\StructB$ to~$\StructA$.
	We only consider \emph{finite-domain} CSPs, i.e., the template~$\StructA$ is always finite.
	It had long been conjectured by Feder and Vardi \cite{feder1993monotone} that every finite-domain CSP is NP-complete or in P.
	In 2017, the conjecture was confirmed independently by Bulatov \cite{bulatov} and Zhuk~\cite{zhuk}.
	The complexity of a CSP is determined by the polymorphisms (``higher-dimensional symmetries'') of its template.
	If the template has no a so-called  \emph{weak near-unanimity} polymorphism, then the corresponding CSP is NP-complete.
	For the other case,
	Bulatov and Zhuk presented sophisticated polynomial-time algorithms.
	A less involved algorithm had been known earlier for templates with a \emph{Maltsev} polymorphism~\cite{BulatovDalmau2006}.
	None of these algorithms is \emph{universal} in the sense that on input $(\StructB, \StructA)$ they decide whether~$\StructB$ maps homomorphically into $\StructA$ in time polynomial in both $|\StructB|$ and $|\StructA|$ \footnote{The cardinalities of the structures also count the number of tuples in the relations, not just the universe size.}. Instead, these are \emph{families} of algorithms, one for each template. %For each fixed $\StructA$, the respective algorithm then runs in polynomial time in $|\StructB|$. 
	The question whether there is a universal, and ideally ``simple'', algorithm for all tractable CSPs, or even just for all Maltsev CSPs, is still open.  
	
	One natural approach towards universal algorithms is via affine relaxations 
	of systems of linear equations over $\set{0,1}$: Given a template~$\StructA$, an instance~$\StructB$, and possibly a width parameter $k$, the existence of a homomorphism $\StructB \to \StructA$ is encoded
	into a system of linear equations. 
	 If the domain of the variables is relaxed from $\{0,1\}$ to~$\bbZ$, the system can be solved in polynomial time \cite{BachemRavindran, Schrijver1986}, and the transformation of the CSP into the equation system is also computationally easy. 
	Thus, if this integer relaxation were exact for all tractable CSPs, or at least an interesting subclass thereof, such as all Maltsev CSPs, then computing and solving it would constitute a universal polynomial-time algorithm for that class.
	Several algorithms based on this idea have been developed in recent years, motivated specifically by the study of \emph{Promise CSPs} \cite{BrakensiekGWZ2020, BrakensiekGuruswami19, CiardoZivny2023CLAP, DalmauOprsal2024}. This is a relatively new variant of CSPs which generalize for example the approximate graph coloring problem and are still not very well understood. 
	The algorithms can be applied just the same to classical CSPs, and not even for these, much is known about their power. 
	In the present paper, we prove strong limitations for all these algorithms and show that even for Maltsev CSPs, none of them is universal: We construct a template~$\StructA$ whose CSP is \emph{not solved} by these algorithms by providing instances~$\StructB$ that admit no homomorphism to~$\StructA$
	but which are accepted by the algorithms. 
	This also refutes a conjecture by Dalmau and Opr\v{s}al \cite{DalmauOprsal2024}, that we expand upon below.
	Our result is in stark contrast to the situation for \emph{valued CSPs}, an optimization version of CSP. For these, a surprisingly simple linear-algebraic algorithm solves all tractable cases optimally \cite{ThapperZivny}.
	
	Let us briefly introduce the algorithms that are addressed by our construction.
	All of them make use of (slightly) different systems of equations, which can all be reduced to the \emph{width-$k$ affine relaxation}.
	Given a template structure~$\StructA$, an instance $\StructB$, and a width $k \in \bbN$, the variables of the equation system are indexed with partial homomorphisms from induced size-$k$ substructures of $\StructB$ to $\StructA$. 
	A solution to the width-$k$ affine relaxation is thus an assignment of numerical values to partial homomorphisms. 
	The equations enforce a consistency condition, i.e., express that partial homomorphisms with overlapping domains receive values that fit together. 
	This is related to, but stronger than, the \emph{$k$-consistency} method: The $k$-consistency algorithm is a well-studied simple combinatorial procedure that checks for inconsistencies between local solutions and propagates these iteratively. 
	This solves the \emph{bounded width} CSPs (see e.g.\ \cite{feder1993monotone, barto2009constraint, barto2014}) but is not powerful enough to deal with \emph{all} tractable CSPs~\cite{AtseriasBulatovDalmau2007}. 
	The consistency conditions of the width\nobreakdash-$k$ affine relaxation are stronger in the sense that they enforce a \emph{global} notion of consistency rather than a local one.
	The algorithms that fail to solve our counterexample are the following:
	
	The \textbf{$\bbZ$-affine $k$-consistency algorithm}~\cite{DalmauOprsal2024} (Section~\ref{sec:zAffineConsistency}) runs the $k$-consistency procedure.
	All non-$k$-consistent partial homomorphisms are removed from the width-$k$ affine relaxation.
	The algorithm accepts the instance $\StructB$ if and only if this modified version of the width-$k$ affine relaxation has an integral solution.
	Dalmau and Opr\v{s}al~\cite{DalmauOprsal2024} conjectured that for all finite structures~$\StructA$, $\CSP{\StructA}$ is either Datalog$^\cup$-reducible to $\CSP{\bbZ}$ and thus solved by $\bbZ$-affine $k$-consistency for a fixed $k$, or 3-colorability is Datalog$^\cup$-reducible to~$\CSP{\StructA}$ (see Conjecture \ref{con:s3-or-Z}). Assuming P $\neq$ NP, the conjecture implies that every tractable finite-domain CSP is solved by $\bbZ$-affine $k$-consistency.
	
	The \textbf{BLP+AIP algorithm} by Brakensiek, Guruswami, Wrochna, and Živný~\cite{BrakensiekGWZ2020} (Section~\ref{sec:BLP}) first solves the width-$k$ affine relaxation over the non-negative rationals,
	where~$k$ is the arity of the template.
	Next, the integral width-$k$ affine relaxation is checked for a solution,
	but every variable is set to $0$ that is set to $0$ by every rational solution.
	The \textbf{BA$^k$-algorithm}  proposed by Ciardo and Živný~\cite{CiardoZivny2023BAk} (Section \ref{sec:BLP})
	generalizes BLP+AIP: The width $k$ is not fixed to be the arity of the template but is a parameter of the algorithm, like in $\bbZ$-affine $k$-consistency. In \cite{CiardoZivny2023BAk}, it is shown that there is an \emph{NP-complete} (promise) CSP on which the algorithm fails, but no tractable counterexample had been known until now.
	
	The \textbf{CLAP algorithm}, due to Ciardo and Živný~\cite{CiardoZivny2023CLAP} (Section \ref{sec:CLAP}),
	tests in the first step, for each partial homomorphism $f$,
	whether $f$ can receive weight exactly $1$ in a non-negative rational solution of the width-$k$ affine relaxation,
	where $k$ is the arity of the template.
	If not, it is discarded.
	This is repeated until the process stabilizes.
	Then the width-$k$ affine relaxation is solved over the integers,
	where all discarded partial homomorphisms are forced~to~$0$.

	\begin{theorem}
		\label{thm:mainResultInformal}
		There is a Maltsev template with $7$ elements that is neither solved by
		\begin{enumerate}
			\item $\bbZ$-affine $k$-consistency, for every $k \in o(n)$, where $n$ is the instance size,
			\item BLP+AIP,
			\item BA$^k$, for every $k \in o(n)$, nor
			\item  the CLAP algorithm.
		\end{enumerate}
		Hence, none of the algorithms solves all tractable CSPs.
	\end{theorem}	
	\noindent In particular, this answers a question of 
	Brakensiek, Guruswami, Wrochna, and Živný~\cite{BrakensiekGWZ2020} whether a constant level of the BA$^k$ hierarchy solves all tractable CSPs in the negative: Indeed, not even a sublinear level suffices.
	It also refutes the aforementioned Conjecture~\ref{con:s3-or-Z} regarding the power of the $\bbZ$\nobreakdash-affine $k$-consistency relaxation~\cite{DalmauOprsal2024}, under the assumption that P $\neq$ NP. 
	But we actually show a stronger statement: Namely, 3-colorability is not Datalog$^\cup$-reducible to the CSP that we use in the proof of the above theorem (Lemma~\ref{lem:not-datalog-reducible}). This is shown via a known inexpressibility result for \emph{rank logic}~\cite{GradelPakusa19} and disproves the conjecture unconditionally.
	
	To prove Theorem~\ref{thm:mainResultInformal}, in Sections \ref{sec:tseitin} and \ref{sec:power-of-affine}
	we construct and analyze instances. 
	Our template is a combination of systems of linear equations over the Abelian groups $\bbZ_2$ and $\bbZ_3$, but the~template itself is not a group.
	Since the affine algorithms reduce CSPs to a problem over the infinite Abelian group $(\bbZ,+)$, 
	we investigate for which finite groups this is possible:
	we study what we call \emph{group coset-CSPs} (to distinguish them from equation systems over groups).
	The template of a coset-CSP consists of a finite group $\Gamma$, and its relations are cosets of powers of subgroups of $\Gamma$.
	They always have a Maltsev polymorphism~\cite{BerkholzGrohe2015}.
	Coset-CSPs have been studied as ``group-CSPs'' by Berkholz and Grohe~\cite{BerkholzGrohe2015, BerkholzGrohe2017} or as ``subgroup-CSPs'' by Feder and Vardi~\cite{feder1993monotone}. 
	\begin{restatable}{theorem}{mainGroupCSPs}
		\label{thm:mainPowerOnGroupCSPs}	
		For each of the algorithms $\bbZ$-affine $k$-consistency, BLP+AIP, BA$^k$, and CLAP, the following is true:
		\begin{enumerate}
			\item Every coset-CSP over a finite Abelian group is solved (for $\bbZ$-affine $k$-consistency,~$k$ must be at least the arity of the template structure). \label{itm:powerOnGroupsSolveAbelian}
			\label{itm:mainPowerOnGroupCSPs-abelian}
			\item There exists a non-Abelian coset-CSP that is not solved, namely over $\Sym{18}$, the symmetric group on $18$ elements (for any constant or even sublinearly growing $k$).\label{itm:powerOnGroupsDontSolveNonAbelian}
			\item There are non-Abelian coset-CSPs that are solved, namely over any 2-nilpotent group of odd order. For example, there are non-Abelian $2$-nilpotent  semidirect products $\bbZ_{p^2} \rtimes \bbZ_p$ of order $p^3$ for each odd prime $p$. \label{itm:powerOnGroupsSolveSomeNonAbelian}
		\end{enumerate}
	\end{restatable}
	\noindent The detailed proof of this theorem is given in Section \ref{sec:groupStuff}.
	While Assertion~\ref{itm:mainPowerOnGroupCSPs-abelian} is easily derived from the literature~\cite{OConghaile22, BartoBKO2021}, it turns out somewhat surprisingly
	that Abelian groups are not the border of tractability for the affine algorithms:
	They also work over certain 2-nilpotent groups; these are in a sense the non-Abelian groups that are closest possible to being Abelian.
	Assertion~\ref{itm:powerOnGroupsDontSolveNonAbelian} is shown
	with a construction that is ``semantically equivalent''
	to the one that we use for Theorem \ref{thm:mainResultInformal}, but whose template is a coset-CSP.
	However, the analysis of the instances is technically much more involved.
	The construction in Theorem~\ref{thm:mainResultInformal} is simpler and yields a smaller template.
	We show that our first counterexample can be expressed as instances of the \emph{graph isomorphism} problem with \emph{bounded color class size},
	that is, the isomorphism problem of vertex-colored graphs, in which each color is only used for a constant number of vertices.
	This problem is expressible as a coset-CSP over the symmetric group \cite{BerkholzGrohe2017}.
	This also shows that the affine CSP algorithms cannot be adapted to solve the graph isomorphism problem. They fail already on the bounded color class version, which is known to be in P~\cite{FurstHopcroftLuks80}.
	
	There exists another highly interesting affine CSP algorithm that we have not addressed so far. This is the \emph{cohomological $k$\nobreakdash-consistency} algorithm due to Ó Conghaile \cite{OConghaile22} (see Section~\ref{sec:cohomology}). 
	As it turns out, this algorithm is actually able to solve our counterexample correctly.
	Hence, for all we know, it is possible that cohomological $k$-consistency is a universal polynomial-time algorithm for Maltsev or even all tractable CSPs, for a $k$ that suitably depends on the arity of the template. However, we can show \emph{without} complexity-theoretic assumptions that it fails on NP-complete CSPs, even if $k$ is a sublinearly growing function in the instance size.
	Recent work by Chan and Ng \cite{ChanNg} independently shows a similar result, but with a different technique: They prove that random instances of certain types of NP-complete templates are not solved by cohomological $k$-consistency, for every $k \in o(n)$. The difference between this and our lower bound is that we use a specifically designed template in our proof, while \cite{ChanNg} works generally for all templates satisfying two conditions called ``null-constraining'' and ``lax'', which are satisfied for example by hypergraph-colouring problems. In a preprint of this article, that the authors of \cite{ChanNg} refer to, we stated the lower bound only for constant $k \in \bbN$, but in the present version, we show that the same example is indeed hard for any sublinear $k$.
	\begin{theorem} 
		\label{thm:mainPowerOfCohomology}
		The CSP on which the algorithms in Theorem~\ref{thm:mainResultInformal} fail is solved by cohomological\\ $k$-consistency, for every $k \geq 4$.  
		There exists an NP-complete CSP that is not solved by cohomological $k$-consistency, for every $k \in o(n)$, where $n$ is the instance size.
	\end{theorem}	
	
	However, neither our techniques nor the ones from \cite{ChanNg} suggest an immediate route towards a tractable counterexample. 
	\begin{description}
		\item[Open question:] Is there a tractable finite template $\StructA$ such that cohomological $k$-consistency fails to solve $\CSP{\StructA}$ for every constant/sublinear $k$?
	\end{description}

	A very recent article by Zhuk \cite{zhuk2025singletonalgorithms} develops techniques that might be helpful to approach this question. It is concerned with various combinations of so-called ``singleton'' algorithms; these solve the linear-algebraic relaxations with fixed local solutions, which is precisely the feature of cohomological $k$-consistency that allows it to solve our counterexample. Hence, improving our understanding of the power of singleton algorithms and their hierarchies can possibly lead to stronger counterexamples for cohomological $k$-consistency.

	\paragraph*{Our Techniques.}
	Our proof of Theorem~\ref{thm:mainResultInformal}
	combines results due to Berkholz and Grohe \cite{BerkholzGrohe2017} with a new \emph{homomorphism or-construction} that encodes the disjunction of two CSPs.
	For a system of linear equations 
	to have an integral solution, it suffices to have a rational \emph{$p$-solution} and a rational \emph{$q$-solution} (for $p$ and $q$ coprime),
	in which all non-zero values are of the form $p^z$, with $z \in \bbZ$, or $q^z$, respectively.
	Thus, it suffices to design the instances in such a way that these two co-prime rational solutions exist. 
	For the algorithms that involve a width-parameter~$k$, the additional challenge is to make the construction robust so that it works against any choice of~$k$ (in our case it works even if~$k$ grows with the instance size).
	The \emph{Tseitin contradictions} \cite{Tseitin1983} over \emph{expander graphs} (see Section \ref{sec:tseitin}) achieve this robustness. 
	It is known that these cannot be solved by ``local'' algorithms, e.g., the $k$-consistency method, for any constant~$k$ \cite{AtseriasBulatovDalmau2007}.
	Berkholz and Grohe showed that
	the width-$k$ relaxation for unsatisfiable Tseitin contradictions over~$\ZZ_p$, for a prime~$p$, still has a $p$-solution.
	We combine two unsatisfiable Tseitin systems over~$\bbZ_2$ and~$\bbZ_3$ in the aforementioned homomorphism or-construction (Section \ref{sec:orConstruction}). This yields an unsatisfiable CSP instance whose width-$k$ relaxation has a $2$- and a $3$-solution and thereby also an integral solution.
	The reason why this approach fails for the cohomological algorithm (Theorem~\ref{thm:mainPowerOfCohomology})
	is that it solves the width-$k$ relaxation when a partial homomorphism is fixed. This fixing of local solutions reduces the homomorphism or-construction to just solving equations over $\bbZ_2$ and $\bbZ_3$, respectively, which the affine relaxation can do. To prove the second part of Theorem~\ref{thm:mainPowerOfCohomology}, we modify the homomorphism or-construction so that cohomology no longer solves it, but this also makes the template NP-complete.

\paragraph*{Acknowledgments.} 
We thank a number of people for helpful discussions and valuable input at various stages of this work, especially also for acquainting us with the problem:
We are grateful to Anuj Dawar, Martin Grohe, Andrei Krokhin, Adam Ó Conghaile, Jakub Opr\v{s}al, Standa \v{Z}ivn\'y, and Dmitriy Zhuk.
We are especially indebted to Michael Kompatscher,
who kindly provided us with the proof of Theorem~\ref{thm:mainPowerOnGroupCSPs} \nolinebreak\ref{itm:powerOnGroupsSolveSomeNonAbelian}.

	\section{Preliminaries}
	We write $[k]$ for $\set{1,\dots ,k}$.
	For $k \in \nat$ and a set $N$, let $\binom{N}{\leq k}$ be
	the set of  all subsets of~$N$  of size at most~$k$.
	A \defining{relational vocabulary} $\sig$ is a set of relation symbols $\set{R_1,\dots,R_k}$
	with associated arities $\arity{R_i}$.
	A \defining{relational $\sig$-structure} is a tuple $\StructA= (A,R_1^\StructA,\dots, R_k^\StructA)$ of a \defining{universe} $A$ and interpretations of the relation symbols
	such that $R_i^\StructA \subseteq A^{\arity{R_i}}$ for all $i \in [k]$.
	We use letters $\StructA$, $\StructB$, and $\StructC$ for finite relational structures.
	Their universes are denoted $A$,~$B$, and~$C$, respectively.
	If $\StructA$ is a structure and $X \subseteq A$, then $\StructA[X]$ denotes the induced substructure with universe $X$.
	
	For two $\sig$-structures $\StructA$ and $\StructB$, we write $\Hom{\StructA}{\StructB}$ for the set of \defining{homomorphisms} ${\StructA \to \StructB}$
	and $\isos{\StructA}{\StructB}$ for the set of \defining{isomorphisms} $\StructA \to \StructB$.
	
	A \defining{graph} $G=(V,E)$ is a binary $\set{E}$-structure,
	where we denote its \defining{vertex set} by $V(G)$ and its \defining{edge set} by $E(G)$.
	The graph $G$ is undirected if $E(G)$ is a symmetric relation
	and we write $uv$ for an edge incident to vertices $u$ and $v$.
	Unless specified otherwise, we consider undirected graphs.
	
	We use the letters $\Gamma$ and $\Delta$ for \defining{finite groups}
	and usually use letters $\alpha,\beta,\gamma$, and $\delta$ for group elements.
	For arbitrary groups, we write the group operation as multiplication.
	If we specifically consider Abelian groups, we write the group operation as addition. For the \defining{symmetric group} on $d$ elements, we write $\Sym{d}$.
	
	For an \defining{equation system} $\leqs$ over $K$ (where $K$ can be a finite group, $\QQ$, $\ZZ$, or the like and is specified in the context),
	we denote the set of its \defining{variables} by $\Var{\leqs}$.
	We use the letters $\Phi$ and $\Psi$
	for \defining{assignments} $\Var{\leqs} \to K$.
	By a system of linear equations we refer to,
	unless stated otherwise,
	a system over the rationals or integers.

	\paragraph{CSPs and Polymorphisms.}
	For a finite $\sig$-structure~$\StructA$, denote by $\CSP{\StructA}$
	the \defining{CSP with template~$\StructA$}, i.e., the class of finite $\sig$-structures~$\StructB$
	such that there is a homomorphism $\StructB \to \StructA$.
	We call a structure~$\StructB$ a $\CSP{\StructA}$-instance
	if~$\StructB$ has the same vocabulary as~$\StructA$.
	The complexity of $\CSP{\StructA}$, and also the applicability of certain algorithms, is determined by the \defining{polymorphisms} of the $\sig$\nobreakdash-structure~$\StructA$. An $\ell$-ary polymorphism is a map $p\colon A^{\ell} \to A$ such that for every $R \in \sigma$
	of arity $r=\arity{R}$ and all $\bar{a}_1,\dots,\bar{a}_\ell \in R^{\StructA}$,
	the tuple $(p(a_{11}, a_{21}, \dots, a_{\ell1}), \dots, p(a_{1r}, a_{2r}, \dots, a_{\ell r}))$ is also in~$R^{\StructA}$ (where~$a_{ij}$ denotes the $j$-th entry of the tuple~$\bar{a}_i$).
	The polymorphisms of a structure are closed under composition.
	A ternary operation $p$ is \defining{Maltsev} if it satisfies the identity $p(x,x,y) = p(y,x,x) = y$ for all inputs. 
	For a group $\Gamma$ the map $f(x,y,z) = x \inv{y} z$ is a typical example of a Maltsev operation.
	The templates with Maltsev polymorphisms form a subclass of all tractable CSPs \cite{BulatovDalmau2006}. 
	For more background on the algebraic approach to CSPs, see for example \cite{BartoKrokhinRoss}.

	\paragraph{Logics, Interpretations, and Reductions.}
	A logic $L$ defines $\CSP{\StructA}$ if there is an $L$\nobreakdash-formula~$F$
	such that each instance~$\StructB$ satisfies~$F$ if and only if $\StructB \in \CSP{\StructA}$.
	\defining{Inflationary fixed-point logic} (IFP) is the extension of first-order logic by an operator that defines inflationary fixed-points. Roughly speaking, this operator defines a $k$-ary relation $R$ from a formula $F(x_1,\dots, x_k)$ with free variables $x_1,\dots,x_k$, which itself uses $R$.
The fixed-point is iteratively computed starting from the empty relation
and adding in each iteration the tuples $(v_1,\dots,v_k)$ of the input structure to~$R$, for which the assignment $x_i \mapsto v_i$ satisfies~$F$.
This process is repeated until~$R$ stabilizes.
This will always occur because~$R$ only becomes larger. In particular, IFP can define connected components of graphs, which is not possible in pure first-order logic.
For a rigorous introduction of the logic we refer to~\cite[Chapter~8.1]{EbbinghausFlum1995},
formal details are not needed in this article.

	A \defining{logical interpretation} is a (partial) map from $\sigma$\nobreakdash-structures
	to $\tau$\nobreakdash-structures defined by logical formulas. If $L$ is a logic, then $L[\sigma, \tau]$ denotes the set of all $L$-interpretations from $\sigma$- to $\tau$-structures.
	For $d \in \bbN$, a $d$\nobreakdash-dimensional $L[\sigma, \tau]$-interpretation $I$ is a tuple consisting of formulas $\phi_\delta$, $\phi_\approx,$ and $\phi_R$ for all $R \in \tau$.
	Given a $\sigma$-structure $\StructA$, the interpretation $I$
	defines a $\tau$-structure $I(\StructA)$ in the following way.
	Let $B \coloneqq \{  \bar{a} \in A^d \mid \StructA \models \phi_\delta(\bar{a})  \}$.
	For every relation symbol $R \in \tau$, the relation $R$ in $I(\StructA)$ is defined as the set of all $\arity{R}$-tuples over $B$ satisfying $\phi_{R}$.
	Finally, the interpretation $I$ can also define an equivalence relation $\approx$ on $B$ (via the formula $\phi_\approx$),
	which has to be compatible with the defined relations,
	to take the quotient of the structure defined so far by $\approx$.
	This means that each $\approx$-equivalence class gets contracted into a single vertex.
	If $I$ does not define such an equivalence (i.e., the formula $\phi_\approx$ is omitted or trivial), it is called \defining{congruence-free}.
	For more formal details we also refer to~\cite{EbbinghausFlum1995}, but they are not needed.
	
	The notion of a logical interpretation can also be used as reduction between decision problems.
	Given two $\sig_i$-structures $\StructA_i$ (for $i \in [2]$),
	$\CSP{\StructA_1}$ is \defining{$L$-reducible} to $\CSP{\StructA_2}$
	if there is an $L[\sig_1,\sig_2]$-interpretation~$I$
	such that $\StructB \in \CSP{\StructA_1}$ if and only if $I(\StructB) \in \CSP{\StructA_2}$ for all $\CSP{\StructA_1}$-instances~$\StructB$ (of course this notion applies also to other means of reductions).
	
	Of particular interest in the context of CSP are Datalog-interpretations.
	Datalog can be seen as the existential positive fragment of IFP and we do not introduce it in this paper.
	We only note that every Datalog interpretation can be expressed by an IFP-interpretation (again see~\cite[Theorem~9.1.4]{EbbinghausFlum1995} for details).
	Dalmau and Opr\v{s}al~\cite{DalmauOprsal2024}
	also consider a variant of these reductions called \defining{Datalog$^\cup$ reductions}.
	These are a composition of congruence-free Datalog reductions (without inequality) and a so-called union gadget.
	Formally, Dalmau and Opr\v{s}al work with structures with disjoint sorts, and the union gadget allows to take unions of relations and of sorts. 
	When working in IFP, these sorts can for example be encoded with unary relations. An IFP-interpretation can then define the unification of sorts by defining the new unary relation as the union of the relevant unary relations in the input structure, and unions of other relations are also easily IFP-definable.
	Thus, every Datalog$^\cup$-reduction can be expressed as an IFP-interpretation.

	\paragraph{The $k$-Consistency Algorithm.}
	A well-known heuristic for CSPs is the $k$-consistency algorithm.
	For a template $\StructA$ and an instance $\StructB$,
	the $k$-consistency algorithm computes a map $\kcol{k}{\StructA}{\StructB}$ assigning to each $X \in \tbinom{B} {\leq k}$
	a set of partial homomorphisms $\StructB[X] \to \StructA$:
	it is the unique greatest fixed-point that satisfies
	the following properties for all  $Y \subset X \in \tbinom{B} {\leq k}$.
	 
	\begin{description}
		\item[Forth-Condition:]Every $f \in \kcol{k}{\StructA}{\StructB}(Y)$
		extends to some $g \in \kcol{k}{\StructA}{\StructB}(X)$,
		that is, $\restrict{g}{Y}=f$.
		\item[Down-Closure:] For every $g \in \kcol{k}{\StructA}{\StructB}(X)$,
		we have $\restrict{g}{Y} \in \kcol{k}{\StructA}{\StructB}(Y)$.
	\end{description}
	If $\kcol{k}{\StructA}{\StructB}(X) = \emptyset$
	for some $X \in \tbinom{B} {\leq k}$, then the algorithm rejects $\StructB$,
	otherwise it accepts.
	We remark that there are different versions of the $k$\nobreakdash-consistency algorithm in the literature, in particular there are ones in which the $k$-consistency algorithm considers partial homomorphisms whose domains have size $k+1$~\cite{AtseriasBulatovDalmau2007}.
	We follow the one given in~\cite{DalmauOprsal2024}.

	\paragraph{CSP-Relaxation via Affine Systems of Linear Equations.}
	We introduce a system of linear equations due to Berkholz and Grohe~\cite{BerkholzGrohe2017},
	which will be used to (approximately) solve CSPs.
	We transfer hardness results for this system
	to other systems used in the different algorithms.
	Let~$\StructA$ be a template structure and~$\StructB$ be an instance.
	The \defining{width-$k$ affine relaxation} $\cspiso{k}{\StructA}{\StructB}$
	aims to encode (approximately) whether $\StructB$ is in $\CSP{\StructA}$. 
	\begin{systembox}{$\cspiso{k}{\StructA}{\StructB}$: variables $x_{X,f}$
		for all $X \in \tbinom{B}{\leq k}$ and all $f \in  \Hom{\StructB[X]}{\StructA}$}
	\begin{align*}
		\sum_{\substack{f \in \Hom{\StructB[X]}{\StructA},\\ \restrict{f}{X\setminus\set{b}} = g}} x_{X,f} &=  x_{X\setminus{\set{b}},g}  &\text{for all } X \in \tbinom{B}{\leq k}, b \in X, g \in \Hom{\StructB[X\setminus\set{b}]}{\StructA} \label{eqn:csp-iso-agree}\tag{L1} \\
		x_{\emptyset,\emptyset }&= 1\label{eqn:csp-iso-empty}\tag{L2}
	\end{align*}
\end{systembox}
\noindent In Equation~\ref{eqn:csp-iso-empty}, $\emptyset$ denotes the unique homomorphism $\StructB[\emptyset] \to \StructA$.
If~$k$ is at least the arity of~$\StructA$, then
$\StructB \in \CSP{\StructA}$ if and only if $\cspiso{k}{\StructA}{\StructB}$
has a nonnegative integral solution (and actually a $\set{0,1}$-solution)~\cite{BerkholzGrohe2015}.
We will be mainly interested in
integral solutions of $\cspiso{k}{\StructA}{\StructB}$,
so without the non-negativity restriction.
Such solutions can be computed in polynomial time \cite{BachemRavindran}.
To show the existence of these solutions,
we consider special rational solutions:
	
	\begin{definition}
		For $p\in \bbN$, a \defining{$p$-solution} of a system of linear equations $\leqs$ with variables $\Var{\leqs}$ is a solution $\Phi \colon \Var{\leqs} \to \QQ$ 
		of $\leqs$  such that, for all  $x \in \Var{\leqs}$,
		$\Phi(x)=0$ or $\Phi(x) =p^i$ for some $i \in \ZZ$.
	\end{definition}
	\begin{lemma}[{\cite[Lemma 2.1]{BerkholzGrohe2017}}]
		\label{lem:p-q-solution-implies-integral}
		If $p$ and $q$ are coprime integers and
		a system $\leqs$ of linear equations over~$\QQ$ has a $p$-solution and a $q$-solution,
		then $\leqs$ has an integral solution,
		which is only non-zero for variables
		on which the $p$-solution or the $q$-solution is non-zero.
	\end{lemma}
	
		\begin{lemma}
		\label{lem:csp-iso-subsets}
		All solutions $\Phi$ of $\cspiso{k}{\StructA}{\StructB}$ satisfy for all $X \in \tbinom{B}{\leq k}$, $Y \subseteq X$, and $g \in \Hom{\StructB[Y]}{\StructA}$
		that 
		\[\sum_{\substack{f\in \Hom{\StructB[X]}{\StructA},\\\restrict{f}{Y} = g}} \Phi(x_{X,f}) = \Phi(x_{Y,g}). \]
		In particular, for all $X \in \tbinom{B}{\leq k}$, we have
		\[\sum_{f\in \Hom{\StructB[X]}{\StructA}} \Phi(x_{X,f}) = 1.\]
	\end{lemma}
	\begin{proof}
		Let $X = \set{v_1,\dots,v_\ell}$ and $Y= \set{v_1,\dots,v_j}$ for some $j \leq k$.
		The first claim is proven by induction on $\ell-j$ using equations of Type~\ref{eqn:csp-iso-agree}.
		The second claim follows as a special case for $Y=\emptyset$
		and Equation~\ref{eqn:csp-iso-empty}.
	\end{proof}

	\paragraph{Group Coset-CSPs.}
	\label{sec:groupCSP}

	Let $\Gamma$ be a finite group.
We define \defining{$\Gamma$-coset-CSPs}~\cite{BerkholzGrohe2015, feder1993monotone}, a class of CSPs, in which variables range over $\Gamma$ and the constraints are of the following form. 
For an $r$-tuple of variables $\tup{x} = (x_1, \dots, x_r)$,
an $r$-ary \defining{$\Gamma$-coset-constraint} is the constraint
$\tup{x} \in \Delta\delta$, where $\Delta \leq \Gamma^r$ is a subgroup
of $\Gamma^r$ and $\delta \in \Gamma^r$.
Hence, $\Delta\delta$ is a right coset of $\Gamma^r$.
When we use the term \defining{coset-CSP}, we refer to a $\Gamma$-coset-CSP in this sense.

It is known that, for each fixed $\Gamma$ and each fixed arity $r$, every $r$-ary $\Gamma$-coset-CSP is polynomial-time solvable \cite{feder1993monotone}.
For every finite group $\Gamma$ and every arity $r$, there is a template structure $\CosetGrpTmplt{\Gamma}{r}$ such that every $r$-ary $\Gamma$-coset-CSP
can be seen as a $\CosetGrpTmplt{\Gamma}{r}$-instance
and $\CSP{\CosetGrpTmplt{\Gamma}{r}}$ contains all $r$-ary $\Gamma$\nobreakdash-coset-CSPs that have a solution.
The tractability of $\CSP{\CosetGrpTmplt{\Gamma}{r}}$ can also be seen from the fact that $\CosetGrpTmplt{\Gamma}{r}$ admits a Maltsev polymorphism~\cite{BulatovDalmau2006}
(whose existence was already noted, but not made explicit, in~\cite{BerkholzGrohe2015}).
In fact, the universe of a CSP can be extended to a group such that
the CSP is a coset-CSP in this sense if and only if its template has the Maltsev polymorphism $f(x,y,z) = x\inv{y}z$.

\begin{lemma}
	\label{lem:group-csp-maltsev}
	For every finite template $\StructA = (A, R_1^\StructA,...,R_m^\StructA)$
	and every binary operation $\cdot \colon A \times A \to A$ such that $\Gamma = (A,\cdot)$ is a group,
	\begin{itemize}
		\item  the map $f: \Gamma^3 \to \Gamma$ defined by $f(x,y,z) = x\inv{y}z$ is a polymorphism of $\StructA$ if and only if 
		\item  each relation $R_i^{\StructA}$ is a coset of a subgroup of $\Gamma^{\arity{R_i}}$.
	\end{itemize}
\end{lemma}	
\begin{proof}
	For the backwards direction,
	let $\Delta\delta$ be a right coset of $\Gamma^r$
	and consider $r$-tuples
	\[
	(\alpha_1,\dots,\alpha_r), (\beta_1,\dots,\beta_r), (\gamma_1,\dots,\gamma_r) \in \Delta\delta.
	\]
	Then
	$(f(\alpha_1,\beta_1,\gamma_1), \dots, f(\alpha_r, \beta_r, \gamma_r))
	= (\alpha_1\inv{\beta_1}\gamma_1, \dots , \alpha_r\inv{\beta_r}\gamma_r)
	\in \Delta\delta\inv{\delta}\Delta\Delta\delta = \Delta\delta$
	because $\Delta$ is a subgroup of $\Gamma^r$.
	Hence,~$f$ is a polymorphism.
	For the other direction, suppose that~$f$ is a polymorphism of the $r$-ary relation $R_i^{\StructA}$. We can write $R_i^{\StructA} = K\gamma$, for some $K \subseteq \Gamma^r, \gamma \in \Gamma^r$ such that $K$ contains the neutral element $\bar{0} \in \Gamma^r$. It remains to show that $K$ is a subgroup of~$\Gamma^r$. Let $\alpha, \beta \in K$. Then $\alpha\gamma, \gamma, \beta\gamma \in R_i^{\StructA}$. Apply $f$ to these three tuples in this order. For each $j \in [r]$, we have $f(\alpha_j\gamma_j, \gamma_j, \beta_j\gamma_j) = \alpha_j \beta_j \gamma_j$. Because~$f$ is a polymorphism of $R_i^{\StructA}$, it follows that $\alpha \beta \gamma \in R_i^{\StructA} = K\gamma$. So $\alpha \beta \in K$, and~$K$ is a subgroup.
\end{proof}	

Thus, coset-CSPs are a natural class to study. In particular, being Maltsev, they are always tractable even if $\Gamma$ is non-Abelian. By contrast, for \emph{systems of linear equations}, we have NP-completeness if (and only if) $\Gamma$ is non-Abelian \cite{GOLDMANN}. 
Systems of linear equations over an \emph{Abelian} group $\Gamma$
can however be viewed as a $\Gamma$-coset-CSP:
A linear equation $x_1 + \cdots + x_k = \alpha$ for $\alpha \in \Gamma$
is equivalent to the $\Gamma$-coset-constraint
$(x_1,...,x_k) \in \Delta\delta_\alpha$,
where $\Delta = \setcond{ (b_1,...,b_k)}{b_1+\dots+b_k = 0}$,
and $\delta_\alpha = (\alpha,0,...,0)$.
Hence, when we consider equation systems over Abelian groups in Section~\ref{sec:tseitin}, we can treat them uniformly as coset-CSPs. 
For coset-CSPs over the cyclic group $\bbZ_p$, we will also need the (first-order definable) reverse translation from coset-CSP to linear equations: 
\begin{lemma}
	\label{lem:ZpcosetsAreEquations}
	Let $p$ be a prime and $\StructB$ an instance of $\CSP{\CosetGrpTmplt{\bbZ_p}{r}}$.
	Then there is a system of linear equations over $\bbZ_p$ that has a solution if and only if $\StructB \in \CSP{\CosetGrpTmplt{\bbZ_p}{r}}$. Moreover, the equation system is definable from $\StructB$ in first-order logic. 
\end{lemma}	
\begin{proof}
	Let $\bar{b} \in \Delta \gamma$ be an $r$-ary constraint in $\StructB$. 
	Let $\{\delta_1, ..., \delta_m\} \subseteq \Delta$ be a set of generators of the subgroup $\Delta$. Then an $\alpha \in \bbZ_p^r$ is in $\Delta \gamma$ if and only if it satisfies:
	$\alpha = \Big(\sum_{i \in [m]} z_i \cdot \delta_i\Big) + \gamma$, for some $z_i \in \bbZ_p$. 
	In this equation,~$\gamma$ and the~$\delta_i$ are $r$-tuples, so we can break this up into~$r$ many equations, one for each $j \in [r]$: $\alpha_j = \Big(\sum_{i \in [m]} z_i \cdot \delta_{ij}\Big) + \gamma_j$.
	Each constraint $\bar{b} \in \Delta \gamma$  in $\StructB$ is translated into this set of $r$ equations, with the~$z_i$ being the variables. Formally, we use different variables for each constraint, so the~$z_i$ are also indexed with the constraint $\bar{b} \in \Delta \gamma$ in $\StructB$ that they belong to.
	For each $\Delta \leq \bbZ_p^r$, we can use a fixed generating set, so with respect to this, the translation from coset constraints into the equations is first-order definable in~$\StructB$. 
\end{proof}

It is also known that if $\Gamma$ is Abelian, then the tractable CSPs over $\Gamma$ are precisely the $\Gamma$-coset-CSPs. For some non-Abelian groups $\Gamma$, there exist examples of tractable templates that contain non-coset relations. But, even then, we can only have tractability if the constraints are so-called ``nearsubgroup'' constraints (see \cite{feder1993monotone}).
So the $\Gamma$-coset-CSPs that we study here exactly cover the tractable regime for Abelian $\Gamma$, and nearly cover it for general $\Gamma$.

	\section{Homomorphism OR-Construction}
	\label{sec:orConstruction}
	
	We now describe a generic construction of templates and instances that will later be applied to specific Abelian coset-CSPs in order to obtain hard examples for the affine algorithms. Different variations of this construction will be used to prove Theorem \ref{thm:mainResultInformal} and the second part of Theorem \ref{thm:mainPowerOfCohomology}. The construction realizes the disjunction of two CSPs.
	For $i \in [2]$, let $\StructA_i$ and $\StructB_i$ be nonempty $\sig_i$-structures,
for which we assume that $\sig_1$ and $\sig_2$ are disjoint.
We see the $\StructA_i$ as template structures and the $\StructB_i$ as the corresponding instances.
We aim to define two structures $\StructA$ and $\StructB$ such that $\StructB \in \CSP{\StructA}$ if and only if there is an $i \in [2]$ such that $\StructB_i \in \CSP{\StructA_i}$.

Let $S$ be a fresh binary relation symbol.
Set $\sig := \sig_1 \cup \sig_2 \cup \set{S}$,
where the arities of the relations are inherited from $\sig_1$ and $\sig_2$.
Our construction is parameterized by subsets $W_i \subseteq A_i$ for each $i\in[2]$.
Different choices of these $W_i$ yield tractable and intractable versions of the OR-construction, which are needed for the proofs of Theorem \ref{thm:mainResultInformal} and \ref{thm:mainPowerOfCohomology}, respectively.
We also let $c_1$ and $c_2$ be two fresh vertices.
For $i\in[2]$ and $\ell\in\nat$, 
we let 
\[W_i^{\ell,c_i} := \setcond*{(u_1,\dots,u_k)\in (A_i\cup\set{c_i})^\ell}{u_j \in W_i, u_k = c_i \text{ for some } j,k \in [k]}\]
be the set of all $\ell$-tuples that containing $c_i$ and some element of $W_i$.
We define the $\sig$-structures $\StructA = \ORparam{\StructA_1,\StructA_2,W_1,W_2}$ and 
$\StructB = \OR{\StructB_1,\StructB_2}$.
In the following, we assume that the universes of~$\StructA_1$ and~$\StructA_2$
and the ones of~$\StructB_1$ and~$\StructB_2$ are disjoint (and non-empty), so that
the following unions are disjoint.
\begin{align*}
	A &:= A_1 \cup A_2 \cup \set{c_1,c_2}\\
	R^\StructA & :=R^{\StructA_i}  \cup   \set{c_i}^{\arity{R}} \cup W_i^{\arity{R}, c_i} & \text{for all } i \in [2], R \in \sig_i\\
	S^\StructA & :=  \Big( A_1 \times (W_2 \cup \set{c_{2}})\Big)
	\cup \Big( (W_1 \cup \set{c_1}) \times A_{2}\Big)\\\\
	B &:= B_1 \cup B_2\\
	R^\StructB &:= R^{\StructB_i}  & \text{for all } i \in [2], R \in \sig_i\\
	S^\StructB &:= B_1 \times B_2
\end{align*}

\begin{figure}
	\begin{tikzpicture}[scale=0.75]
		\colorlet{col1}{blue}
		\colorlet{col2}{red}
		\colorlet{colInstance}{red}
		
		\tikzstyle{template1} = [draw=col1, fill=col1!20!white, text=black, rectangle, minimum width=1.5cm, minimum height=1.5cm];
		\tikzstyle{template2} = [draw=col2, fill=col2!20!white, text=black, rectangle, minimum width=1.5cm, minimum height=1.5cm];
		\tikzstyle{instance1} = [draw=col1, fill=col1!20!white, text=black,  circle, minimum width=2cm, minimum height=2cm];
		\tikzstyle{instance2} = [draw=col2, fill=col2!20!white, text=black,  circle, minimum width=2cm, minimum height=2cm];
		\tikzstyle{vertex} = [fill=none, minimum width =0.5cm, minimum height = 0.5cm, thick, circle, inner sep = 0mm, font=\footnotesize];
		
		\tikzstyle{sedge} = [black, thick];
		
		\begin{scope}[shift={(7.5,12.5)}]
			\node [instance1] (b1) at (-2,-2) {$\StructB_1$};
			\node [instance2] (b2) at (2,-2) {$\StructB_2$};

			\begin{scope}[on background layer]	
				
				\def\ystep{0.5}
				\foreach \i in {-2,...,2}{
					\foreach \j in {-2,...,2} {
						\draw[thick, sedge] ($(b1) + (0,\i*\ystep)$) --($(b2) + (0,\j*\ystep)$);
					}
				}
			\end{scope}
			
			\node[text width=3cm, align=center] at (0,-4.25) {instance\\$\OR{\StructB_1,\StructB_2}$};
			
		\end{scope}

		\def\adstep{0.45}
		\def\avstep{0.45}
		\def\wstep{0.25}
		
		\begin{scope}[shift = {(0.5,0)}]
			\node [template1, vertex] (c1) at (-2,6) {$c_1$};
			\node [template2, vertex] (c2) at (2,6) {$c_2$};
			\node [template1] (a1) at (-2,2) {$\StructA_1$};
			\node [template2] (a2) at (2,2) {$\StructA_2$};
			\draw [col1, thick] ($(a1.north east)-(0.8,0)$) arc (180:270:0.8cm);
			\draw [col2, thick] ($(a2.north west)-(0,0.8)$) arc (270:360:0.8cm);
			\node[font=\footnotesize] (w1) at (-1.3,2.7)  {$W_1$};
			\node[font=\footnotesize] (w2) at (1.3,2.7)  {$W_2$};

			\begin{scope}[on background layer, every loop/.style={}]
				
				\foreach \i in {-2,...,2}{
					\draw[thick, sedge] ($(a1.center) + (\i*\adstep,-\i*\adstep)$) --(c2);
					\draw[thick, sedge] ($(a2) + (\i*\adstep,\i*\adstep)$) --(c1);
				}

				\foreach \i in {-2,...,2}{
					\foreach \j in {1,...,2} {
						\draw[thick, sedge] ($(a2) + (-1,\i*\avstep)$) --($(a1) + (1,\j*\avstep)$);
						\draw[thick, sedge] ($(a1) + (1,\i*\avstep)$) --($(a2) + (-1,\j*\avstep)$);
					}
				}
				
				\foreach \i in {1,...,2}{
					\draw[thick, col1] ($(a1) + (\i*\avstep,1)$) --(c1);
					\draw[thick, col2] ($(a2) + (-\i*\avstep,1)$) --(c2);
				}
				
				\draw [thick, col1] (c1) edge [loop] (c1);
				\draw [thick, col2] (c2) edge [loop] (c2);
			\end{scope}
			
			\node[text width=3.5cm, align=center] at (0,-0.25) {parameterized or-construction\\$\ORparam{\StructA_1,\StructA_2,W_1,W_2$}};
		\end{scope}
		
		\begin{scope}[shift = {(7.5,0)}]
			\node [template1, vertex] (c1) at (-2,6) {$c_1$};
			\node [template2, vertex] (c2) at (2,6) {$c_2$};
			\node [template1] (a1) at (-2,2) {$\StructA_1$};
			\node [template2] (a2) at (2,2) {$\StructA_2$};

			\begin{scope}[on background layer, every loop/.style={}]
				
				\foreach \i in {-2,...,2}{
					\draw[thick, sedge] ($(a1.center) + (\i*\adstep,-\i*\adstep)$) --(c2);
					\draw[thick, sedge] ($(a2) + (\i*\adstep,\i*\adstep)$) --(c1);
				}

				\draw [thick, col1] (c1) edge [loop] (c1);
				\draw [thick, col2] (c2) edge [loop] (c2);
			\end{scope}
			
			\node[text width=3cm, align=center] at (0,-0.25) {tractable\\ or-construction\\$\ORT{\StructA_1,\StructA_2}$};
		\end{scope}
		
		\begin{scope}[shift = {(14.5,0)}]
			\node [template1, vertex] (c1) at (-2,6) {$c_1$};
			\node [template2, vertex] (c2) at (2,6) {$c_2$};
			\node [template1] (a1) at (-2,2) {$\StructA_1$};
			\node [template2] (a2) at (2,2) {$\StructA_2$};

			\begin{scope}[on background layer, every loop/.style={}]
				
				\foreach \i in {-2,...,2}{
					\draw[thick, sedge] ($(a1.center) + (\i*\adstep,-\i*\adstep)$) --(c2);
					\draw[thick, sedge] ($(a2) + (\i*\adstep,\i*\adstep)$) --(c1);
				}

				\foreach \i in {-2,...,2}{
					\foreach \j in {-2,...,2} {
						\draw[thick, sedge] ($(a2) + (-1,\i*\avstep)$) --($(a1) + (1,\j*\avstep)$);
					}
				}
				
				\foreach \i in {-2,...,2}{
					\draw[thick, col1] ($(a1) + (\i*\avstep,1)$) --(c1);
					\draw[thick, col2] ($(a2) + (\i*\avstep,1)$) --(c2);
				}
				
				\draw [thick, col1] (c1) edge [loop] (c1);
				\draw [thick, col2] (c2) edge [loop] (c2);
			\end{scope}
			\node[text width=3cm, align=center] at (0,-0.25) {intractable or-construction\\$\ORNPC{\StructA_1,\StructA_2}$};
		\end{scope}
	\end{tikzpicture}
	\caption{The different homomorphism or-constructions:
		The figure assumes that the two vocabularies~$\sig_1$ and~$\sig_2$
		are binary and contain a single relation each (blue and red).
		At the top the instance $\OR{\StructB_1,\StructB_2}$.
		At the bottom three different version on the templates:
		the general parameterized construction,
		and the special cases of the tractable and intractable construction,
		which will be discussed in Sections~\ref{app:tractable-or} and~\ref{app:intractable-or}, respectively.
		The new $S$-relation is drawn in black,
		where the edges are all oriented from left to right.
		Pairs added to the relation of $\sig_1$ or $\sig_2$
		are drawn in blue or red, respectively.
		\label{fig:hom-or-construction}
	}
\end{figure}	
\noindent Figure~\ref{fig:hom-or-construction} illustrates the construction.
The following lemma shows that this definition yields a homomorphism or-construction, for all choices of the sets $W_i$.
As we will see, the choice of $W_i$ controls embeddings of partial homomorphisms 
and the complexity of the resulting template.
We will later work with two concrete instantiations of the sets $W_i$, but first we prove all properties that hold for any choice of $W_i$.

\begin{lemma}
	\label{lem:hom-or-construction-correct}
	$\StructB \in \CSP{\StructA}$ if and only if there is an $i \in [2]$
	such that $\StructB_i \in \CSP{\StructA_i}$.
\end{lemma}
\begin{proof}
	First, assume that there is a homomorphism $f \colon \StructB_i \to \StructA_i$.
	We define a homomorphism $g \colon \StructB \to \StructA$ via
	\begin{align*}
		g(b) := \begin{cases}
			f(b) & \text{if } b \in B_i, \\
			c_{3-i} & \text{otherwise.}
		\end{cases}
	\end{align*}
	We show that $g$ is a homomorphism.
	Let $R \in \sig_i$.
	Then $g(R^\StructB) = g(R^{\StructB_i}) = 
	f(R^\StructB_i) \subseteq R^{\StructA_i} \subseteq R^\StructA$.
	Let $R \in \sig_{3-i}$. Then $g(R^\StructB) = g(R^{\StructB_{3-i}}) =  \set{c_{3-i}}^{\arity{R}} \subseteq R^\StructA$.
	Finally, we consider the relation~$S$.
	We have $g(S^\StructB)  = g(B_1 \times B_2)\subseteq A_1 \times \set{c_2} \cup \set{c_1} \times A_2   \subseteq S^\StructA$
	by the definitions of $\StructA$ and $\StructB$.
	
	Conversely, assume that there is a homomorphism $f \colon \StructB \to \StructA$.
	Because $f$ preserves the relation~$S$,
	we have $f(B_{i}) \subseteq A_i \cup \set{c_i}$ for both $i \in [2]$.
	We  claim that for some $i \in [2]$, we actually have $f(B_{i}) \subseteq A_i$.
	Assume that for $i \in [2]$ this is not the case, that is,
	there is some $b \in B_i$ such that $f(b) = c_{i}$.
	By definition of $\StructB$, we have $\set{b} \times B_{3-i} \subseteq S^\StructB$ if $i=1$ and   $B_{3-i} \times \set{b} \subseteq S^\StructB$ if $i=2$.
	Because $f$ is a homomorphism and by the definition of $\StructA$,
	we have $f(\set{b} \times B_{3-i}) \subseteq \set{ c_{i}} \times A_{3-i}$
	if $i=1$ and similar for $i=2$.
	But this in particular implies that $f( B_{3-i}) \subseteq  A_{3-i}$ as claimed.
	
	So there is an $i \in [2]$ such that $f(B_{i}) \subseteq A_i$.
	We define $g \colon B_i \to A_i$ via $g(b) = f(b)$ for all $b \in B_i$.
	Because $R^\StructB = R^{\StructB_i}$ and $g(R^{\StructB_i}) \subseteq A_i^{\arity{R}}$, and $f$ is a homomorphism, $g$ maps tuples in~$R^{\StructB_i}$ to tuples in~$R^{\StructA_i}$, for all $R \in \sig_i$. 
	Hence the function $g$ is a homomorphism $\StructB_i \to \StructA_i$.
\end{proof}

\noindent We now analyze which partial homomorphisms $\StructB_i \to \StructA_i$
can be extended to partial or global homomorphisms $\StructB \to \StructA$.
For $i\in[2]$, denote by $\restrict{\StructA}{i}$ the structure $\StructA[A_i\cup\set{c_i}]$.

\begin{lemma}
	\label{lem:hom-or-embedd}
	Let $i \in [2]$, $X \subseteq B_i$, and
	$f \in \Hom{\StructB_i[X]}{\StructA_i}$.
	Then $f \in \Hom{\StructB[X]}{\restrict{\StructA}{i}}$.
\end{lemma}
\begin{proof}
	The relation $S$ is clearly preserved
	since $S^{\StructB[X]} =\emptyset$.
	Because $R^\StructB \subseteq B_i^{\arity{R}}$,
	the map $f$ also preserves $R$.
\end{proof}

\begin{lemma}
	\label{lem:hom-or-extend}
	Let $i \in [2]$.
	For every $X \subseteq B_i$ and $f\in \Hom{\StructB_i[X]}{\restrict{\StructA}{i}[(W_i\cup \set{c_i})]}$,
	the map $g \colon B_i \to W_i \cup \set{c_i}$ defined by
	\begin{align*}
		g(b) := \begin{cases}
			f(b) & \text{if } b \in X, \\
			c_i	& \text{otherwise.}
		\end{cases}
	\end{align*}
	is a homomorphism in  $\Hom{\StructB[B_i]}{\restrict{\StructA}{i}[W_i\cup \set{c_i}]}$.
	It satisfies in particular $\restrict{g}{X} = f$. 
\end{lemma}
\begin{proof}
	Clearly, $g$ preserves $S$ since $S^\StructB \cap B_i^2 = \emptyset$.
	Let $R \in \sig_i$ and $\tup{b} \in R^\StructB = R^{\StructB_i}$.
	\begin{itemize}
		\item If all elements of $\tup{b}$ are contained in $X$,
		then $g(\tup{b}) = f(\tup{b}) \in R^{\StructA_i} \subseteq R^\StructA$.
		\item If no elements of $\tup{b}$ are contained in $X$,
		then $g(\tup{b}) \in \set{c_i}^{\arity{R}} \subseteq R^\StructA$.
		\item Otherwise, some element of $\tup{b}$ is contained in $X$
		but not all of them. This means that $g(\tup{b})$ contains~$c_i$ at least once
		and at least one element of $W_i$.
		Hence $g(\tup{b}) \in W_i^{\arity{R},c_i} \subseteq R^\StructA$. \qedhere
	\end{itemize}
\end{proof}

\begin{lemma}
	\label{lem:hom-or-compose}
	Let $X_i \subseteq B_i$ and $f_i \in \Hom{\StructB[X_i]}{\restrict{\StructA}{i}}$ for both $i\in[2]$.
	Let $f\colon X_1\cup X_2 \to A$ be the map defined by each $f_i$ on $X_i$ for $i \in[2]$.
	If there is an $i\in[2]$ such that $c_i \notin f_i(X_i)$
	and $f_{3-i}(B_{3-i}) \subseteq W_{3-i} \cup \set{c_{3-i}}$,
	then $f \in \Hom{\StructB[X_1\cup X_2]}{\StructA}$.
\end{lemma}
\begin{proof}
	Let $j\in[2]$ and $R \in \sig_j$.
	Because $R^\StructB \subseteq B_j^{\arity{R}}$,
	the map $f$ also preserves $R$.
	It remains to show that $f$ preserves $S$.
	Assume that $c_1 \notin f_1(X_1)$
	and $f_{2}(X_{2}) \subseteq W_{2} \cup \set{c_{2}}$ (the case for $X_1$ and $X_2$ swapped is similar).
	Then, $f(S^\StructB \cap(X_1 \times X_2)) \subseteq A_1 \times (W_{2} \cup \set{c_2}) \subseteq S^{\StructA}$
	and thus the relation $S$ is preserved.
\end{proof}

\begin{corollary}
	\label{cor:hom-or-compose-simple}
	Let $i \in [2]$,  $f_i \in \Hom{\StructB_i}{\StructA_i}$, $X \subseteq B_{3-i}$ and 
	$f_{3-i} \in \Hom{\StructB_{3-i}[X]}{\StructA_{3-i}[W_{3-i}]}$.
	Then there is a $g \in \Hom{\StructB}{\StructA}$
	that agrees with $f_i$ on $B_i$ and with $f_{3-i}$ on $X$.
\end{corollary}
\begin{proof}
	By Lemma~\ref{lem:hom-or-embedd}, we have $f_i \in \Hom{\StructB[B_i]}{\restrict{\StructA}{i}}$.
	By Lemma~\ref{lem:hom-or-extend},
	$f_{3-i}$ extends to some $f' \in \Hom{\StructB[B_{3-i}]}{\restrict{\StructA}{3-i}[W_{3-i} \cup \set{c_{3-i}}]}$.
	Together, the homomorphisms$f_i$ and $f'$ yield the desired homomorphism $g$ by Lemma~\ref{lem:hom-or-compose}.
\end{proof}

\noindent Next, we show that the homomorphism or-construction is compatible
with $k$-consistency and solving $\cspiso{k}{\StructA}{\StructB}$
in the sense that if $k$-consistency
accepts $\StructB_i$, or $\cspiso{k}{\StructA_i}{\StructB_i}$
is satisfiable,
then $k$-consistency accepts $\StructB$, or $\cspiso{k}{\StructA}{\StructB}$
is satisfiable, respectively.

\begin{lemma}
	\label{lem:hom-or-k-consistency}
	Let $k \in \nat$, $X_i \subseteq B_i$  and $f_i \in \Hom{\StructB[X_i]}{\restrict{\StructA}{i}}$ for both $i\in[2]$
	such that $|X_1 \cup X_2|\leq k$.
	Let $f \colon X_1\cup X_2 \to A$ be the map induced by $f_1$ and $f_2$.
	If there is an $i\in[2]$ such that  
	\begin{itemize}
		\item $f_i \in \kcol{k}{\StructA_i}{\StructB_i}(X_i)$
		(so in particular $c_i \notin f_i(X_i)$),
		\item $f_{3-i}(X_{3-i}) \subseteq W_{3-i} \cup \set{c_{3-i}}$, and
		\item for $X_{3-i}' := \inv{f_{3-i}}(W_{3-i})$ we have $\restrict{f_{3-i}}{X_{3-i}'} \in \kcol{k}{\StructA_{3-i}}{\StructB_{3-i}}(X_{3-i}')$,
	\end{itemize}
	then $f \in \kcol{k}{\StructA}{\StructB}(X_1 \cup X_2)$.
\end{lemma}
\begin{proof}
	For a set $Z \subseteq B$, let $Z_j := Z \cap B_j$ for both $j \in[2]$.
	For all $X \subseteq \tbinom{B}{\leq k}$,
	let  $H(X) \subseteq \Hom{\StructB[X]}{\StructA}$ be the set 
	of all $f \colon X \to A$ for which there is an $i\in[2]$ such that
	$\restrict{f}{X_i} \in \kcol{k}{\StructA_i}{\StructB_i}(X_i)$,
	we have $f_{3-i}(B_{3-i}) \subseteq (W_{3-i} \cup \set{c_{3-i}})$,
	and for $X_{3-i}' := \inv{f_{3-i}}(W_{3-i})$
	we have $\restrict{f}{X_{3-i}'} \in \kcol{k}{\StructA_{3-i}}{\StructB_{3-i}}(X_{3-i}')$.
	By Lemma~\ref{lem:hom-or-extend}, the function $f_{3-i}$ satisfies
	$f_{3-i} \in \Hom{\StructB_{3-i}[X_{3-i}]}{\restrict{\StructA}{3-i}}$.
	Hence, the function $f$ is indeed a homomorphism in $\Hom{\StructB[X_1\cup X_2]}{\StructA}$
	by  Lemma~\ref{lem:hom-or-compose}.
	We show that the family of the $H(X)$ satisfies the Forth-Condition and the Down-Closure. Hence, the partial homomorphisms in the $H(X)$ are not
	discarded by the $k$-consistency algorithm.
	
	Let $Y \subset X \subseteq \tbinom{B}{\leq k}$.
	We first show the Forth-Condition.
	Let $f \in H(Y)$
	and $f_j = \restrict{f}{Y_j}$ for both $j\in[2]$.
	By construction, we have $\restrict{f}{X_i} \in \kcol{k}{\StructA_i}{\StructB_i}(X_i)$,
	$f_{3-i}(B_{3-i}) \subseteq (W_{3-i} \cup \set{c_{3-i}})$
	and for $X_{3-i}' := \inv{f_{3-i}}(W_{3-i})$ that
	$\restrict{f}{X_{3-i}'} \in \kcol{k}{\StructA_{3-i}}{\StructB_{3-i}}(X_{3-i}')$.
	By the Forth-Condition for $\StructB_i$, there is a $g_i \in  \kcol{k}{\StructA_i}{\StructB_i}(X_i)$ such that $\restrict{g}{X_i} = f_i$.
	Let $g_{3-i}$ be the extension of $f_{3-i}$ to $X_{3-i}$
	by $f_{3-i}(b) = c_{3-i}$ for all $b \in X_{3-i} \setminus Y_{3-i}$.
	Then the map $g\colon X_1\cup X_2 \to A$ induced by $g_1$ and $g_2$
	is in $H(X)$.

	We secondly show the Down-Closure.
	Let $g \in H(X)$, and $g_j := \restrict{g}{Y_j}$ for both $j \in[2]$.
	Let $g' \colon Y \to A$ be the function induced by $g_1$ and $g_2$.
	By construction, there is an $i\in[2]$ such that we have $g_i \in \kcol{k}{\StructA_i}{\StructB_i}(X_i)$,
	$g_{3-i}(B_j) \subseteq (W_{3-j} \cup \set{c_{3-j}}$,
	and for $X_{3-i}' := \inv{g_{3-i}}(W_{3-i})$ that
	$\restrict{g_{3-i}}{X_{3-i}'} \in \kcol{k}{\StructA_{3-i}}{\StructB_{3-i}}(X_{3-i}')$.
	By the Down-Closure for the $\StructB_j$,
	we have that $\restrict{g_i}{Y_i} \in \kcol{k}{\StructA_i}{\StructB_i}(Y_i)$
	and $\restrict{g_{3-i}}{Y_{3-i} \cap X_{3-i}'} \in \kcol{k}{\StructA_{3-i}}{\StructB_{3-i}}(Y_{3-i} \cap X_{3-i}')$.
	Clearly, we also have $\restrict{g_{3-i}}{Y_{3-i} \cap X_{3-i}'}(B_j) \subseteq (W_{3-j} \cup \set{c_{3-j}})$.
	So indeed, $g' \in H(Y)$.	
\end{proof}

\begin{lemma}
	\label{lem:hom-or-solution}
	Let  $i \in [2]$, let $\Phi$ be a solution to $\cspiso{k}{\StructA_i}{\StructB_i}$, and let $h \in \Hom{\StructB_{3-i}}{\restrict{\StructA}{3-i}[W_{3-i}\cup\set{c_{3-i}}]}$.
	Then the following map $\Psi$ is a solution to $\cspiso{k}{\StructA}{\StructB}$.
	Let $X \in \tbinom{\StructB}{\leq k}$, $X_j := X \cap B_j$ for $j\in[2]$, and $f \in \Hom{\StructB[X]}{\StructA}$.
	We set
	\[
	\Psi(x_{X,f}) := \begin{cases}
		\Phi(x_{X_i, \restrict{f}{X_i}}) & \text{if }
		\restrict{f}{X_i} \in \Hom{\StructB_i[X_i]}{\StructA_i} \text { and } \restrict{f}{X_{3-i}} = \restrict{h}{X_{3-i}},\\
		0 & \text{otherwise.}
	\end{cases}
	\]
	In particular, if $\Phi$ is an integral solution or a $p$-solution,
	then $\Psi$ is an integral solution or $p$-solution, respectively.
\end{lemma}
\begin{proof}
	We show that the equations of Type~\ref{eqn:csp-iso-agree} are satisfied by $\Psi$.
	Let $X \in \tbinom{B}{\leq k}$, $b \in X$, $X_j := X \cap B_j$ for $j\in[2]$,  and $g \in \Hom{\StructB[X \setminus \{b\}]}{ \StructA}$. By definition of $\Psi$, we have
	\begin{align*}
		\sum_{\substack{f \in \Hom{\StructB[X]}{\StructA},\\\restrict{f}{X\setminus \set{b}} = g}} \Psi(x_{X,f}) = 
		\sum_{\substack{f \in \Hom{\StructB[X]}{\StructA},\\\restrict{f}{X\setminus \set{b}} = g,\\
				\restrict{f}{X_i} \in \Hom{\StructB_i[X_i]}{\StructA_i},\\
				\restrict{f}{X_{3-i}} = \restrict{h}{X_{3-i}}
		}}
		\Phi(x_{X_i,\restrict{f}{X_i}}). \tag{$\star$}
	\end{align*}
	We make a case distinction.
	
	\begin{enumerate}
		\item Assume that $\restrict{g}{X_{3-i} \setminus \set{b}} \neq \restrict{h}{X_{3-i}\setminus \set{b}}$.
		Then $\restrict{f}{X_{3-i}} \neq \restrict{h}{X_{3-i}}$
		for every $f \in \Hom{\StructB[X]}{\StructA}$ such that $\restrict{f}{X\setminus \set{b}} = g$ 
		because $\restrict{f}{X_{3-i} \setminus \set{b}} =  \restrict{g}{X_{3-i} \setminus \set{b}} \neq \restrict{h}{X_{3-i}\setminus \set{b}}$. This implies that the summation in $(\star)$ sums over zero many numbers, and thus 
		\begin{align*}
			(\star) = 0 = \Psi(x_{X\setminus \set{b}, g}).
		\end{align*}
		\item Assume that $\restrict{g}{X_i \setminus \set{b}} \not\in \Hom{\StructB_i[X_i \setminus \set{b}]}{\StructA_i}$. Then $\restrict{f}{X_i} \not\in \Hom{\StructB_i[X_i]}{\StructA_i}$
		for every  $f \in \Hom{\StructB[X], \StructA}$ such that $\restrict{f}{X\setminus\set{b}} = g$
		because $\restrict{f}{X_i\setminus \set{b}} = \restrict{g}{X_i \setminus \set{b}} \not\in \Hom{\StructB_i[X_i \setminus \set{b}]}{\StructA_i}$.
		We again have $(\star) = 0 = \Psi(x_{X\setminus \set{b}, g})$ as in the case before.
		\item Assume that $b \in B_{3-i}$.
		By the cases before, we can assume that 
		$\restrict{g}{X_{3-i} \setminus \set{b}}= \restrict{h}{X_{3-i}\setminus \set{b}}$ and $\restrict{g}{X_i} \in \Hom{\StructB_i[X_i]}{\StructA_i}$.
		Then there is at most one $f'\in \Hom{\StructB[X]}{\StructA}$
		such that $\restrict{f'}{X\setminus \set{b}} = g$ and $\restrict{f'}{X_{3-i}} = \restrict{h}{X_{3-i}}$, namely the extension of $g$ by
		mapping $b$ to $h(b)$.
		By Lemmas~\ref{lem:hom-or-embedd} and~\ref{lem:hom-or-compose},
		we indeed have $f'\in \Hom{\StructB[X]}{\StructA}$.
		Thus, $f'$ is unique.
		We also have  $\restrict{f'}{X_i} \in \Hom{\StructB_i[X_i]}{\StructA_i}$
		because $\restrict{f'}{X_i} = \restrict{g}{X_i} \in \Hom{\StructB_i[X_i]}{\StructA_i}$. It follows
		\begin{align*}
			(\star) = \Phi(x_{X_i,\restrict{f'}{X_i}}) = \Phi(x_{X_i,\restrict{g}{X_i}}) = \Psi(x_{X,g}).
		\end{align*}
		\item Lastly, we assume that $b \in B_i$.
		By the cases before, we may assume that 
		$\restrict{g}{X_{3-i}} = \restrict{h}{X_{3-i}}$ and $\restrict{g}{X_i \setminus \set{b}} \in \Hom{\StructB_i[X_i \setminus \set{b}]}{\StructA_i}$.
		Since, as in the case before, there is a unique extension of a partial homomorphism in $\Hom{\StructB_i[X_i]}{\StructA_i}$
		to a partial homomorphism in $\Hom{\StructB[X]}{\StructA_i}$
		that agrees with $h$ on $X_{3-i}$, we have
		\begin{align*}
			(\star) = \sum_{\substack{
					\restrict{f}{X_i} \in \Hom{\StructB_i[X_i]}{\StructA_i},\\
					\restrict{f}{X_i \setminus \set{b} } = \restrict{g}{X_i \setminus \set{b}},\\
					\restrict{f}{X_{3-i}} = \restrict{h}{X_{3-i}}
			}}
			\Phi(x_{X_i,\restrict{f}{X_i}}) =  \Phi(x_{X_i\setminus {b}, \restrict{g}{X_i \setminus {b}}}) = \Psi(x_{X, g})
		\end{align*}
		because $\Phi$ is a solution to $\cspiso{k}{\StructA_i}{\StructB_i}$.
	\end{enumerate}
	It remains to verify that Equation~\ref{eqn:csp-iso-empty} is satisfied: $\Psi(x_{\emptyset,\emptyset}) = \Phi(x_{\emptyset,\emptyset}) = 1$
	by the definition of $\Psi$ and because $\Phi$ is a solution to $\cspiso{k}{\StructA_i}{\StructB_i}$.
\end{proof}

\subsection{The Tractable OR-construction}
\label{app:tractable-or}
We now consider the or-construction for the specific setting $W_1=W_2=\emptyset$.
For this choice of~$W_1$ and~$W_2$, the homomorphism or-construction yields a tractable CSP
if $\CSP{\StructA_i}$ is tractable for both $i\in[2]$. 
We refer to this construction as the \defining{tractable homomorphism or-construction} and write $\ORT{\StructA_1,\StructA_2} := \ORparam{\StructA_1,\StructA_2,\emptyset,\emptyset}$. This will later be used in the proof of Theorem \ref{thm:mainResultInformal}.
We start with corollaries from the lemmas of the previous subsection:

	\begin{lemma}
	\label{lem:hom-or-tractable-k-consistency}
	Let $\StructA=\ORT{\StructA_1,\StructA_2}$, $\StructB = \OR{\StructB_1,\StructB_2}$,
	$k \in \nat$, $i\in [2]$, $X \in \tbinom{B}{\leq k}$,
	and $f \in \Hom{\StructB[X]}{\StructA}$.
	If $f(X\cap B_{3-i}) = \set{c_{3-i}}$ and
	$\restrict{f}{X\cap B_i} \in \kcol{k}{\StructA_i}{\StructB_i}(X\cap B_i)$,
	then $f \in \kcol{k}{\StructA}{\StructB}(X)$.
\end{lemma}
\begin{proof}
	Immediately follows from Lemma~\ref{lem:hom-or-k-consistency}.
\end{proof}

\begin{lemma}
	\label{lem:hom-or-tractable-solution}
	Let $\StructA=\ORT{\StructA_1,\StructA_2}$, $\StructB = \OR{\StructB_1,\StructB_2}$, 
	$i \in [2]$, and $\Phi$ be a solution to $\cspiso{k}{\StructA_i}{\StructB_i}$.
	Then there is a solution $\Psi$ to $\cspiso{k}{\StructA}{\StructB}$
	defined, for every $X \in \tbinom{B}{\leq k}$
	and $f \in \Hom{\StructB[X]}{\StructA}$,
	by $\Psi(x_{X,f}) = \Phi(x_{X\cap B_i}, \restrict{f}{X\cap B_i})$ if  $f(X \cap B_{3-i}) = \set{c_{3-i}}$
	and $\Psi(x_{X,f}) = 0$ otherwise.
	In particular, $\Psi$ is a $p$-solution or integral,
	if $\Phi$ is a $p$-solution or integral, respectively.
\end{lemma}
\begin{proof}
	Follows from Lemma~\ref{lem:hom-or-solution}
	and the fact that the map $B_{3-i} \to \set{c_{3-i}}$ is a partial homomorphism.
\end{proof}

\noindent We now prove that the tractable or-construction indeed deserves its name because it generally preserves tractability of $\StructA_1$ and $\StructA_2$. In the special case of Maltsev templates, also this stronger condition is preserved: 

	\begin{lemma}
	\label{lem:maltsevForOrConstruction}
	If $\StructA_1$ and $\StructA_2$ have a Maltsev polymorphism, then $\ORT{\StructA_1,\StructA_2}$ has one.
\end{lemma}	
\begin{proof}
	Let $\StructA = \ORT{\StructA_1,\StructA_2}$.
	Let $f_1, f_2$ be the Maltsev polymorphisms of $\StructA_1, \StructA_2$, respectively.
	Define $f : A^3 \to A$ as follows: 
	\begin{align*}
		f(x,y,z) := \begin{cases}
			f_i(x,y,z) & \text{if } x,y,z \in A_i, \\
			c_{i} & \text{if } x,y,z \in A_i \cup \{c_i\} \text{ and exactly one or all are equal to } c_i,\\ 
			a& \text{otherwise, where } a \text{ is the left-most input}\\
			&\text{that does not occur exactly twice}
		\end{cases}
	\end{align*}
	It can be checked that $f$ is an (idempotent) Maltsev operation:  If all three inputs are in $A_i$, then this is inherited from $f_i$. If all inputs are in $A_i \cup \{c_i\}$, and one or three of them are equal to $c_i$, then the second case applies and so, $f(y,x,x) = f(x,x,y) = y$ and $f(x,x,x) = x$ are ensured. If all inputs are in $A_i \cup \{c_i\}$ and exactly two of them are equal to $c_i$, then the third case applies and produces the correct outcome. If the inputs are mixed between $A_1 \cup \{c_1\}$ and $A_2 \cup \{c_2\}$, then also the third case ensures $f(y,x,x) = f(x,x,y) = y$.\\
	The relation $R^{\StructA}$ is preserved by $f$: If $f$ is applied to three tuples in $R^{\StructA_i}$, then it behaves like $f_i$, which is a polymorphism of $\StructA_i$, so it will produce a tuple in $R^{\StructA_i} \subseteq R^{\StructA}$. If $f$ is applied to two tuples in $R^{\StructA_i}$ and the third tuple $(c_i,c_i,c_i)$, then the output will be $(c_i,c_i,c_i) \in R^{\StructA}$. The same applies if all three input tuples are $(c_i,c_i,c_i)$. If two of the input tuples are $(c_i,c_i,c_i)$ and the third one is $\bar{a} \in R^{\StructA_i}$, then $f$ will output $\bar{a}$, which is correct. 
	Also $S^{\StructA}$ is preserved by $f$: Let $\bar{a}_1, \bar{a}_2, \bar{a}_3 \in S^{\StructA}$. Each of these pairs is either of the type $(a,c_2)$ or $(c_1,a)$, with $a \in A_1$ or $a \in A_2$, respectively. If all three pairs are of the same type, then $f$ will also produce a pair of that type, which is in $S^{\StructA}$. Otherwise, two of the three pairs are of type, say, $(a,c_2)$, and the other one is of type $(c_1,a)$. Then $f$ produces a pair of the type $(c_1, a)$, with $a \in A_2$, so this is in $S^{\StructA}$. 
\end{proof}	

\noindent In particular, the above lemma shows tractability of the or-construction in the case that $\CSP{\StructA_1}$ and $\CSP{\StructA_2}$ have a Maltsev polymorphism.
The next lemma shows this for the general case that $\CSP{\StructA_1}$ and $\CSP{\StructA_2}$ are tractable by providing a polynomial-time algorithm
for $\CSP{\ORT{\StructA_1,\StructA_2}}$.
More strongly, we will show, using that algorithm, that if both $\CSP{\StructA_i}$
are definable in a logic subsuming inflationary fixed-point logic,
then so is $\CSP{\ORT{\StructA_1,\StructA_2}}$.

\begin{lemma}
	\label{lem:hom-or-tractable}
	If $\CSP{\StructA_1}$ and $\CSP{\StructA_2}$ are tractable, then $\CSP{\ORT{\StructA_1,\StructA_2}}$ is tractable.
\end{lemma}
\begin{proof}
	Let $\StructA = \ORT{\StructA_1,\StructA_2}$.
	Assume $\StructC$ is a $\CSP{\StructA}$-instance.
	A vertex $\vertA$ of $\StructC$ is a $\sig_i$-vertex if 
	\begin{itemize}
		\item $\vertA$ is in a tuple in a  $\sig_i$-relation, or
		\item there is a pair $(\vertB_1,\vertB_1) \in S^\StructC$
		with $\vertB_i = \vertA$.
	\end{itemize}
	If there is a vertex that is both a $\sig_1$-vertex and a $\sig_2$-vertex,
	then $\StructC \notin \CSP{\StructA}$.
	So assume this is not the case.
	If a vertex is neither a $\sig_1$-vertex nor a $\sig_2$-vertex,
	then this vertex is an isolated vertex.
	Hence, we can assume that every vertex is either a $\sig_1$-vertex or a $\sig_2$-vertex.
	A $\sig_i$-component of $\StructC$ is a connected component of $\restrict{\StructC}{\sig_i}$ (the reduct of $\StructC$ to the vocabulary $\sig_i$)
	consisting only of $\sig_i$-vertices.
	Note that every $\sig_i$-vertex is in a $\sig_i$-component of $\StructC$. 
	
	Consider the following graph $G_S$.
	The vertices of $G_S$ are the $\sig_1$\nobreakdash-components and the $\sig_2$\nobreakdash-components of $\StructC$.
	There is an edge between a $\sig_1$-component $D$ and a $\sig_2$-component $D'$,
	if $(D\times D') \cap S^\StructC \neq \emptyset$.
	Hence, we can also obtain~$G_S$ by contracting each $\sig_i$-component into a single vertex and only keeping the relation~$S$.
	An $S$-component is a connected component of~$G_S$ when viewing~$S$ as an undirected edge relation.
	For an $S$-component~$D$, denote by $\restrict{D}{i}$ the set of $\sig_i$-vertices contained in~$D$.
	
	We claim that $\StructC \in \CSP{\StructA}$ if and only if
	for every $S$-component $D$ there is an $i\in[2]$, such that $\StructC[\restrict{D}{i}] \in \CSP{\StructA_i}$.
	First assume that $\StructC \in \CSP{\StructA}$,
	witnessed by a homomorphism $f \colon \StructC \to \StructA$.
	Let~$D$ be an $S$-component.
	Then $\StructC[D] \in \CSP{\StructA}$ because CSPs are closed under induced substructures.
	If $D$ contains no $\sig_i$-component,
	then $\StructC[\restrict{D}{i}]$ is trivially in $\CSP{\StructA_i}$.
	So assume $D$ contains both $\sig_1$-components and $\sig_2$-components.
	Let $i\in[2]$ and $D^i \in D$ be some $\sig_i$-component
	and assume there is a vertex $\vertA \in D^i$ such that $f(\vertA) \in A_i$.
	Because $D^i$ is a connected component in $\restrict{\StructC}{\sig_i}$, all vertices in $D^i$ have to be mapped onto $A_i$ by $f$
	(the vertex $c_i$ in $A$ is not connected to the $A_i$ vertices).
	All neighbors of $D^i$ in $G_S$ are $\sig_{3-i}$-components
	(otherwise, a vertex was both, a $\sig_1$-vertex and a $\sig_2$-vertex).
	Let $D^{3-i}$ be such a neighbor of~$D^i$.
	Then there are vertices $\vertA \in D^i$ and $\vertB \in D^{3-i}$
	such that $(\vertA,\vertB) \in S^\StructC$ or $(\vertB,\vertA) \in S^\StructC$ (depending on whether $i=1$ or $i=2$).
	By symmetry, assume $(\vertA,\vertB) \in S^\StructC$.
	Then $f(\vertB) = c_{3-i}$ because vertices of $\StructA_i$ are only connected to $c_{3-i}$ in $S^\StructA$.
	Now because $D^{3-i}$ is a connected component of $\StructC[\restrict{D}{3-i}]$,
	all vertices of $D^{3-i}$ are mapped to $c_{3-i}$ by $f$.
	One similarly shows, that if $D^i$ contains a vertex that is mapped to $c_i$
	by $f$, then all vertices in $D^i$ are mapped to $c_i$ and all
	neighbors of $D^i$ in $G_S$ are mapped onto $A_{3-i}$.
	
	Hence, we can inductively show,
	that if for some $i\in[2]$ one $\sig_i$-component in $D$ is mapped onto $A_i$,
	then all $\sig_i$-components in $D$ are mapped onto $A_i$. 
	This implies that  $\StructC[\restrict{D}{i}] \in \CSP{\StructA_i}$.
	There has to be one $i\in[2]$ such that  a $\sig_i$-component
	is mapped onto $A_i$ because
	we have seen that such a component is either mapped onto $A_i$ or onto $c_i$,
	that neighbored components are mapped on exactly the other option,
	and there are both $\sig_1$-components and $\sig_2$-components in~$C$.
	
	Second, assume that for every $S$-component $D$  there is an $i\in[2]$
	such that $\StructC[\restrict{D}{i}] \in \CSP{\StructA_i}$.
	We construct a homomorphism $f \colon \StructC \to \StructA$.
	For an  $S$-component $D$ of $\StructC$,
	note that $\StructC[D]$ is a connected component of $\StructC$
	because all vertices related by a $\sig_i$-relation are in the same $\sig_i$-component and two $\sig_1$- and $\sig_2$\nobreakdash-components, which are related by $S$, are contained in the same $S$-component.
	Hence, it suffices to show that $\StructC[D] \in \CSP{\StructA}$
	for every $S$-component $D$.
	Let $D$ be an $S$-component of $\StructC$
	and $i\in[2]$ such that $\StructC[\restrict{D}{i}] \in \CSP{\StructA_i}$.
	Essentially by the arguments in the proof of Lemma~\ref{lem:hom-or-construction-correct},
	one shows that the extension of a homomorphism $\StructC[\restrict{D}{i}] \to \StructA_i$ to all $\sig_{3-i}$ vertices by mapping them to $c_{3-i}$
	is a homomorphism $\StructC[D] \to \StructA$.
	
	So the algorithm deciding $\CSP{\StructA}$ works as follows.
	Because both $\CSP{\StructA_i}$ are tractable,
	we pick polynomial-time algorithms for these CSPs.
	First compute the $\sig_i$-components of $\StructC$,
	then the graph $G_S$,
	and then check for every $S$-component $D$ of $G_S$
	wether for some $i\in[2]$
	the algorithm for $\CSP{\StructA_i}$ accepts $\StructC[\restrict{D}{i}]$.
	If this is the case, then accept $\StructC$,
	and otherwise reject it.
	Clearly this algorithm runs in polynomial time
	and by the reasoning before it correctly decides $\CSP{\StructA}$.
\end{proof}

\begin{corollary}
	\label{cor:hom-or-definable}
	Let $L$ be a logic that is at least expressive as inflationary fixed-point logic.
	If $\CSP{\StructA_1}$ and $\CSP{\StructA_2}$ are $L$-definable,
	then $\CSP{\ORT{\StructA_1,\StructA_2}}$ is $L$-definable.
\end{corollary}
\begin{proof}
	It can be easily seen that the algorithm in the proof of Lemma~\ref{lem:hom-or-tractable} is $L$-definable 
	given $L$-formulas defining $\CSP{\StructA_1}$ and $\CSP{\StructA_2}$.
	The algorithm essentially computes different connected components,
	which is definable in inflationary fixed-point logic.
\end{proof}

\noindent We will use the tractable homomorphism or-construction to build
tractable CSPs that are not solved correctly by the algorithms in Theorem \ref{thm:mainResultInformal}.
We consider templates $\StructA_1$ and $\StructA_2$ for which the tractable or-construction $\StructA = \ORT{\StructA_1,\StructA_2}$ will guarantee the existence of integral solutions to $\cspiso{\StructA}{k}{\StructB}$ for certain instances $\StructB = \OR{\StructB_1,\StructB_2} \notin \CSP{\StructA}$. This will in particular be the case even though no such integral solution exists for $\cspiso{\StructA_1}{k}{\StructB_1}$ and $\cspiso{\StructA_2}{k}{\StructB_2}$.
However, the \emph{cohomological $k$-consistency} algorithm will be able to tell that $\cspiso{\StructA_1}{k}{\StructB_1}$ and $\cspiso{\StructA_2}{k}{\StructB_2}$ do not have an integral solution, and this will be sufficient for it to correctly output that $\StructB \notin \CSP{\StructA}$. The next two lemmas are the technical foundation for this and will be used in the proof of the first part of Theorem \ref{thm:mainPowerOfCohomology}.
The crucial point is that the cohomological algorithm considers solutions to $\cspiso{\StructA}{k}{\StructB}$ in which for certain sets $X$, every $f: X \to A$ that has $c_i$ in its image receives value $0$.

\begin{lemma}
	\label{lem:fixingIntegerSolutionInOr}
	Let $k \geq 2$, $i \in [2]$, 
	and $\Phi$ be a solution to $\cspiso{\StructA}{k+1}{\StructB}$.
	If there is a set $Z \in \tbinom{B_i}{\leq k}$ such that
	for every $f\in \Hom{\StructB[Z]}{\StructA}$ with $c_i \in f(Z)$
	it holds that
	$\Phi(x_{Z,f}) = 0$, then $\Phi(x_{X,g}) = 0$
	for every $X \in \tbinom{B_i}{\leq k}$ and every
	$g\in \Hom{\StructB[X]}{\StructA}$ with $c_i \in g(X)$.
\end{lemma}
\begin{proof}
	We first show that every $\Phi$ 
	is only non-zero for $x_{X,f}$, for non-empty $X \subseteq B_i$, if either $f(X) \subseteq A_i$ or $f(X) = \set{c_i}$.
	\begin{claim}
		\label{clm:fixingIntegerSolutionInOr-A}
		Let $X \in \binom{B_i}{\leq k}$, $g \in \Hom{\StructB[X]}{\StructA}$ such that there is a $b \in X$ with $g(b) \in A_i$, and assume there is a $b' \in X$ with $g(b') \notin A_i$. Then $\Phi(x_{X,g}) = 0$.
	\end{claim}
	\begin{claimproof}
		We extend the domain and let $Y = X \cup \{y\}$ for an arbitrary $y \in B_{3-i}$. Since $\Phi$ is a solution, it satisfies Equation \ref{eqn:csp-iso-agree}:
		\begin{align*}
			\Phi(x_{X,g}) = \sum_{f \in \Hom{\StructB[Y]}{\StructA}, f|_X = g} \Phi(x_{X,f}) = 0.  
		\end{align*}
		The last equality holds because the sum is over the empty set.
		Indeed, every partial homomorphism that maps $b \in B_i$ to an element of $A_i$ has to map $y \in B_{3-i}$ to $c_{3-i}$ (because of the relation $S$). But then it also has to map $b' \in B_i$ to an element of $A_i$, again to preserve the relation $S$ between $b'$ and $y$. So~$g$ is not extendable to a partial homomorphism with domain~$Y$.
	\end{claimproof}	
	\noindent Similarly, it is not possible to have non-zero values for partial homomorphisms which do not map any element in $X \subseteq B_i$ to $A_i$.
	\begin{claim}
		\label{clm:fixingIntegerSolutionInOr-B}
		Let $X \in \binom{B_i}{\leq k}$ be non-empty and $g \in \Hom{\StructB[X]}{\StructA}$ such that $g(X) \cap (A_{3-i} \cup \{c_{3-i}\}) \neq \emptyset$. Then $\Phi(x_{X,g}) = 0$.
	\end{claim}
	\begin{claimproof}
		We argue in the same way as in the proof of the previous claim and extend $X$ by some $y \in B_{3-i}$. Since $g$ maps at least one element $x \in X$ to $A_{3-i} \cup \{c_{3-i}\}$, it is impossible to preserve the directed $S^{\StructB}$-edge between $x$ and $y$, so $g$ cannot be extended to a partial homomorphism with domain $Y$.
	\end{claimproof}
	From these two claims it follows that for $X \in \binom{B_i}{\leq k}$, the number $\Phi(x_{X,f})$ can only be non-zero if $f(X) \subseteq A_i$ or $f(X) = \{c_i\}$.
	It thus remains to show that $\Phi(x_{X,f}) = 0$ if $f(X) = \{c_i\}$. For $X = Z$, we have this by assumption of the lemma. For other sets $X$, we use the next claim to propagate this inductively. In the following, we write $X \mapsto c_i$ for the partial homomorphism with domain $X$ that sends every element to $c_i$.
	\begin{claim}
		\label{clm:fixingIntegerSolutionInOr-C}
		Let $X \in \binom{B_i}{\leq k}$ be non-empty such that $\Phi(x_{X,X \mapsto c_i}) = 0$.
		\begin{enumerate}
			\item Let $b \in B_i \setminus X$ and $Y = X \uplus \{b\}$. Then $\Phi(x_{Y,Y \mapsto c_i}) = 0$.
			\item Let $Y \subseteq X$ be non-empty. Then $\Phi(x_{Y,Y\mapsto c_i}) = 0$.
		\end{enumerate}
	\end{claim}
	\begin{claimproof}
		We start with the first part of the claim. Again we use Equation \ref{eqn:csp-iso-agree}:
		\begin{align*}
			0 = \Phi(x_{X,X \mapsto c_i}) &= \sum_{f \in \Hom{\StructB[Y]}{\StructA}, f|_X = g} \Phi(x_{Y,f}) = \\
			&= \Phi(x_{Y, Y \mapsto c_i}) + \sum_{f \in \Hom{\StructB[Y]}{\StructA}, f|_X = g, f(b) \neq c_i} \Phi(x_{Y,f})\\
			&=  \Phi(x_{Y, Y \mapsto c_i}) + 0.
		\end{align*}
		The last equality uses Claim~\ref{clm:fixingIntegerSolutionInOr-A} and~\ref{clm:fixingIntegerSolutionInOr-B}, which tell us that $\Phi$ is zero whenever $f(Y) \neq \{c_i\}$ or $f(Y) \not\subseteq A_i$.  
		The second part of the claim is proved similarly:
		\begin{align*}
			\Phi(x_{Y,Y \mapsto c_i}) &= \sum_{f \in \Hom{\StructB[X]}{\StructA}, f(Y) = \{c_i\}} \Phi(x_{X,f}) = \\
			&= \Phi(x_{X, X \mapsto c_i}) + 0 = 0.
		\end{align*}
		Again we use that by Claim~\ref{clm:fixingIntegerSolutionInOr-A} and~\ref{clm:fixingIntegerSolutionInOr-B}, all extensions of $Y \mapsto c_i$ which do not map all of $X$ to $c_i$ receive the value $0$ in $\Phi$.
	\end{claimproof}
	The lemma follows by inductively applying Claim~\ref{clm:fixingIntegerSolutionInOr-C} and noting that by Claim~\ref{clm:fixingIntegerSolutionInOr-A} and~\ref{clm:fixingIntegerSolutionInOr-B}, $\Phi(x_{X,f}) = 0$ whenever $f(X) \not\subseteq A_i$ or $f(X) \neq \{c_i\}$.
\end{proof}

	\begin{lemma}
	\label{lem:integerSolutionWithLocalFixingSolvesBi}
	Let $k \geq 2$, $i \in [2]$, 
	and $\Phi$ be a solution to $\cspiso{\StructA}{k+1}{\StructB}$.
	If there is a set $X\in \tbinom{B_i}{\leq k}$ such that
	for $f \colon X \to \set{c_i}$ it holds that
	$\Phi(x_{X,f}) = 0$,
	then $\restrict{\Phi}{B_i}$ is a solution to $\cspiso{\StructA_i}{k}{\StructB_i}$.
\end{lemma}
\begin{proof}
	We have to argue that $\restrict{\Psi}{B_i}$ satisfies all equations of Type~\ref{eqn:csp-iso-agree} in $\cspiso{\StructA_i}{k}{\StructB_i}$. So let $X \in \binom{B_i}{\leq k}, b \in X, g \in \Hom{\StructB_i[X \setminus \{b\}]}{\StructA_i}$. The associated equation is
	\begin{align*}
		\sum_{\substack{f \in \Hom{\StructB_i[X]}{\StructA_i},\\ \restrict{f}{X\setminus\set{b}} = g}} x_{X,f} &=  x_{X\setminus{\set{b}},g}  \tag{$\star$}
	\end{align*}
	We know that $\Phi$ satisfies the corresponding equation in $\cspiso{\StructA}{k+1}{\StructB}$, so
	\begin{align*}
		\sum_{\substack{f \in \Hom{\StructB[X]}{\StructA},\\ \restrict{f}{X\setminus\set{b}} = g}} \Phi(x_{X,f}) &=  \Phi(x_{X\setminus{\set{b}},g})  
	\end{align*}
	Lemma \ref{lem:fixingIntegerSolutionInOr} implies that only those extensions $f$ of $g$ with $f(b) \in A_i$ can have a non-zero value in $\Phi$ (this is also true if $g$ has empty domain). These are precisely the $f$ that appear in the sum in $(\star)$. Thus, $\Phi$ also satisfies $(\star)$. 	
	Finally, $\restrict{\Psi}{B_i}$ trivially satisfies Equation~\ref{eqn:csp-iso-empty}.
\end{proof}

\noindent From this lemma it follows that for any algorithm that solves $\cspiso{k}{\StructA}{\StructB}$ while fixing local solutions, the tractable or-construction of two templates is not harder than these templates individually. In particular, the cohomological $k$-consistency algorithm is able to undo the or-construction, and it is therefore not useful to obtain hard examples for that algorithm.
To deal with this and be able to prove the second part of Theorem \ref{thm:mainPowerOfCohomology}, we now sacrifice tractability of the homomorphism or-construction, which will also make it harder for cohomological $k$-consistency.

\subsection{The Intractable OR-construction}
\label{app:intractable-or}
In this section, we fix the setting $W_1 =A_1$ and $W_2 = A_2$.
In this case, the homomorphism or-construction has the drawback to yield an NP-complete CSP even if $\CSP{\StructA_1}$ and $\CSP{\StructA_2}$ are tractable.
But it has the benefit that more partial homomorphisms can be extended to global ones, in particular if $\StructB_1\in\CSP{\StructA_1}$ and $\StructB_2 \notin \CSP{\StructA_2}$,
we can still extend partial homomorphisms $\StructB_2 \to \StructA_2$ to
global homomorphisms.
We set $\ORNPC{\StructA_1,\StructA_2} := \ORparam{\StructA_1,\StructA_2,A_1,A_2}$.
We refer to this as the \defining{intractable homomorphism or-construction}.
We again start with corollaries from the general or-construction.

\begin{lemma}
	\label{lem:hom-or-intractable-compose}
	Let $\StructA=\ORNPC{\StructA_1,\StructA_2}$,
	$\StructB = \OR{\StructB_1,\StructB_2}$,
	$X_i \subseteq B_i$, and $f_i \in \Hom{\StructB_i}{\StructA_i}$ for both $i\in[2]$.
	The map $f\colon X_1\cup X_2 \to A$ induced by $f_1$ and $f_2$
	satisfies $f \in \Hom{\StructB[X_1\cup X_2]}{\StructA}$.
\end{lemma}
\begin{proof}
	The lemma is a consequence of Lemma~\ref{lem:hom-or-compose}.
\end{proof}

\begin{lemma}
	\label{lem:hom-or-intractable-k-consistency}
	Assume $\StructA=\ORNPC{\StructA_1,\StructA_2}$
	and $\StructB = \OR{\StructB_1,\StructB_2}$.
	Let $k \in \nat$, $X_i \subseteq B_i$  and $f_i \in \Hom{\StructB_i}{\StructA_i}$ for both $i\in[2]$
	such that $|X_1 \cup X_2|\leq k$.
	If for some $i\in [2]$
	we have $f_i \in \kcol{k}{\StructA_i}{\StructB_i}(X_i)$,
	then the map $f\colon X_i\cup X_2 \to A$ induced by $f_1$ and $f_2$
	satisfies $f \in \kcol{k}{\StructA}{\StructB}(X_1 \cup X_2)$.
\end{lemma}
\begin{proof}
	The lemma follows from Lemma~\ref{lem:hom-or-k-consistency}.
\end{proof}

\begin{lemma}
	\label{lem:hom-or-intractable-solution}
	Assume $\StructA=\ORNPC{\StructA_1,\StructA_2}$
	and $\StructB = \OR{\StructB_1,\StructB_2}$.
	Let  $i \in [2]$, let $\Phi$ be a solution to $\cspiso{k}{\StructA_i}{\StructB_i}$,
	and let $Y \subseteq B_{3-i}$ and $h \in \Hom{\StructB_{3-i}[Y]}{\StructA_{3-i}}$.
	Then there is a solution~$\Psi$ to $\cspiso{k}{\StructA}{\StructB}$
	such that for every $X \in \tbinom{\StructB}{\leq k}$ and $f \in \Hom{\StructB[X]}{\StructA}$,
	$\Psi(x_{X,f})\neq 0$ implies
	$\Psi(x_{X,f}) = \Phi(x_{X\cap B_i, \restrict{f}{X\cap B_i}})$
	and $\restrict{f}{X\cap Y} = \restrict{h}{X\cap Y}$.
	In particular, if $\Psi(x_{Y,h}) = 1$
	and $\Phi$ is an integral solution or a $p$-solution,
	then $\Psi$ is an integral solution or $p$-solution, respectively.
\end{lemma}
\begin{proof}
	Let $\hat{h}$ be the extension of $h$
	to a homomorphism $\StructB_{3-i} \to \StructA[A_{3-i} \cup \set{c_{3-i}}]$
	by Lemma~\ref{lem:hom-or-extend}.
	Then the claimed solution exists by Lemma~\ref{lem:hom-or-solution}.
	It satisfies $\Psi(x_{Y,h}) = 1$ because $\Phi(x_{\emptyset,\emptyset}) =1 $
	since $\Phi$ is a solution to $\cspiso{k}{\StructA_i}{\StructB_i}$.
\end{proof}

The following lemmas show that the intractable homomorphism or-construction yields NP-complete CSPs for various template structures.
We start to consider the very simple case that both $\StructA_i$ are of size~$1$
and contain a ternary relation that is empty.
For a ternary relational symbol $R$,
denote by $\onestruc{R}$ such an $\set{R}$-structure.
Clearly, $\CSP{\onestruc{R}}$ is polynomial-time decidable because an instance is a yes-instance if and only if its $R$-relation is empty.
However, even for this very simple template, the intractable homomorphism or-construction yields an NP-complete CSP.

\begin{lemma}
	\label{lem:intractable-or-NPC-basic}
	For each $i \in [2]$, let $R_i$ be a ternary relation symbol and set
	$\StructA_i := \onestruc{R_i}$.
	Then $\CSP{\ORNPC{\StructA_1,\StructA_2}}$ is NP-complete.
\end{lemma}
\begin{proof}
	Monotone $3$-SAT
	is the NP-complete variant of $3$-SAT in which in every clause the variables all have to be either positive or all negated.
	We show that monotone $3$-SAT
	polynomial-time reduces (and in particular also Karp-reduces) to $\CSP{\ORNPC{\StructA_1,\StructA_2}}$.
	Let $F$ be a $3$-CNF formula with variables $V = \set{x_1, \dots, x_n}$
	such that in each clause the variables are either all positive or all negative.
	We define an instance $\StructB$ of $\CSP{\ORNPC{\StructA_1,\StructA_2}}$.
	For every variable $x_i$, we add two vertices $x_i$ and $\bar{x}_i$
	and add the pair $(x_i, \bar{x}_i)$ to the $S^\StructB$-relation.
	For every clause $\set{x_1, x_2,x_3}$
	in which the variables are positive,
	we add the triple $(x_1,x_2,x_3)$ to $R_1^\StructB$.
	For every clause $\set{\bar{x}_1, \bar{x}_2, \bar{x}_3}$,
	in which the variables are negated,
	we add the triple $(\bar{x}_1,\bar{x}_2,\bar{x}_3)$ to $R_2^\StructB$.
	
	We show that $\StructB \in \CSP{\ORNPC{\StructA_1,\StructA_2}}$ if and only if $F$ is satisfiable.
	Assume that the universe of~$\StructA_i$ is $A_i = \set{a_i}$ for both $i\in[2]$.
	We first assume that~$F$ is satisfiable and let $\Phi\colon V \to \set{0,1}$ be a satisfying assignment.
	We show that the following map $g$ is a homomorphism $\StructB \to \CSP{\ORNPC{\StructA_1,\StructA_2}}$: for a variable $x\in V$ define
	\begin{align*}g(x) &:= \begin{cases}
			a_1 & \text {if } \Phi(x) = 0,\\
			c_1 & \text {if } \Phi(x) = 1,
		\end{cases}\\
		g(\bar{x}) &:= \begin{cases}
			a_2 & \text {if } \Phi(x) = 1,\\
			c_2 & \text {if } \Phi(x) = 0.
		\end{cases}
	\end{align*}
	Now consider a triple $(x_1,x_2,x_2) \in R_1^\StructB$.
	Then for at least one $j \in[3]$, we have that $\Phi(x_j) = 1$ because
	$\Phi$ satisfies all clauses.
	Hence, the tuple $g(x_1,x_2,x_3)$ contains $c_1$ at least once.
	And so, $g(x_1,x_2,x_3) \in R^{\ORNPC{\StructA_1,\StructA_2}}$.
	Similarly, consider a triple $(\bar{x}_1,\bar{x}_2,\bar{x}_2) \in R_2^\StructB$.
	Then for at least one $j \in[3]$, we have that $\Phi(x_j) = 0$ because
	$\Phi$ satisfies all clauses.
	Hence, the tuple $g(\bar{x}_1,\bar{x}_2,\bar{x}_3)$ contains $c_2$ at least once. 
	And so, $g(\bar{x}_1,\bar{x}_2,\bar{x}_3) \in R_2^{\ORNPC{\StructA_1,\StructA_2}}$.
	
	Second, assume that $\StructB \in \CSP{\ORNPC{\StructA_1,\StructA_2}}$
	and let $g \in \Hom{\StructB}{\ORNPC{\StructA_1,\StructA_2}}$.
	Consider the assignment $\Phi \colon V \to \set{0,1}$ that is defined as follows:
	\[\Phi(x) := \begin{cases}
		1 & \text{if } g(x) = c_1,\\
		0 & \text{otherwise.}
	\end{cases}\]
	We show that $\Phi$ satisfies $F$.
	We first see, that if $g(x) = c_1$, then $g(\bar{x})=a_2$ and similarly, if
	$g(x) = c_2$, then $g(\bar{x})=a_1$
	because $(x,\bar{x})\in S^\StructB$.
	Let $\set{x_1,x_2,x_3}$ be a clause of $F$
	in which all variables occur positively.
	Because $g$ is a homomorphism, $g$ maps one $x_i$ to $c_1$ because $R_1^{\StructA_1}$ is empty.
	Hence, $\Phi(x_i)=1$ and $\Phi$ satisfies the clause $\set{x_1,x_2,x_3}$.
	Let $\set{\bar{x}_1,\bar{x}_2,\bar{x}_3}$ be a clause of $F$
	in which all variables occur negatively.
	Because $g$ is a homomorphism, $g$ maps one $\bar{x}_i$ to $c_2$.
	Hence, $g$ cannot map $x_i$ to $c_1$ and thus $g(x_i) = 0$.
	Thus, $\Phi$ satisfies the clause $\set{\bar{x}_1,\bar{x}_2,\bar{x}_3}$.
	Because we considered arbitrary clauses, $\Phi$ satisfies $F$.
\end{proof}

The next lemma shows that many CSPs reduce to the one in the former lemma.
We call a no-instance of a CSP \defining{inclusion-wise minimal}
if every proper induced subinstance of it is a yes-instance.
The following lemma requires that each $\StructA_i$
has an inclusion-wise minimal no-instance of size at least $3$.
This covers various natural CSPs, for example many kinds of equation systems or
group coset CSPs.

\begin{lemma}
	\label{lem:intractable-or-NPC-reduce}
	Let $R_i$ be ternary relation symbols
	and $\StructA_i$ be $\sig_i$-structures for $i\in[2]$.
	If for both $i \in [2]$,
	there is an inclusion-wise minimal $\sig_i$-structure $\StructC_i \notin \CSP{\StructA_i}$
	of size at least $3$,
	then $\ORNPC{\onestruc{R_1}, \onestruc{R_2}}$ is Karp-reducible to $\ORNPC{\StructA_1,\StructA_2}$.
\end{lemma}
\begin{proof}
	For $i\in[2]$,
	let $\StructC_i$ be an inclusion-wise minimal $\sig_i$-structure $\StructC_i  \notin \CSP{\StructA_i}$ of size at least~$3$.
	Pick a partition of $C_i$ into $C_i^1$, $C_i^2$, and $C_i^3$.
	Such a partition exists because $\StructC_i$ has size at least~$3$.
	
	Let $\StructB$ be an $\ORNPC{1_{R_1}, 1_{R_2}}$-instance.
	We define an $\ORNPC{\StructA_1,\StructA_2}$-instance $\StructB'$.
	For every tuple $\tup{u} = (u_1, u_2, u_3) \in R_i^{\StructB}$,
	introduce a fresh copy $\StructC_i^{\tup{u}}$ of $\StructC_i$.
	We partition $C_i^{\tup{u}}$ corresponding to $C_i^1$, $C_i^2$, and $C_i^3$
	into $C_i^{\tup{u},1}$, $C_i^{\tup{u},2}$, and $C_i^{\tup{u},3}$.
	For every pair $(u_1, u_2) \in S^\StructB$,
	add the pairs $(v_1,v_2)$ to $S^{\StructB'}$
	such that for each $i\in[2]$ there exists a tuple $\tup{w}=(w_1,w_2,w_3) \in R_i^{\StructB}$
	for which $w_j = u_i$ and $v_i \in C_i^{\tup{w},j}$ for some $j \in [3]$.
	
	We show that $\StructB \in \CSP{\ORNPC{\onestruc{R_1}, \onestruc{R_2}}}$
	if and only if $\StructB' \in \CSP{\ORNPC{\StructA_1,\StructA_2}}$.
	First assume that $\StructB \in \CSP{\ORNPC{\onestruc{R_1}, \onestruc{R_2}}}$.
	Let $g \in \Hom{\StructB}{\CSP{\ORNPC{\onestruc{R_1}, \onestruc{R_2}}}}$.
	For every tuple $\tup{u} = (u_1, u_2, u_3) \in R_i^{\StructB}$,
	define $W_{\tup{u}} := \setcond{u_i}{g(u_i) = c_i}$,
	where $c_i$ is the fresh vertex added to~$\onestruc{R_i}$ in the homomorphism or-construction $\ORNPC{\onestruc{R_1}, \onestruc{R_2}}$.
	Let $i \in [2]$ and $\tup{u} \in R_i^{\StructB}$.
	Because~$g$ is a homomorphism, we have $W_{\tup{u}} \neq \emptyset$.
	Pick a homomorphism $f_{\tup{u}} \in \Hom{\StructC_i^{\tup{u}}[C_i^{\tup{u}} \setminus W_{\tup{u}}]}{\StructA_i}$.
	Such a homomorphism exists because $\StructC_i$ is inclusion-wise minimal.
	Now define a map $h \colon \StructB' \to \ORNPC{\StructA_1,\StructA_2}$ via
	\begin{align*}
		h(x) = \begin{cases}
			f_{\tup{u}}(x) & \text {if } x \in C_i^{\tup{u},j}, \tup{u}=(u_1,u_2,u_3), \text { and } g(u_j) \in A_i, \\
			c_i' & \text{if otherwise }x \in C_i^{\tup{u},j} \text { and } g(u_j) \notin A_i.
		\end{cases}
	\end{align*}
	Here,  $c_i'$ denotes the the fresh vertex added to $\StructA_i$ in the homomorphism or-construction $\ORNPC{\StructA_1, \StructA_2}$.
	We show that $h$ is indeed a homomorphism.
	Let $\tup{u} \in R_i^{\StructB}$.
	Then $\restrict{h}{C_i^{\tup{u}} \setminus W_{\tup{u}}} = f_{\tup{u}}$
	and $h(W_{\tup{u}}) = \set{c_i'}$.
	Hence $\restrict{h}{C_i^{\tup{u}}} \in \Hom{B'[C^{\tup{u}}]}{\ORNPC{\StructA_1,\StructA_2}}$
	because extending a partial homomorphism by mapping to $c_i'$ always yields a partial homomorphism in the intractable construction (Lemma~\ref{lem:hom-or-compose}).
	It remains to show that $h$ also preserves the $S$ relation.
	Let $(v_1,v_2) \in S^{\StructB'}$.
	Then there exists a pair $(u_1, u_2) \in S^\StructB$
	such that for each $i$ there exists a tuple $\tup{w} =(w_1,w_2,w_3)\in R_i^\StructB$ 
	for which $w_j = u_i$ and $v_i \in C_i^{\tup{w},j}$ for some $j \in [3]$.
	Because $g$ is a homomorphism, $g$ cannot map both $u_i$ to $c_i$.
	Hence, for one $i \in [2]$, we have $g(u_i) \in A_i$ and thus
	$h(v_i) \in A_i$.
	We conclude that $h(u_1, u_2) \in S^{\ORNPC{\StructA_1,\StructA_2}}$.

	Second, assume that $\StructB' \in \CSP{\ORNPC{\StructA_1,\StructA_2}}$
	and let $g \in \Hom{\StructB}{\ORNPC{\StructA_1,\StructA_2}}$.
	We define a map $h \colon \StructB \to \ORNPC{\onestruc{R_1}, \onestruc{R_2}}$ via
	\[h(u) := \begin{cases}
		c_i & \text{if } g(v) = c_i'  \text{ for some }  \tup{w}=(w_1,w_2,w_3) \in R_i^{\StructB}, u = w_j, v \in C_{i}^{\tup{w},j},\\
		1_i & \text{otherwise.}
	\end{cases}\]
	Here $1_i$ denotes the unique universe member of $\onestruc{R_i}$.
	We show that $h$ is a homomorphism.
	
	Let $\tup{u}=(u_1,u_2,u_3) \in R_i^{\StructB}$.
	Because $g$ is a homomorphism, there is a $j \in [3]$ and a $v \in C_i^{\tup{u}, j}$ such that $g(v) = c_i'$ because
	$\StructC_i\notin \CSP{\StructA_i}$ but removing any vertex from $\StructC_i$ turns it into a yes-instance.
	Hence, $h(u_j) = c_i$.
	We conclude that $h(\tup{u}) \in R_i^{\ORNPC{1_{R_1}, 1_{R_2}}}$.
	It remains to show that $g$ preserves the $S$ relation.
	Let $(u_1,u_2) \in S^\StructB$.
	We see for all tuples $\tup{w}^i = (w^i_1,w^i_2,w^i_3) \in R_i^{\StructB}$ 
	such that $w^i_{j_i} = u_i$ and all vertices
	$v_i \in C_{i}^{\tup{u},j_i}$ (for both $i \in[2]$) that,
	if $g(v_i) =  c_i'$,
	then $g(v_{3-i}) \neq c_{3-i}$ by the construction of $S^{\StructB'}$.
	This means that there is at most $i \in [2]$, such that $h$ maps $u_i$ to $c_i$. This implies $h(u_1,u_2) \in S^{\ORNPC{1_{R_1}, 1_{R_2}}}$.
\end{proof}

	Lemmas \ref{lem:intractable-or-NPC-basic} and \ref{lem:intractable-or-NPC-reduce} together imply the NP-completeness of the intractable OR-template, as long as $\CSP{\StructA_1}$ and $\CSP{\StructA_2}$ have an inclusion-wise minimal no-instance of size at least $3$.

	\newcommand{\grpCSP}[3]{\mathcal{C}^{#1,#2,#3}}
	\newcommand{\grpCSPf}[3]{\mathcal{C}^{#3}}
	\newcommand{\egrpCSP}[4]{\mathcal{C}^{#1,#2,#3}_{#4}}
	\newcommand{\egrpCSPf}[4]{\mathcal{C}^{#3}_{#4}}

	\section{Tseitin Formulas over Abelian Groups and Expanders}
	\label{sec:tseitin}
	In this section, we define the CSP instances to which we will later apply the or-construction from the previous section in order to fool the affine algorithms.
	Our construction is based on families of graphs with high connectivity, so-called \emph{expanders}. 
	
	\paragraph*{Expander Graphs.} We introduce a well-studied notion of expander graphs and analyze some of their properties.
	
	\begin{definition}[Expander graphs]
		\label{def:expanderGraphs}
		The \defining{expansion ratio} of a graph $G$ is
		\[
		h(G) \coloneqq \min_{\stackrel{W \subseteq V(G)}{0<|W|<|G|/2}} \frac{|\delta(W)|}{|W|}.
		\]
		The expansion ratio of a family $(G_n)_{n \in \bbN}$  of graphs is 
		\[
		h((G_n)_{n \in \bbN}) = \inf_{n \in \bbN} h(G_n).
		\]
		If $h((G_n)_{n \in \bbN}) > 0$, then $(G_n)_{n \in \bbN}$ is a family of \defining{expander graphs}.
	\end{definition}
	Here, $\delta(W) \subseteq E(G)$ denotes the boundary of~$W$, i.e., the set of edges between $W$ and its complement in $G$.
	It is folklore that families of expander graphs exist, and can be obtained for example with a randomized construction.  
	
	\begin{lemma}[{\cite[Fact A.1]{berkholzGrohe2016fullVersion}}]
		\label{lem:existenceOfExpanders}
		There exists a family $(G_n)_{n \in \bbN}$ of $3$-regular $2$-connected expander graphs such that $|G_n| \in \Theta(n)$. 
	\end{lemma}	
	A graph being $3$-regular means that every vertex has degree exactly $3$, and $2$-connectivity means that the graph is connected and cannot be disconnected by removing only one vertex. We are never going to need the properties of expanders as stated in Definition \ref{def:expanderGraphs}, but rather the following consequence of the definition.
	
	\begin{definition}[Closure properties of $2$-connected expanders]
		\label{def:expanderClosureProperty}
		Let $(G_n)_{n \in \bbN}$ be a family of $2$-connected graphs. 
		\begin{itemize}
		\item For a constant $c > 0$, we say that the family is \defining{$c$-closed} if, for every $n \in \bbN$ and every $X \subseteq E(G_n)$, there is an edge set $\hat{X} \supseteq X$ of size $|\hat{X}| \leq c \cdot |X|$ such that $E(G_n) \setminus \hat{X}$ is empty or the edge set of a $2$-connected subgraph of~$G_n$.
		\item For a constant $c > 0$, we say that the family is \defining{weakly $c$-closed} if there is a function $f \in o(n)$, such that, for every $n \in \bbN$ and every $X \subseteq E(G_n)$, there is an edge set $\hat{X} \supseteq X$ of size $|\hat{X}| \leq c \cdot |X| + f(|E(G_n)|)$ such that $E(G_n) \setminus \hat{X}$ is empty or the edge set of a $2$-connected subgraph of~$G_n$.
		\end{itemize}
	\end{definition}	
		We will also sometimes call an individual graph $G$ (weakly) $c$-closed, when it stems from a (weakly) $c$-closed family.
	
	\begin{lemma}[{\cite[Fact A.4]{berkholzGrohe2016fullVersion}}]
		\label{lem:expanderKeyProperty}
		For every family of expander graphs $(G_n)_{n \in \bbN}$, there is a constant $c > 0$ such that the family is $c$-closed.
	\end{lemma}	
	
	We introduce the notion of \emph{weakly} $c$-closed families because we will consider graphs obtained by removing an edge set $\hat{X}$ as in the above definition from a $c$-closed graph. 
	Being $c$-closed itself is not closed under such edge removals; one only obtains \emph{weakly} $c$-closed graph families in this way:
	\begin{lemma}
		\label{lem:expandersRobust}
		Let $k \colon \bbN \to \bbN$ be a function in $o(n)$. Let $(G_n)_{n \in \bbN}$ be a $c$-closed graph family. Fix a set $X_n \in \binom{E(G_n)}{\leq k(|E(G_n)|)}$ for every $n$.
		Let $G'_n$ be the $2$-connected subgraph of $G_n$ that remains after removing the edges $\hat{X}_n \supseteq X_n$ (and potentially isolated vertices) from $G_n$. Then the graph family $(G'_n)_{n \in \bbN}$ is weakly $c$-closed. 
	\end{lemma}	
	\begin{proof}
		Let $Y \subseteq E(G'_n)$. Then $Y \cup \hat{X}_n \subseteq E(G_n)$ and $|Y \cup \hat{X}_n| \leq  |Y| + ck(|E(G_n)|)$. Because $G_n$ is $c$-closed, by Definition \ref{def:expanderClosureProperty}, there exists $\hat{Y} \supseteq Y \cup \hat{X}_n$ such that $E(G) \setminus \hat{Y} = E(G') \setminus (\hat{Y} \setminus \hat{X}_n)$ is the edge set of a $2$-connected subgraph. Moreover, $|\hat{Y}| \leq c \cdot |Y \cup \hat{X}_n| = c|Y| + c^2k(|E(G_n)|)$.
		Note that \[k(|E(G_n)|) = k(|E(G'_n)| + |\hat{X}_n|) \leq k(|E(G'_n)| + c \cdot k(|E(G_n)|)) \leq k((1+c) \cdot |E(G'_n)|) \in \Theta(k(|E(G'_n)|)).\] 
		Hence, $|\hat{Y}| \leq c \cdot |Y| + c^2k(|E(G_n)|) = c \cdot |Y| + \Theta(k(|E(G'_n)|))$, as desired.
	\end{proof}

	In the following, we will frequently apply Lemmas 4.4, 4.5, and 4.6 from \cite{BerkholzGrohe2017}. All these lemmas assume the context of a $c$-closed family of graphs. However, it is not hard to check that they are also true for \emph{weakly} $c$-closed families. To be precise: These lemmas make the assumption that for the respective edge set $X \subseteq E(G)$ in question, $|X|$ (more precisely a certain matroid rank of $X$) is upper-bounded by $\ell \coloneqq \lfloor \frac{|E(G)|-1}{3c}  \rfloor$. For weakly $c$-closed families, this assumption has to be strengthened so that $|X| \leq \ell - f(|E(G)|)$, for the sublinear function $f$ from Definition \ref{def:expanderClosureProperty}. Whenever we are going to need these lemmas from \cite{BerkholzGrohe2017}, our respective edge set $X$ will be of size at most $o(|E(G)|) < \ell - f(|E(G)|)$, so this stronger condition will always hold. Therefore, we can always safely use Lemmas 4.4, 4.5, and 4.6 from \cite{BerkholzGrohe2017} even though we are only assuming \emph{weakly} $c$-closed graphs most of the time.

	\paragraph*{Tseitin Equations.}
	Families of expander graphs will be used as the basis for so-called \emph{Tseitin} systems of equations.
	Let $G$ be a $3$-regular $2$-connected graph from a weakly $c$-closed family $(G_n)_{n \in \bbN}$; by Lemma \ref{lem:expanderKeyProperty}, in particular, families of expander graphs satisfy this. 
	Fix an orientation~$H$ of~$G$, i.e., a directed graph
	with one direction of each edge of $G$. Let $V:= V(G), E := E(G)$ in this section.
	For a set $W \subseteq V$, denote by $\delta_-(W) \subseteq E$ the in-boundary of $W$, that is, 
	the set of all $uv\in E$
	such that $(u,v) \in E(H) \cap (V \setminus W) \times W$.
	Analogously, $\delta_+(W)\subseteq E$ is the out-boundary, the set of all edges leaving $W$, and $\delta(W) := \delta_+(W) \cup \delta_{-}(W)$.
	Fix a finite Abelian group $\Gamma$. Let $\lambda \colon V \to \Gamma$.
	Define the $\Gamma$-coset-CSP $\grpCSP{H}{\Gamma}{\lambda}$, or $\grpCSPf{H}{\Gamma}{\lambda}$ for short, with variable set $\setcond{y_e }{ e \in E}$ and linear equations
	\begin{align*}
		\sum_{e\in \delta_+(v)} y_e - 	\sum_{e\in \delta_-(v)} y_e &= \lambda(v) &\text{for all } v\in V.
	\end{align*}
	In the case $\Gamma = \ZZ_2$, we obtain the classic Tseitin contradictions~\cite{Tseitin1983}.
	The CSP $\grpCSPf{H}{\Gamma}{\lambda}$ is solvable if and only if $\sum_{v \in V} \lambda(v) = 0$ \cite{BerkholzGrohe2017}. When we view $\grpCSPf{H}{\Gamma}{\lambda}$ as a homomorphism problem,
	then for every $X \subseteq E$,
	a partial solution $f \colon X \to \Gamma$ of $\grpCSPf{H}{\Gamma}{\lambda}$ (so in particular a robustly consistent partial assignment)
	is a homomorphism in $\Hom{\grpCSPf{H}{\Gamma}{\lambda}}{\CosetGrpTmplt{\Gamma}{3}}$
	and vice versa
	(recall that $\CosetGrpTmplt{\Gamma}{3}$ is the template structure for ternary $\Gamma$-coset-CSPs).
	
	\paragraph*{Locally Consistent Solutions.}
	The key feature of Tseitin systems over expanders is the fact that they are always locally consistent even if they are globally inconsistent. The framework to formally reason about this is the following. 
	For all sets $W \subseteq V$,
	the CSP $\grpCSPf{H}{\Gamma}{\lambda}$ implies the constraint $C(W)$ defined via
	\[ \sum_{e\in \delta_+(W)} y_e - 	\sum_{e\in \delta_-(W)} y_e = \sum_{v \in W}\lambda(v).\]

	\begin{definition}[Robustly consistent assignments~\cite{BerkholzGrohe2017}]
		For $\lambda \colon V \to \Gamma$, a set $X \subseteq E$, and an $\ell \in \bbN$,
		a partial assignment $f \colon X \to \Gamma$ for $\grpCSPf{H}{\Gamma}{\lambda}$ is \defining{$\ell$\nobreakdash-consistent},
		if for every $W \in \binom{V}{\leq \ell}$ such that $\delta(W) \subseteq X$, the assignment~$f$ satisfies the constraint $C(W)$.
		Note that $f$ is a partial solution if it is $1$-consistent.
		We call $f$ \defining{robustly consistent} if it is $n/3$-consistent.
	\end{definition}
	\noindent We review facts about robustly consistent assignments for~$\grpCSPf{H}{\Gamma}{\lambda}$. The expansion property of~$G$ ensures that~$\grpCSPf{H}{\Gamma}{\lambda}$ is always locally satisfiable, on subinstances of size up to~$k$, where~$k$ may be a constant or a function in $o(|E|)$.
	This is because the inconsistency can be ``shifted around'' the graph to any equation outside of the local scope. 
	Thus, for every set~$X$ of at most~$k$ variables, there is at least one robustly consistent assignment with domain~$X$. 
	
	In the following lemmas, to avoid cluttered notation, we speak only about a single fixed Tseitin system~$\grpCSPf{H}{\Gamma}{\lambda}$ defined over a graph $G = (V,E)$ from a weakly $c$-closed family. Whenever there are functions of the graph size involved, these are not meant to be constants, but their value explicitly depends on the choice of $G$ from the graph family.

	\begin{lemma}
		\label{lem:robustly-consistent-for-all-small-contexts}
		Let $k \colon \bbN \to \bbN$ be a function in $o(n)$.
		For all $\lambda \colon V\to \Gamma$ and $X \in \tbinom{E}{\leq k(|E|)}$,
		there is a robustly consistent partial assignment $f \colon X \to \Gamma$ for $\grpCSPf{H}{\Gamma}{\lambda}$. 
	\end{lemma}
	\begin{proof}
		For $|X| = 1$, it is clear that there are robustly consistent partial solutions. 
		It then follows from Lemmas 4.4 and 4.5 in \cite{BerkholzGrohe2017} that such a robustly consistent partial solution can always be extended while maintaining robust consistency, as long as the domain size is sublinear in $|E|$. 
	\end{proof}	
	
	Robustly consistent assignments are also not discarded by the $k$-consistency procedure. 
	In particular, $k$-consistency always accepts $\grpCSPf{H}{\Gamma}{\lambda}$ even if it has no solution. 
		\begin{lemma}
		\label{lem:robustlyConsistentSurviveKconsistency}
		Let $k \colon \bbN \to \bbN$ be a function in $o(n)$.
		For every $\lambda \colon V \to \Gamma$,
		the $k(|E|)$-consistency algorithm does not rule out any robustly consistent partial assignments.
		This means that, for every $X \in \binom{E}{\leq k(|E|)}$,  every robustly consistent assignment
		contained in $\Hom{\grpCSPf{H}{\Gamma}{\lambda}[X]}{\CosetGrpTmplt{\Gamma}{3}}$
		is contained in $\kcol{k(|E|)}{\CosetGrpTmplt{\Gamma}{3}}{\grpCSPf{H}{\Gamma}{\lambda}} (X)$.
	\end{lemma}
	\begin{proof}
		Write $k \coloneqq k(|E|)$.
		For a set $X \in \binom{E}{\leq k}$, 
		denote by $\Hom{\grpCSPf{H}{\Gamma}{\lambda}[X]}{\CosetGrpTmplt{\Gamma}{3}}_{n/3} \subseteq \Hom{\grpCSPf{H}{\Gamma}{\lambda}[X]}{\CosetGrpTmplt{\Gamma}{3}}$ 
		the set of robustly consistent homomorphisms in $\Hom{\grpCSPf{H}{\Gamma}{\lambda}[X]}{\CosetGrpTmplt{\Gamma}{3}}$. 
		In order to prove the lemma, we show that the collection 
		\[\setcond*{\Hom{\grpCSPf{H}{\Gamma}{\lambda}[X]}{\CosetGrpTmplt{\Gamma}{3}}_{n/3} }{ X \in \tbinom{E}{\leq k} }\]
		of robustly consistent partial homomorphisms satisfies the down-closure and the forth-condition of $k$-consistency.
		For the down-closure, let $f \in \Hom{\grpCSPf{H}{\Gamma}{\lambda}[X]}{\CosetGrpTmplt{\Gamma}{3}}_{n/3}$.
		By robust consistency,
		the partial solution $f$ satisfies the constraint $C(W)$ for every $W \in \tbinom{V}{\leq n/3}$ such that $\delta(W) \subseteq X$.
		Then the restriction of $f$ to any subset of $X$ still satisfies $C(W)$ for every $W \in \tbinom{V}{\leq n/3}$ such that $\delta(W)$ is in its domain. So this restriction is also robustly consistent. For the forth-condition, let $f \in \Hom{\grpCSPf{H}{\Gamma}{\lambda}[X]}{\CosetGrpTmplt{\Gamma}{3}}_{n/3}$, for some $|X| < k$. Let $y \in E \setminus X$. We need to show that there exists an $f' \in  \Hom{\grpCSPf{H}{\Gamma}{\lambda}[X \cup \{y\}]}{\CosetGrpTmplt{\Gamma}{3}}_{n/3}$ that extends $f$.
		This is again proven in Lemmas 4.4 and 4.5 in \cite{BerkholzGrohe2017}, as long as $k$ is at most sublinear in $|E|$.
	\end{proof}

	\begin{corollary}
		\label{cor:kConsistencyFails}
		For all $\lambda \colon V \to \Gamma$, all $k(n) \in o(n)$, and all $X \in \binom{E}{\leq k(|E|)}$,
		we have $\kcol{k(|E|)}{\CosetGrpTmplt{\Gamma}{3}}{\grpCSPf{H}{\Gamma}{\lambda}}(X) \neq \emptyset$.
	\end{corollary}	
	\begin{proof}
		Combine Lemmas \ref{lem:robustly-consistent-for-all-small-contexts} and \ref{lem:robustlyConsistentSurviveKconsistency}.
	\end{proof}

	Similarly, Tseitin systems fool the width-$k$ LP relaxation in a certain sense. For a prime $p$, a \defining{$p$-group} is a group in which the order of every element is a power of $p$. For instance, $\bbZ_{p^\ell}$ is a $p$-group. 
	
	\begin{lemma}[consequence of {\cite[Lemma 4.6]{BerkholzGrohe2017}}]
		\label{lem:group-csp-p-solution}
		Let $k \colon \bbN \to \bbN$ be a function in $o(n)$,
		$\Gamma$ be a $p$-group,
		and ${\lambda \colon V \to \Gamma}$.
		Then there is a $p$-solution of $\cspiso{k(|E|)}{\CosetGrpTmplt{\Gamma}{3}}{\grpCSPf{H}{\Gamma}{\lambda}}$ such that all non-robustly consistent partial assignments are set to $0$,
		and each robustly consistent partial solution is mapped to $1/p^\ell$ for some $\ell \in \nat$.
	\end{lemma}
	\noindent 
	In the following, we assume the graph $G$ to be $c$-closed, rather than just weakly $c$-closed.
	With this assumption, the above lemma can be refined so that the resulting $p$-solution assigns the value $1$ to $x_{Z,f}$ for a single fixed robustly consistent partial homomorphism $f\colon Z \to \Gamma$ of our choice, where $Z \in \binom{E}{\leq k}$. 
	\begin{lemma}
		\label{lem:fHatCanBeExtended}
		Assume that $G$ is $c$-closed.
		Let $k \colon \bbN \to \bbN$ be a function in $o(n)$, and let $Z \in \binom{E}{\leq k(|E|)}$. Let $f\colon Z \to \Gamma$ be a robustly consistent partial homomorphism, and let $\hat{Z} \supseteq Z$ such that $E \setminus \hat{Z}$ is the edge set of a $2$-connected subgraph of $G$.
		There is an assignment $h \colon \hat{Z} \to \Gamma$ such that $h|_{Z} = f$, and $h \in \Hom{\grpCSPf{H}{\Gamma}{\lambda}[\hat{Z}]}{\CosetGrpTmplt{\Gamma}{3}}$ is robustly consistent.
	\end{lemma}	
	\begin{proof}
		By Definition \ref{def:expanderClosureProperty}, we can choose $\hat{Z}$ such that $|\hat{Z}| \leq c \cdot |Z|$.
		Since $|\hat{Z}| \leq c \cdot |Z| \leq c \cdot k(|E|)$, which is sublinear in $|E|$, we can use Lemmas 4.4 and 4.5 in \cite{BerkholzGrohe2017} to extend~$f$ to a robustly consistent~$h$ with domain $\hat{Z}$. This is in particular a partial homomorphism.
	\end{proof}	
	\noindent Fix this partial solution $h \colon \hat{Z} \to \Gamma$ 
	for $\grpCSPf{H}{\Gamma}{\lambda}$ given by Lemma~\ref{lem:fHatCanBeExtended} in the following.
	Let $G'=(V',E')$ be the graph obtained from $G$ by deleting all edges in $\hat{Z}$
	and all vertices that are not in the $2$\nobreakdash-connected subgraph of $G-\hat{Z}$.
	Similarly, obtain the directed graph $H'$ from $H$ by deleting the same (directed) edges and vertices.
	Let $\lambda' : V' \to \Gamma$ be defined as follows. For every $v \in V'$, set
	\[
	\lambda'(v) := \lambda(v) - \sum_{e \in \delta_+(v) \cap \hat{Z}} h(y_e) + \sum_{e \in \delta_-(v) \cap \hat{Z}} h(y_e) .
	\] 
	With this definition, $\grpCSP{H'}{\Gamma}{\lambda'}$ is the CSP that we obtain from $\grpCSPf{H}{\Gamma}{\lambda}$ by fixing values for the variables in~$\hat{Z}$ according to $h$ from Lemma \ref{lem:fHatCanBeExtended}. In what follows, let $\StructC = \grpCSP{H}{\Gamma}{\lambda}$
	and $\StructC' = \grpCSP{H'}{\Gamma}{\lambda'}$. 
	Next, we show that we can lift a $p$-solution to the width-$k$ relaxation of the instance $\StructC'$ to a $p$-solution of the relaxation of the original instance $\StructC$.
	
	\begin{lemma}
		\label{lem:piecingTogetherPhiAndSolution}
		Let $k \in o(n)$ and write $k \coloneqq k(|E|)$.
		If $\cspiso{k}{\CosetGrpTmplt{\Gamma}{3}}{\StructC'}$ has a $p$-solution $\Phi$, then $\cspiso{k}{\CosetGrpTmplt{\Gamma}{3}}{\StructC}$ has a $p$-solution $\Psi$ such that
		\begin{enumerate}
			\item For all $X' \in \binom{E'}{\leq k}$, all $f' \in \Hom{\StructC'[X]}{\CosetGrpTmplt{\Gamma}{3}}$, if $\Phi(x_{X',f'}) = 0$, then $\Psi(x_{X,f}) = 0$, for every $X$ with $X \cap E' = X'$ and $f|_{E'} = f'$. 
			\item  for all sets of variables $X \in \binom{E \setminus E'}{\leq k}$ of the system $\StructC$ and for all partial homomorphisms $f \in \Hom{\StructC[X]}{\CosetGrpTmplt{\Gamma}{3}}$,
			we have $\Psi(x_{X,f}) = 1$ if $\restrict{f}{X\cap \hat{Z}} = \restrict{h}{X\cap \hat{Z}}$, and $\Psi(x_{X,f}) = 0$, otherwise.
		\end{enumerate}
	\end{lemma}	
	\begin{proof}
		Define $\Psi$ as follows. We identify the variable set of $\StructC' = \grpCSP{H'}{\Gamma}{\lambda'}$
		with $E'$. For all $X \in \binom{E}{\leq k}$ and $f \in \Hom{\StructC[X]}{\CosetGrpTmplt{\Gamma}{3}}$, we set
		\[
		\Psi(x_{X,f}) := \begin{cases}
			\Phi(x_{X \cap E',f|_{E'}}) & \text{ if } f|_{X \setminus E'} = h|_{X \setminus E'} \text{ or if } X \subseteq E',\\
			0 & \text{ otherwise.}
		\end{cases}
		\]
		It remains to show that $\Psi$ is a solution for $\cspiso{k}{\CosetGrpTmplt{\Gamma}{3}}{\StructC}$.
		For Equation~\ref{eqn:csp-iso-empty}, this is clear. Now consider an equation of Type~\ref{eqn:csp-iso-agree}:
		Let $X \in \binom{E}{\leq k}$, $b \in X$, and $g \in \Hom{\StructC}{\CosetGrpTmplt{\Gamma}{3}}$.
		We need to show 
		\[\sum_{\substack{f \in \Hom{\StructC[X]}{\CosetGrpTmplt{\Gamma}{3}},\\ \restrict{f}{X\setminus\set{b}} = g}} \Psi(x_{X,f}) =  \Psi(x_{X\setminus \{b\},g}).\]
		If $g|_{(X \setminus b) \setminus C}$ does not agree with $h$, then both sides of the equation are zero. Hence it remains the case that $g|_{(X \setminus b) \setminus E'}$ does agree with $h$.
		For every $f \in \Hom{\StructC[X]}{\CosetGrpTmplt{\Gamma}{3}}$, it holds: If $f|_{X \setminus E'} = h|_{X \setminus E'}$, then $f|_{E'} \in \Hom{\StructC'[X \cap E']}{\CosetGrpTmplt{\Gamma}{3}}$. This is due to the definition of~$\lambda'$. Thus we have
		\begin{align*}
			\sum_{\substack{f \in \Hom{\StructC[X]}{\CosetGrpTmplt{\Gamma}{3}},\\ \restrict{f}{X\setminus\set{b}} = g}} \Psi(x_{X,f}) &= \sum_{\substack{f \in \Hom{\StructC'[X \cap E']}{\CosetGrpTmplt{\Gamma}{3}},\\ \restrict{f}{X \cap E' \setminus\set{b}} = g}} \Phi(x_{X \cap E',f|_{E'}})\\
			&= \Phi(x_{(X \cap E')\setminus \{b\},g|_{E'}}) = \Psi(x_{X\setminus \{b\},g}  ).
		\end{align*}	
		Therefore, $\Psi$ is a solution of $\cspiso{k}{\CosetGrpTmplt{\Gamma}{3}}{\StructC}$. For every $X \in E \setminus E'$, we have $\Psi(x_{X,f}) = 0$ if $f$ disagrees with~$h$, and $\Psi(x_{X,f}) = \Phi(x_{X \cap C,f|_{C}}) = \Phi(x_{\emptyset,\emptyset}) = 1$, otherwise.
	\end{proof}

	\begin{corollary}
		\label{cor:group-csp-p-solution-with-fixed-assignment}
		Let $k \colon \bbN \to \bbN$ be a function in $o(n)$. Let $G = (V,E)$ be a $c$-closed graph.
		Let $Z \in \binom{E(G)}{\leq k(|E|)}$ and assume that $\Gamma$ is a $p$-group.
		If $f \in \Hom{\StructC[Z]}{\Gamma}$ is robustly consistent, then $\cspiso{k(|E|)}{\Gamma}{\StructC}$ has a $p$-solution $\Psi$ such that
		\begin{itemize}
			\item $\Psi$ is $0$ for all partial assignments that are not robustly consistent.
			\item $\Psi(x_{Z,f}) = 1$.
		\end{itemize}
	\end{corollary}
	\begin{proof}
		With Lemma \ref{lem:fHatCanBeExtended}, we extend $f$ to $h \in \Hom{\StructC[\hat{Z}]}{\Gamma}$.
		By Lemma \ref{lem:expandersRobust}, $G'$ as defined above is weakly $c$-closed. Hence
		Lemma \ref{lem:group-csp-p-solution} can be applied and gives us a $p$-solution for $\cspiso{k}{\CosetGrpTmplt{\Gamma}{3}}{\StructC'}$, to which we can apply Lemma \ref{lem:piecingTogetherPhiAndSolution} to get a $p$-solution for $\cspiso{k}{\CosetGrpTmplt{\Gamma}{3}}{\StructC}$.
		This has the property that it is zero for assignments which are not robustly consistent and it is $1$ for $x_{Z,f}$.
	\end{proof}

	\newcommand{\ZtwoOrThreeInst}{\ORT{\CosetGrpTmplt{\ZZ_2}{3}, \CosetGrpTmplt{\ZZ_3}{3}}}
	\section{Limitations of the Affine Algorithms}
	\label{sec:power-of-affine}

	All of the affine algorithms are \emph{sound}: they accept all yes-instances.
	In this section we prove our main result, namely that many of them are not \emph{complete} on tractable CSPs:  they do not reject all no-instances, and thus do not solve the CSP.
	We consider the tractable homomorphism or-construction $\ORT{\CosetGrpTmplt{\ZZ_2}{3}, \CosetGrpTmplt{\ZZ_3}{3}}$
	of the ternary $\ZZ_2$-coset-CSP and  the ternary $\ZZ_3$-coset-CSP.

	\subsection{\texorpdfstring{$\ZZ$}{ℤ}-Affine \texorpdfstring{$k$}{k}-Consistency Relaxation}
	\label{sec:zAffineConsistency}

	The \defining{$\ZZ$-affine $k$-consistency relaxation} \cite{DalmauOprsal2024}
	solves the following system of affine linear equations over the integers.
	Let $\StructA$ be a template, $\StructB$ be an instance,
	and $\kappa$ be a map
	assigning to every set  $X \in \binom{B}{\leq k}$ a set of partial homomorphisms $\StructB[X] \to \StructA$.
	Define the system $\zafkleq{k}{\StructA}{\StructB}{\kappa}$: 
	
	\begin{systembox}{$\zafkleq{k}{\StructA}{\StructB}{\kappa}$: variables $z_{X,f}$
			for all $X \in {\tbinom{B} {\leq k}}$
			and $f \in \kappa(X)$}
		\begin{align*}
			z_{X,f} &\in \ZZ &  \text{for all } X \in \tbinom{B}{\leq k} \text { and } f \in \kappa(X)\\
			\sum_{f \in \kappa(X)}  z_{X,f}&= 1 &  \text{for all } X \in \tbinom{B}{\leq k}\\
			\sum_{f \in \kappa(X), \restrict{f}{Y} = g} z_{X,f} &= z_{Y,g} & \text{for all } Y\subset X \in \tbinom{B}{\leq k} \text { and } g \in \kappa(Y) 
		\end{align*}
	\end{systembox}
	\noindent Recall that $\kcol{k}{\StructA}{\StructB}$ denotes the output of the $k$-consistency algorithm,
	which is a function that  assigns a set of partial homomorphisms to each set $X \in \binom{B}{\leq k}$.
	The $\ZZ$-affine $k$-consistency relaxation runs, for a fixed positive integer $k$ and a template structure $\StructA$, as follows:
	\begin{algobox}{$\ZZ$-affine $k$-consistency relaxation for template $\StructA$: input a $\CSP{\StructA}$-instance $\StructB$}
		\begin{enumerate}
			\item Compute $\kcol{k}{\StructA}{\StructB}$ using the $k$-consistency algorithm.
			\item Accept if the system $\zafkleq{k}{\StructA}{\StructB}{\kcol{k}{\StructA}{\StructB}}$ is solvable and reject otherwise.
		\end{enumerate}
	\end{algobox}

	\noindent In 2024, Dalmau and Opr\v{s}al~\cite{DalmauOprsal2024} put forward the following
	conjecture on the power of the $\ZZ$-affine $k$-consistency relaxation:
	\begin{conjecture}[\cite{DalmauOprsal2024}]
		\label{con:s3-or-Z}
		For every finite structure $\StructA$,
		either
		$\CSP{K_3}$ is Datalog$^\cup$-reducible to $\CSP{\StructA}$
		or
		$\CSP{\StructA}$ is Datalog$^\cup$-reducible to $\CSP{\ZZ}$,
		where $K_3$ denotes the triangle.
	\end{conjecture}
	\noindent
	Under the assumption $\Ptime \neq \NP$, this conjecture is saying that whenever $\CSP{\StructA}$ is not NP-complete, then it is, up to a simple reduction, equivalent to solving linear equations over the integers. 
	In fact, by the results of \cite{DalmauOprsal2024}, Datalog$^\cup$-reducibility to $\CSP{\ZZ}$ already implies that $\CSP{\StructA}$ is solved by the $\ZZ$-affine $k$-consistency algorithm for some constant $k$.
	Our counterexample is, however, not solved by $\ZZ$-affine $k$-consistency (not even for sublinearly growing $k$), and it does not fall into the first case of the conjecture, either. 
	
	\begin{theorem}[restate=zAffineDoesNotSolveBoundedColorClass, name =]
		\label{thm:z-affine-does-not-solve-bounded-color-class}
		For every $k\geq 1$, the $\ZZ$-affine $k$-consistency relaxation does not solve $\ZtwoOrThreeInst$.
		This is even true if $k \in o(n)$ is a sublinear function of the instance size.
	\end{theorem}
	\begin{proof}
		Let $(G_n)_{n \in \bbN}$ be a family of $3$\nobreakdash-regular $2$\nobreakdash-connected expander graphs, which exists by Lemma \ref{lem:existenceOfExpanders}. Fix $G \coloneqq G_n$, for a large enough $n \in \bbN$, such that $|V(G)|$ is sufficiently larger than~$k$. Let~$H$ be an orientation of $G$.
		Let $p_1 := 2$ and $p_2 := 3$. For each $i \in [2]$, let $\Gamma_i := \ZZ_{p_i}$, and ${\lambda_i: V(G) \to \Gamma_i}$ be $0$ everywhere except at one vertex $v^* \in V(G)$, where we set $\lambda_i(v^*) := 1$.
		For each $i \in [2]$, we consider the $3$\nobreakdash-ary $\Gamma_i$-coset-CSP $\StructB_i := \grpCSP{H}{\Gamma_i}{\lambda_i}$. 
		Let $\StructB := \OR{\StructB_1,\StructB_2} $ and $\StructA := \ZtwoOrThreeInst$ be the corresponding
		tractable homomorphism or-template.
		From $\sum_{v\in V(G)} \lambda_i(v) \neq 0$
		it follows $\StructB_i\notin\CSP{\CosetGrpTmplt{\Gamma_i}{3}}$ for both $i \in [2]$. Thus $\StructB\notin \CSP{\StructA}$  by Lemma \ref{lem:hom-or-construction-correct}.
		For a set $X\subseteq B$, a partial homomorphism  $f\colon\StructB[X] \to\StructA$ is called robustly consistent
		if $\restrict{f}{B_i}$
		is a robustly consistent partial homomorphism $\StructB_i[X\cap B_i] \to \StructA_i$ for some $i \in [2]$
		and $f(X\cap B_{3-i}) = \set{c_{3-i}}$, where $c_{3-i}$ is the fresh vertex of the or-construction.
		Robustly consistent partial solutions of the Tseitin systems 
		are not discarded by $k$-consistency, and by Lemma \ref{lem:hom-or-tractable-k-consistency} this is also true for the 
		robustly consistent partial homomorphisms of the or-instance.
		Then $\cspiso{k}{\StructA_i}{\StructB_i}$ has a $p_i$-solution for both $i\in[2]$ by Lemma~\ref{lem:group-csp-p-solution},
		which is only non-zero for variables indexed by robustly consistent partial homomorphisms.
		By Lemma~\ref{lem:hom-or-tractable-solution}, $\cspiso{k}{\StructA}{\StructB}$ has a $p_i$-solution in which only robustly consistent partial homomorphisms are non-zero, too.
		By Lemma~\ref{lem:p-q-solution-implies-integral},
		there is an integral solution to $\cspiso{k}{\StructA}{\StructB}$,
		which is only non-zero for robustly consistent partial solutions of $\StructB$.
		Such solutions to $\cspiso{k}{\StructA}{\StructB}$
		imply solutions to $\zafkleq{k}{\StructA}{\StructB}{\kappa_k^{\StructA}[\StructB]}$.
		Hence, the $\ZZ$-affine $k$-consistency relaxation wrongly accepts $\StructB$.
	\end{proof}

	\begin{lemma}[restate=notDatalogReducible, name =]
		\label{lem:not-datalog-reducible}
		$\CSP{K_3}$ is not Datalog$^\cup$-reducible to $\CSP{\ZtwoOrThreeInst}$.
	\end{lemma}
		\begin{proof}
		Let $p \notin \set{2,3}$ be a prime
		and $r\geq 3$ be an arity.
		Then $\CSP{\CosetGrpTmplt{\ZZ_p}{r}}$ is Datalog$^\cup$-reducible to $\CSP{K_3}$
		because every finite-domain CSP is Datalog$^\cup$-reducible to $\CSP{K_3}$ (see \cite{DalmauOprsal2024}). 
		We claim that $\CSP{\CosetGrpTmplt{\ZZ_p}{r}}$ is not Datalog$^\cup$-reducible to $\CSP{\ZtwoOrThreeInst}$, which by transitivity of Datalog$^\cup$-reducibility~\cite{DalmauOprsal2024} implies the lemma.
		Suppose for the sake of a contradiction that $\CSP{\CosetGrpTmplt{\ZZ_p}{r}}$
		is Datalog$^\cup$-reducible to $\CSP{\ZtwoOrThreeInst}$. 
		\emph{Rank logic}~\cite{DawarGHL09}
		extends inflationary fixed-point logic by operators
		to define the rank of definable matrices
		over finite prime fields.
		For a set of primes $P$, the characteristic-$P$ fragment
		only provides these operators for finite prime fields
		whose characteristic is contained in $P$.
		For all $i\in\set{2,3}$,
		characteristic-$\set{2,3}$ rank logic
		defines $\CSP{\CosetGrpTmplt{\ZZ_i}{3}}$ by Lemma~\ref{lem:ZpcosetsAreEquations}, that is,
		there is a formula of the logic that is satisfied by a $\CSP{\CosetGrpTmplt{\ZZ_i}{3}}$-instance if and only if it has a homomorphism to $\CosetGrpTmplt{\ZZ_i}{3}$.
		By Corollary~\ref{cor:hom-or-definable}, $\CSP{\ZtwoOrThreeInst}$ 
		is definable in characteristic-$\set{2,3}$ rank logic.
		Because Datalog$^\cup$-reductions are definable in inflationary fixed-point logic,
		$\CSP{\CosetGrpTmplt{\ZZ_p}{r}}$
		is then definable in characteristic-$\set{2,3}$ rank logic.
		This contradicts the result by Grädel and Pakusa~\cite[Theorem 3.3]{GradelPakusa19} that non-isomorphic
		Cai-Fürer-Immerman graphs over the group $\bbZ_p$ (see Section 3.3 in \cite{GradelPakusa19} for the construction) are indistinguishable in characteristic-$P$ rank logic whenever $p \notin P$   :
		The problem of distinguishing these graphs is first-order-reducible to solving a system of (ternary) linear equations in $\bbZ_p$~\cite[Lemma 3.18]{GradelPakusa19},
		so to a $\CSP{\CosetGrpTmplt{\ZZ_p}{r}}$-instance.
		Therefore, the latter cannot be definable in characteristic-$\set{2,3}$ rank logic. This contradicts the assumption that $\CSP{\CosetGrpTmplt{\ZZ_p}{r}}$
		is Datalog$^\cup$-reducible to $\CSP{\ZtwoOrThreeInst}$.
	\end{proof}
	
	\noindent
	Theorem~\ref{thm:z-affine-does-not-solve-bounded-color-class}
	and Lemma~\ref{lem:not-datalog-reducible}
	disprove Conjecture~\ref{con:s3-or-Z}.

	\newcommand{\BLPAIP}[1]{\mathsf{BLP{+}AIP}(#1)}
	\newcommand{\BAk}[2]{\mathsf{BA}^{#1}(#2)}
	\newcommand{\VarsIP}[3]{\mathcal{V}^{#1,#2}(#3)}
	\subsection{BLP+AIP and BA\texorpdfstring{$^k$}{k}}
	\label{sec:BLP}
	We introduce another well-studied system of equations for CSPs~\cite{BartoBKO2021,BrakensiekGWZ2020} parameterized by the size of partial solutions~\cite{CiardoZivny2023GraphColoring}.
	Let $k$ be a positive integer, $\StructA$ a template $\sig$-structure
	and~$\StructB$ a $\CSP{\StructA}$-instance.
	We define the system
	$\ipk{k}{\StructA}{\StructB}$ with variable set $\VarsIP{k}{\StructA}{\StructB}$.
	
	\begin{systembox}{$\ipk{k}{\StructA}{\StructB}$: 
			variables $\lambda_{X,f}$ for all $X \in \tbinom{B}{\leq k}$ and  $f\colon X \to A$, and \\\phantom{$\ipk{k}{\StructA}{\StructB}$: }variables $\mu_{R,\tup{b},\tup{a}}$ for all $R \in \sig$, $\tup{b} \in R^\StructB$, and $\tup{a} \in R^\StructA$}
		\setlength{\belowdisplayskip}{2pt}
		\begin{align*}
			\sum_{f \colon X \to A} \lambda_{X,f} &= 1  &\text{for all } X \in \tbinom{B}{\leq k} \label{eqn:ip-only-one-local-solution}\tag{B1},\\
			\sum_{\substack{f \colon X \to A,\\\restrict{f}{Y} = g}} \lambda_{X,f} &= \lambda_{Y,g} & \text{for all } Y\subset X \in \tbinom{B}{\leq k}, g\colon Y \to A \label{eqn:local-solutions-consistent}\tag{B2},\\
			\sum_{\tup{a} \in R^\StructA, a_{\tup{i}} = \tup{a}'} \mu_{R,\tup{b},\tup{a}} &= \lambda_{X(\tup{b}_{\tup{i}}), \tup{b}_{\tup{i}} \mapsto \tup{a}' } &  \text {for all } R \in \sig, \tup{a}' \in A^k, \tup{b} \in R^\StructB, \tup{i} \in [\arity{R}]^k,\label{eqn:ip-hom}\tag{B3}
		\end{align*}
	where $a_{\tup{i}}$ and $b_{\tup{i}}$ are the $k$-tuples $(a_{\tup{i}_1},...,a_{\tup{i}_k})$ and $(b_{\tup{i}_1},...,b_{\tup{i}_k})$, respectively, $X(\tup{b}_{\tup{i}})$ is the set of entries of $\tup{b}_{\tup{i}}$, and $\tup{b}_{\tup{i}} \mapsto \tup{a}'$ is the function $X(\tup{b}_{\tup{i}})\to A$ mapping $\tup{b}_{\tup{i}}$ to $\tup{a}'$.
	\end{systembox}
	\noindent We consider different domains of the variables (see~\cite{BrakensiekGWZ2020}):
	\begin{itemize}
	\item If we restrict the variables to $\set{0,1}$, then
	$\ipk{1}{\StructA}{\StructB}$ is solvable if and only if
	$\StructB \in \CSP{\StructA}$.
	\item The relaxation of $\ipk{k}{\StructA}{\StructB}$ to nonnegative rationals is the \defining{basic linear programming (BLP)} relaxation $\blk{k}{\StructA}{\StructB}$.
	\item The affine relaxation of $\ipk{k}{\StructA}{\StructB}$ to $\bbZ$ is the \defining{affine integer programming (AIP)} relaxation $\aipk{k}{\StructA}{\StructB}$.
	\end{itemize}
	By increasing the parameter~$k$, 
	the BLP and AIP relaxations result in the Sherali-Adams LP hierarchy and
	the affine integer programming hierarchy of the $\{0,1\}$-system, respectively.
	
	Brakensiek, Guruswami, Wrochna, and Živný \cite{BrakensiekGWZ2020} use a certain combination of $\blk{1}{\StructA}{\StructB}$ and $\aipk{1}{\StructA}{\StructB}$
	to formulate the \defining{BLP+AIP algorithm}.
	Similarly to the $\ZZ$-affine $k$-consistency relaxation,
	the BLP+AIP algorithm tries to solve $\CSP{\StructA}$
	in the sense that it is sound.
	However, it may wrongly answer $\StructB \in \CSP{\StructA}$.
	The question is whether the BLP+AIP algorithm is also complete for tractable CSPs.
	In contrast to the $\ZZ$-affine $k$-consistency relaxation,
	the BLP+AIP algorithm is not parameterized by the size of partial solutions $k$.
	This parameterized version was proposed by Ciardo and Živný~\cite{CiardoZivny2023Tensors, CiardoZivny2023BAk}
	and is called BA$^k$, where BA$^1$ is just the BLP+AIP algorithm.
	\begin{algobox}{$\BAk{k}{\StructA}$-algorithm: input a $\CSP{\StructA}$-instance $\StructB$}
	\begin{enumerate}
		\item Compute a relative interior point $\Phi \colon \VarsIP{k}{\StructA}{\StructB} \to \QQ $ in the polytope defined by $\blk{k}{\StructA}{\StructB}$.
		The solution $\Phi$ has in particular the property that for each variable $x \in \VarsIP{k}{\StructA}{\StructB}$ there is a solution $\Psi$ to $\blp{\StructA}{\StructB}$ such that $\Psi(x) \neq 0$
		if and only if $\Phi(x) \neq 0$.
		If such a point does not exist, reject.\label{itm:bak-interior-point}
		\item \label{item:bak-refined-constr} Refine $\aipk{k}{\StructA}{\StructB}$ by adding the constraints
		\[x = 0 \quad \text { whenever } \quad \Phi(x) = 0 \qquad \text{ for all }{x\in \VarsIP{k}{\StructA}{\StructB}}.\] 
		\item If the refined system is feasible (over $\ZZ$), then accept, otherwise reject.
	\end{enumerate}	
	\end{algobox}
	\noindent The original presentation of BA$^k$ in \cite{CiardoZivny2023Tensors} uses a slightly different system of equations but
	one can verify that our presentation is indeed equivalent. The system in \cite{CiardoZivny2023Tensors} does not have variables $\lambda_{X,f}$ but uses variables $\lambda_{R_k,\tup{b},\tup{a}}$ instead, where $R_k$ is the full $k$-ary relation. These have equivalent semantics. Equation~\ref{eqn:ip-only-one-local-solution} corresponds to Equation~$(1)$ in \cite{CiardoZivny2023Tensors}, and Equations~\ref{eqn:local-solutions-consistent} and~\ref{eqn:ip-hom} are expressed by Equation~$(2)$ in \cite{CiardoZivny2023Tensors}. We deviate from the original presentation to keep it consistent with the systems for the other algorithms.

	We show that BA$^k$ fails on the counterexample provided for $\ZZ$-affine $k$-consistency.
	To do so, we relate solutions of $\cspiso{k}{\StructA}{\StructB}$ 
	to solutions of $\blk{k}{\StructA}{\StructB}$ or $\aipk{k}{\StructA}{\StructB}$.
	
	\begin{lemma}\label{lem:cspiso-implies-aip}
	Let $\StructA$ and $\StructB$ be $\sig$-structures and $k \geq \arity{\sig}$.
	If $\cspiso{k}{\StructA}{\StructB}$ has a solution $\Phi$
	over the non-negative rationals or the integers, then the following assignment $\Psi$ is a solution to $\blk{k}{\StructA}{\StructB}$ or $\aipk{k}{\StructA}{\StructB}$, respectively:
	\begin{align*}
		\Psi(\lambda_{X,f}) &:= \begin{cases}
			\Phi(x_{X,f}) &\text{if } f\in \Hom{\StructB[X]}{\StructA},\\
			0 &\text{otherwise}
		\end{cases} & \text{for all } X \in \tbinom{\StructB}{\leq k}, f \colon X \to A,\\
		\Psi(\mu_{R,\tup{b},\tup{a}}) &:= \Phi (x_{X(\tup{b}), \tup{b} \mapsto \tup{a}}) & \text{for all } R \in \sig, \tup{a} \in R^\StructA, \tup{b} \in R^\StructB,
	\end{align*}
	where $X(\tup{b})$ denotes the set of elements appearing in the tuple $\tup{b}$ and $\tup{b} \mapsto \tup{a}$ denotes the partial homomorphism sending $\tup{b}$ to $\tup{a}$. 
	\end{lemma}
	\begin{proof}
	Lemma~\ref{lem:csp-iso-subsets} implies that equations of Type~\ref{eqn:ip-only-one-local-solution} and~\ref{eqn:local-solutions-consistent} are satisfied.
	For $R \in \sig$, $\tup{a}' \in A^k$, $\tup{b} \in R^\StructB$, and $\tup{i} \in [\arity{R}]^k$,
	we let $Y$ be the set of entries of $\tup{b}_{\tup{i}}$, and consider the homomorphism $g \coloneqq \tup{b}_{\tup{i}}\mapsto \tup{a}'$. Then
	Lemma~\ref{lem:csp-iso-subsets} also implies that equations of Type~\ref{eqn:ip-hom} are satisfied. 
	It is clear that non-negativity or integrality, respectively, of the solution is preserved.
	\end{proof}

	\begin{theorem}[restate=BLPDoesNotSolveBoundedColorClass, name=]
		\label{thm:BLP-does-not-solve-bounded-color-class}
		For every integer $k$, the algorithm $\BAk{k}{\StructA}$ does not solve
		$\ZtwoOrThreeInst$. This is even true if $k \in o(n)$ is a sublinear function of the instance size.
	\end{theorem}
	\begin{proof}
	Let $p_1 = 2$, $p_2=3$, and $\StructA := \ZtwoOrThreeInst$.
	As in the proof of Theorem \ref{thm:z-affine-does-not-solve-bounded-color-class},
	we consider the ternary Tseitin CSP instances over $\ZZ_2$ and $\ZZ_3$
	constructed over sufficiently large expander graphs.
	Again, let $\StructB_i$ be the $\ZZ_{p_i}$ instance for each $i\in[2]$.
	We again have that $\StructB_i \not \in \CSP{\CosetGrpTmplt{\ZZ_{p_i}}{3}}$
	for both $i\in[2]$ and hence $\StructB := \OR{\StructB_1,\StructB_2} \notin \CSP{\StructA}$.
	
	By Lemma \ref{lem:group-csp-p-solution}, for each $i \in [2]$, there is a $p_i$-solution $\Phi_i$ for $\cspiso{k}{\CosetGrpTmplt{\ZZ_{p_i}}{3}}{\StructB_i}$ 
	which sets exactly the robustly consistent partial solutions to a non-zero value. With Lemma \ref{lem:hom-or-tractable-solution}, each $\Phi_i$ gives rise to a solution
	$\Psi_i$ of $\cspiso{k}{\StructA}{\StructB}$ that is non-zero exactly for all partial homomorphisms $f\colon X \to A$ that are robustly consistent for $X \cap B_i$ and map $X \cap B_{3-i}$ to $c_{3-i}$.
	We call these partial homomorphisms $\StructB\to\StructA$ also robustly consistent.
	Let $F$ be the set of all robustly consistent partial homomorphisms  $\StructB[X]\to\StructA$ for all $X\in \tbinom{B}{\leq k}$.
	The relative interior point computed in Step~\ref{itm:bak-interior-point} of the BA$^k$-algorithm exists (because the system is solvable)
	and is in particular non-zero for every $f \in F$.
	By Lemma~\ref{lem:p-q-solution-implies-integral},
	the  $p_1$-solution $\Psi_1$ and the $p_2$-solution $\Psi_2$ for $\cspiso{k}{\StructA}{\StructB}$ can be combined to an integral solution that is only non-zero for partial homomorphisms in $F$.
	Therefore by Lemma~\ref{lem:cspiso-implies-aip},
	the system $\aipk{k}{\StructA}{\StructB}$ also has such an integral solution.
	This solution satisfies the refined constraints from
	Step~\ref{item:bak-refined-constr} of the BA$^k$-algorithm.
	Hence, the algorithm incorrectly accepts the unsatisfiable instance $\StructB$. 
\end{proof}

	\subsection{The CLAP Algorithm}
	\label{sec:CLAP}
	The CLAP algorithm~\cite{CiardoZivny2023CLAP} combines the BLP and the AIP relaxationss.
	It first iteratively reduces the solution space with the BLP by fixing partial solutions to $1$ and discarding those for which this refined BLP is not solvable.
	Then BLP+AIP is run on the restricted solution space, where again a partial solution is fixed:
		\begin{algobox}[breakable]{$\CLAP{\StructA}$-algorithm:
			input a $\CSP{\StructA}$-instance $\StructB$}
		\begin{enumerate}
			\item Maintain, for each pair of a relation symbol $R\in \sig$ and a  tuple $\tup{b} \in R^\StructB$,
			a set $S_{\tup{b},R} \subseteq R^\StructA$ of possible images of $\tup{b}$ under a homomorphism.
			Initialize $S_{\tup{b},R} := R^\StructA$ for all $R\in \sig$ and $\tup{b} \in R^\StructB$. \label{itm:clap-init}
			\item Repeat until no set $S_{\tup{b},R}$ changes anymore:
			For each $R\in\sig$, $\tup{b} \in R^\StructB$, and $\tup{a} \in S_{\tup{b},R}$, solve $\blk{1}{\StructA}{\StructB}$ together with the following additional constraints:
			\begin{align*}
				\mu_{R,\tup{b},\tup{a}} &= 1,\\
				\mu_{R,\tup{b}',\tup{a}'} &= 0 &\text{for all } R' \in \sig, \tup{b}'\in R'^\StructB, \tup{a}' \not\in S_{\tup{b'},R'}.
			\end{align*}
			If this system is not feasible, remove $\tup{a}$ from $S_{\tup{b},R}$.\label{itm:clap-refine-images}
			\item If there are $R\in\sig$ and $\tup{b}\in R^\StructB$ such that $S_{\tup{b},R} =\emptyset$, then reject\label{itm:clap-no-possible-image}.
			\item For each $R \in \sig$, $\tup{b} \in R^\StructB$, and $\tup{a} \in S_{\tup{b},R}$, execute $\BAk{1}{\StructA}$ (which is BLP+AIP) on $\StructB$, where we additionally fix
			\begin{align*}
				\mu_{R,\tup{b},\tup{a}} &= 1,\\
				\mu_{R,\tup{b}',\tup{a}'} & = 0 &\text{for all } R' \in \sig, \tup{b}' \in R'^\StructB, \tup{a}' \not\in S_{\tup{b}',R'}
			\end{align*}
			in Step~\ref{itm:bak-interior-point} of $\BAk{1}{\StructA}$
			(and thus also implicitly in $\aipk{1}{\StructA}{\StructB}$ in Step~\ref{item:bak-refined-constr} of $\BAk{1}{\StructA}$).
			If $\BAk{1}{\StructA}$ accepts, then accept.\label{itm:clap-call-bak}
			\item If $\BAk{1}{\StructA}$ rejects all inputs in the step before, then reject.\label{itm:clap-reject}
		\end{enumerate}
	\end{algobox}
	
	\noindent To simplify the analysis,
	we consider a variant of the CLAP algorithm
	which we call CLAP$'$.
	
	\begin{algobox}[top=0.4em]{$\CLAPw{\StructA}$-algorithm: input a $\CSP{\StructA}$-instance $\StructB$}
		Execute Steps~\ref{itm:clap-init} to~\ref{itm:clap-no-possible-image} of $\CLAP{\StructA}$. Then execute
		\begin{enumerate}[label=\arabic*{'}.,ref=\arabic*{'}]
			\setcounter{enumi}{3}
			\item Execute $\BAk{1}{\StructA}$ on $\StructB$\, where we additionally fix
			\begin{align*}
				\mu_{R,\tup{b}',\tup{a}'}& = 0 &\text{for all } R' \in \sig, \tup{b}' \in R'^\StructB, \tup{a}' \not\in S_{\tup{b}',R'}.
			\end{align*}
			Accept if $\BAk{1}{\StructA}$ accepts this input and reject otherwise.\label{itm:clap'-call-bak}
		\end{enumerate}
	\end{algobox}
	\noindent It is immediate that $\CLAPw{\StructA}$ does not solve more CSPs than $\CLAP{\StructA}$.
	We show that it actually solves the same ones:
	\begin{lemma}
		\label{lem:clap-iff-weak}
		For every structure $\StructA$,
		$\CLAP{\StructA}$ solves $\CSP{\StructA}$ if and only if $\CLAPw{\StructA}$ solves $\CSP{\StructA}$.
	\end{lemma}
	\begin{proof}
		Let $\StructA$ be a template $\sig$-structure.
		It is clear that if $\CLAPw{\StructA}$ solves  $\CSP{\StructA}$, then also $\CLAP{\StructA}$ solves  $\CSP{\StructA}$.
		We show that if $\CLAPw{\StructA}$ does not solve $\CSP{\StructA}$, then $\CLAP{\StructA}$ does not solve $\CSP{\StructA}$, either.
		Let $\StructB$ be a $\sig$-structure such that $\StructB \notin \CSP{\StructA}$, but $\CLAPw{\StructA}$ accepts $\CSP{\StructA}$.
		We create a modified variant of $\StructB$ as follows.
		Let $T \in \sig$ be some relation symbol of arity~$r$
		that is non-empty in $\StructA$ (if $\StructA$ contains only empty relations, then CLAP and CLAP' can trivially solve $\CSP{\StructA}$).
		Let $\StructB'$ be the disjoint union of $\StructB$
		and the $r$\nobreakdash-element $\sig$-structure, for which one $r$-tuple of distinct elements $\tup{x}$ is contained in $T$.
		Obviously, we have $\StructB' \notin \CSP{\StructA}$.
		We show that $\CLAP{\StructA}$ accepts $\StructB'$.
		Since $\StructB'$ is a disjoint union,
		after Steps~\ref{itm:clap-init} to~\ref{itm:clap-no-possible-image}, the sets $S_{\tup{b},R}$ on input $\StructB'$ will contain at least the elements as on input $\StructB$.
		The set $S_{\tup{x},T}$ will be equal to $T^\StructA$
		because fixing the assignment of $\tup{x}$ does not restrict any other partial homomorphisms,
		and since CLAP' accepts $\StructB$, the system $\blp{\StructA}{\StructB}$ is solvable
		when an image of~$\tup{x}$ is fixed.
		In particular, no set  $S_{\tup{b},R}$ will be empty after Step~\ref{itm:clap-refine-images}.
		Hence, Step~\ref{itm:clap-no-possible-image} is passed successfully.
		Now for Step~\ref{itm:clap-call-bak}, we consider the relation $T$ and the tuple $\tup{x}$.
		We consider the execution of BA$^1$, where an arbitrary image of $\tup{x}$ contained in $T$ is fixed.
		Because $\StructB'$ is a disjoint union and the mapping of $\tup{x}$ is a valid homomorphism from the attached structure to $\StructA$
		and because $\BAk{1}{\StructA}$ accepts in Step~\ref{itm:clap'-call-bak},
		$\BAk{1}{\StructA}$ will accept in Step~\ref{itm:clap-call-bak} for the tuple $\tup{x}$.
		Hence, $\CLAP{\StructA}$ wrongly accepts $\StructB'$, which means that it does not solve $\CSP{\StructA}$.
	\end{proof}
	
	\begin{theorem}
	\label{thm:clap-does-not-solve-all}
	$\CLAP{\ZtwoOrThreeInst}$ does not solve $\CSP{\ZtwoOrThreeInst}$.
\end{theorem}
	\begin{proof}
		We prove the result for $\CLAPw{\StructA}$, which is sufficient by Lemma \ref{lem:clap-iff-weak}.
		Let $k=3$.
		As in the proofs of Theorems \ref{thm:BLP-does-not-solve-bounded-color-class} and \ref{thm:z-affine-does-not-solve-bounded-color-class}, 
		we consider ternary Tseitin systems over $\ZZ_2$ and $\ZZ_3$
		for a sufficiently large $3$-regular $2$-connected expander graph.
		Let again $\StructB_1$ and $\StructB_2$ be these instances,
		which for $p_i= i+1$ are no-instances for $\CSP{\CosetGrpTmplt{\ZZ_{p_i}}{3}}$
		for both $i \in[2]$.
		Again, let $\StructB := \OR{\StructB_1,\StructB_2}$
		and $\StructA := \CSP{\ZtwoOrThreeInst}$.
		So, $\StructB \notin \CSP{\StructA}$.
		
		By Corollary \ref{cor:group-csp-p-solution-with-fixed-assignment}, for each $i \in [2]$, and every $f \in \Hom{\StructB_i[X]}{\StructA_i}$ that is robustly consistent, 
		there exists a $p_i$-solution $\Phi_{i,f}$ to $\cspiso{k}{\CosetGrpTmplt{\ZZ_{p_i}}{3}}{\StructB_i}$ which sets $x_{X,f}$ to~$1$ and every partial homomorphism that is not robustly consistent to~$0$.
		By Lemma \ref{lem:hom-or-tractable-solution}, each $\Phi_{i,f}$ translates into a $p_i$-solution $\Psi_{i,f}$ for $\cspiso{k}{\StructA}{\StructB}$ that sets every partial homomorphism to $1$ which agrees with $f$ on $B_i$ and sends the $B_{3-i}$-part of its domain to $c_{3-i}$.
		We call such partial homomorphisms (for both $i\in[2]$)
		again also robustly consistent.
		Let $F$ denote the set of al robustly consistent  partial solutions 
		$\StructB[X] \to \StructA$ for every $X \in \tbinom{B}{\leq k}$.
		Now consider Step \ref{itm:clap-refine-images} of CLAP. The algorithm adds in particular the equation $\mu_{R,\tup{b},\tup{a}} = 1$ to the systems considered in $\BAk{1}{\StructA}$.
		If there is some $f \in F$ that contains the assignment $\tup{b} \mapsto \tup{a}$, then $\Psi_{i,f}$ gives us a solution for $\BAk{1}{\StructA}$ (via Lemma \ref{lem:cspiso-implies-aip}) that also satisfies $\mu_{R,\tup{b},\tup{a}} = 1$.
		By Lemmas \ref{lem:robustly-consistent-for-all-small-contexts} and \ref{lem:hom-or-tractable-solution}, for every $R \in \tau(\StructA)$ and $\bar{b} \in R^{\StructA}$, there is at least one $\bar{a} \in R^{\StructA}$ such that we can find an $f \in F$ containing $\bar{b} \mapsto \bar{a}$. All tuples $\bar{a}$ that are removed from $S_{\bar{b},R}$ in Step \ref{itm:clap-refine-images} do not satisfy that $\bar{b} \mapsto \bar{a}$ is contained in an $f \in F$. This means that $\bar{b} \mapsto \bar{a}$ is not part of a robustly consistent partial solution of $\StructB_1$ or $\StructB_2$. Thus, it will be set to zero by all the $\Phi_{i,f}$ and $\Psi_{i,f}$, and hence, these solutions also satisfy the extra equations $\mu_{R,\tup{b},\tup{a}} = 0$ in Step \ref{itm:clap-refine-images}.
		In total, Step \ref{itm:clap-no-possible-image} of CLAP is passed successfully, and the only tuples $\bar{a}$ that are removed from $S_{\bar{b},R}$ are such that $\bar{b} \mapsto \bar{a}$ is not part of a robustly consistent partial solution.
		In Step~\ref{itm:clap'-call-bak}, $\CLAPw{\StructA}$ will then accept:
		The proof of Theorem~\ref{thm:BLP-does-not-solve-bounded-color-class} shows that $\BAk{1}{\StructA}$ accepts $\StructB$, and it can be seen that this proof also goes through if we set $\mu_{R,\tup{b}',\tup{a}'} = 0$ for all partial solutions $\bar{b}' \mapsto \bar{a}'$ that are not robustly consistent.
	\end{proof}	
	
	\noindent In contrast to the $\ZZ$-affine $k$-consistency relaxation and the BA$^k$ algorithms,
	CLAP is not parameterized by a width $k$.
	However, we did not exploit this fact, and our techniques could also be applied to a version of CLAP parameterized by a width.
	
	The reason why our simplified algorithm CLAP' is equivalent to CLAP is because CLAP immediately accepts	if Step~\ref{itm:clap-call-bak} is passed successfully for at least one tuple.
	One could modify CLAP so that Step~\ref{itm:clap-call-bak} has to find one possible image for all $R\in\sig$ and all $\tup{b}\in R^\StructB$. This would still be a sound algorithm.
	Ciardo and Živný~\cite{CiardoZivny2023CLAP} already noted this possibility
	when introducing CLAP, and moreover suggested a possibly even stronger version:
	replace BLP with BLP+AIP in Step~\ref{itm:clap-refine-images},
	which in turn would make Steps~\ref{itm:clap-call-bak} and~\ref{itm:clap-reject} unnecessary.
	The authors refer to this algorithm as C(BLP+AIP)
	but considered CLAP because it allows  to characterize
	the CSPs solved by CLAP in terms
	of the polymorphisms of the template structure~$\StructA$.
	A similar characterization for C(BLP+AIP) has recently been found by Zhuk \cite{zhuk2025singletonalgorithms} (where the algorithm is called CSingl(BLP + AIP)).
	We do not study C(BLP+AIP) in this article but suspect that it has similar properties as the cohomological algorithm, which we turn to next.
	In particular, we believe that Theorem \ref{thm:clap-does-not-solve-all} is not true for C(BLP+AIP).
	
	\subsection{The Cohomological \texorpdfstring{$k$}{k}-Consistency Algorithm}
	\label{sec:cohomology}
		We review the cohomological $k$-consistency algorithm due to Ó Conghaile \cite{OConghaile22}.
	It combines techniques of the algorithms we have seen so far --
	the iterative approach of $k$-consistency with solving the AIP with a fixed local solution (called Singleton-AIP in \cite{zhuk2025singletonalgorithms})
	in \emph{every} iteration.
	The name references \emph{cohomology} because solving the AIP 
	can be interpreted as checking for the existence of a cohomological obstruction in the presheaf of partial homomorphisms. The algorithm itself is straightforward and can be stated without the categorical terminology:
	\begin{algobox}{Cohomological $k$-consistency algorithm:
			input a $\CSP{\StructA}$-instance $\StructB$}
		\begin{enumerate}
			\item Maintain, for each $X \in \binom{B}{ \leq k}$, a set $\Hh(X) \subseteq \Hom{\StructB[X]}{\StructA}$. Initialize $\Hh(X) := \Hom{\StructB[X]}{\StructA}$.
			\item Repeat until none of the sets $\Hh(X)$ changes anymore: 
			\begin{enumerate}[label=(\alph*)]
				\item Run the $k$-consistency algorithm on $\Hh$ to remove from each $\Hh(X)$ the partial homomorphisms that fail the forth-condition or down-closure property.
				\item For each $X \in \binom{B}{ \leq k}$ and $f \in \Hh(X)$, check whether $\zafkleq{k}{\StructA}{\StructB}{\Hh}$ has a solution that satisfies $x_{X,f} = 1$ and $x_{X,f'} = 0$ for every $f' \in \Hh(X) \setminus \{f\}$. If it does not, then remove $f$ from $\Hh(X)$ for the next iteration of the loop. 
			\end{enumerate}	
			\item If $\Hh(X) = \emptyset$ for some $X \in \binom{B}{\leq k}$, then reject; otherwise accept. 
		\end{enumerate}
	\end{algobox}
	Step 2(b) of the algorithm tries to check whether
	there is a global homomorphism whose restriction to $X$ is equal to $f$
	-- and this check is approximated by solving the AIP
	in which we set $x_{X,f} = 1$ and $x_{X,f'} = 0$ for all other $f'$ with domain $X$. It is specifically this fixing of a local solution $f$ in the AIP that makes the cohomological algorithm more powerful than the previous ones: 
	Indeed, as shown in Theorem \ref{thm:cohomologySolvesCounterexample} below, it correctly solves the template $\CSP{\ZtwoOrThreeInst}$ that we have used as a counterexample for the other algorithms. 
	
	Nonetheless, we can also show a limitation of this algorithm: It fails to solve the \emph{intractable} homomorphism or-construction on $\bbZ_2$ and $\bbZ_3$. This proves unconditionally, that is, without any complexity-theoretic assumptions like P $\neq$ NP, that this polynomial-time algorithm does not solve all finite-domain CSPs.

	\begin{theorem}
	\label{thm:cohomologySolvesCounterexample}
	If $\StructA_1, \StructA_2$ are templates of Abelian coset-CSPs
	and $k\geq \arity{A_i}+1$ for both $i\in[2]$,
	then cohomological $k$-consistency solves $\CSP{\ORT{\StructA_1,\StructA_2}}$.
	\end{theorem}
	\begin{proof}
		We first argue that the cohomological $k$-consistency algorithm
		correctly rejects or-instances $\StructB=\OR{\StructB_1,\StructB_2} \notin \CSP{\ORT{\StructA_1, \StructA_2}}$. Let $\StructA = \ORT{\StructA_1, \StructA_2}$.
		
		We argue that when the algorithm terminates, $\Hh(X) = \emptyset$ for at least one $X \in \binom{B}{\leq k}$.
		Consider some $X \subseteq B_1$ that is exactly the set of entries of some tuple $\bar{b} \in R^{\StructB_1}$, for some $R \in \tau(\StructB_1)$.
		Let $f : X \to A_1$ be an arbitrary partial homomorphism. Then $\cspiso{k}{\StructA}{\StructB}$ does not have an integral solution $\Phi$ that satisfies $\Phi(x_{X,f}) = 1$ and $\Phi(x_{X,f'}) = 0$ for every $f' \in \Hh(X) \setminus \{f\}$: By Lemma~\ref{lem:integerSolutionWithLocalFixingSolvesBi}, such a $\Phi$ would in particular be a solution to $\cspiso{\StructA_1}{k-1}{\StructB_1}$. 
		But it is known that for Abelian coset-CSPs, the existence of an integral solution to $\cspiso{\StructA_1}{k-1}{\StructB_1}$, where $k-1 \geq \arity{\StructA_1}$, is equivalent to $\StructB_1 \in \CSP{\StructA_1}$ (see also Theorem \ref{thm:abelianSolvedByAIP}). Thus, since $\StructB_1 \notin \CSP{\StructA_1}$, such a $\Phi$ cannot exist.
		Then in particular, $\zafkleq{k}{\StructA}{\StructB}{\Hh}$ does not have such an integral solution, even if $\Hh(X) = \Hom{\StructB[X]}{\StructA}$, as initially. 
		Hence, all $f$ with $f(X) \subseteq A_1$ are removed from $\Hh(X)$ in this iteration. Since $X$ is the entry set of $\bar{b} \in R^{\StructB_1}$, and $R^{\StructA} \subseteq A_1^3 \cup \{(c_1,c_1,c_1)\}$, the only other partial homomorphism in $\Hh(X)$ is the one with $f(X) = \{c_1\}$. This is the only homomorphism that may still be in $\Hh(X)$ after the first iteration.
		We can also consider another $X' \subseteq B_{2}$ that is the set of entries of some $\bar{b}' \in R'^{\StructB_2}$, and the same argument shows that after the first iteration, there is at most the partial homomorphism $f$ with $f(X') = \{c_2\}$ in $\Hh(X')$.
		Then consider the set $\{x,x'\}$ for some $x \in X, x' \in X'$. 
		After $k$-consistency is run in the second iteration, $\Hh(\{x,x'\})$ will be empty. This is because $\Hh(X)$ and $\Hh(X')$ enforce that $x$ is mapped to $c_1$ and $x'$ is mapped to $c_2$, but every partial homomorphism in $\Hom{\StructB[x,x']}{\StructA}$ maps either $x$ or $x'$ to an element of $A_1$ or $A_2$, respectively.

		We now not only want to show that the cohomological $k$-consistency algorithm correctly rejects or-instances but every $\CSP{\ORT{\StructA_1,\StructA_2}}$-instance that does not have a solution.
		Assume that $\StructB$ is an unsatisfiable $\CSP{\ORT{\StructA_1,\StructA_2}}$-instance.
		To show that the algorithm rejects $\StructB$, we follow the algorithm to solve this CSP provided in the proof of Lemma~\ref{lem:hom-or-tractable}.
		For $i\in[2]$, let $B_i\subseteq B$ be the set of all $\sig_i$-vertices.
		We can assume that~$B_1$ and~$B_2$ form a partition of~$B$:
		a vertex which is neither in~$B_1$ nor in~$B_2$ is isolated,
		and a vertex in $B_1\cap B_2$ appears in relations of both~$\sig_1$ and~$\sig_2$,
		which is not the case for any element of $\StructA$, and thus $k$-consistency would immediately reject $\StructB$.
		The proof of Lemma~\ref{lem:hom-or-tractable} shows that if $\StructB$ is unsatisfiable, then there is some $S$-component $D$ in the graph $G_S$ that contains both an unsatisfiable $\tau_1$- and unsatisfiable $\tau_2$-component (we refer to that proof for the terminology). Therefore we know that such an $S$-component $D$ exists and we will now argue that the cohomological algorithm detects it.

		We first note that $k$-consistency detects $\sig_i$-components in the following sense: for every $i\in[2]$ and  $X\subseteq B_i$ that is also a subset of a single $\sig_i$-component,
		$k$-consistency discards all partial homomorphisms $f\colon \StructB[X] \to \StructA$ such that there are~$b$ and~$b'$
		with $f(b) \in A_i$ and $f(b') = c_i$.
		This is the case because~$b$ and~$b'$ are connected via relations in $\sig_i$,
		but $f(b)$ and $f(b')$ are not connected via $\sig_i$-relations in the tractable homomorphism-or construction,
		which can easily be detected by $k$-consistency because $k \geq \arity{\StructA_i} + 1$.
		
		For the now  following Step~2 in the cohomological algorithm,
		we show that unsatisfiable $\sig_i$\nobreakdash-components (for both $i \in [2]$) are detected:
		Let $i\in[2]$ and $D$ be a $\sig_i$-component such that $\StructB[D] \notin \CSP{\StructA_i}$.
		Let $X \subseteq D$ be of size at most $k$.
		If we now fix a partial homomorphism $f\colon \StructB[X] \to \StructA_i$
		by setting its variable to $1$,
		we in particular set the partial homomorphism $g\colon \StructB[X] \to \StructA$ with $g(X) = \{c_i\}$ to $0$.
		Now consider Lemma~\ref{lem:fixingIntegerSolutionInOr} but only for the component~$D$.
		One can easily show that $k$-consistency has already discarded the partial homomorphisms
		for which we showed in Claim~\ref{clm:fixingIntegerSolutionInOr-A} and~\ref{clm:fixingIntegerSolutionInOr-B} in the proof of Lemma~\ref{lem:fixingIntegerSolutionInOr} that their variable is set to $0$.
		Then similarly to Claim~\ref{clm:fixingIntegerSolutionInOr-C} in that proof,
		a solution to $\zafkleq{k}{\StructA}{\StructB}{\Hh}$
		where we set~$f$ to~$1$ and the other partial homomorphisms to~$0$,
		has to induce a solution to $\cspiso{k}{\StructA_i}{\StructB[D]}$.
		But for~$\StructA_i$, we know that AIP solves $\CSP{\StructA_i}$.
		This implies that there is no solution to $\zafkleq{k}{\StructA}{\StructB}{\Hh}$ that fixes~$f$ in this sense.
		So similarly to the case of proper or-instances above, we can show that
		$f$ is discarded by the cohomological algorithm.
		
		Now consider the next iteration of the algorithm,
		in which $k$-consistency is executed again.
		As with $\sig_i$-components, the $k$-consistency algorithm detects $S$\nobreakdash-components.
		So let $D$ be an $S$\nobreakdash-component
		in which neither all $\sig_1$-components are solvable nor all $\sig_2$-components are solvable.
		Let $D^i$ be an unsolvable $\sig_i$-component in $D$ and $D^{3-i}$
		a neighbored $\sig_{3-i}$-component of~$D^i$ in~$G_S$,
		which means that there are vertices $u \in D^i$ and $v \in D^{3-i}$
		connected via an $S$-edge.
		All partial homomorphisms mapping vertices in~$D^i$ to~$A_i$ have already been discarded,
		which means that only maps to $c_i$ remain.
		That means that all homomorphisms $f \colon \set{u,v} \to B$
		which map $v$ to $c_{3-i}$ are also discarded because the map $uv \mapsto c_ic_{3-i}$ is not a partial homomorphism.
		In particular, the partial homomorphism $v\to c_{3-i}$ is discarded.
		By the reasoning before, all partial homomorphisms from $D^{3-i}$ to $c_{3-i}$
		get discarded by $k$-consistency.
		This is then propagated to the $\sig_i$-components which are neighbors of $D^{3-i}$ in the sense that in those, the partial homomorphisms to $A_i$ are discarded and only those to $c_i$ remain.
		This propagation continues through the whole $S$-component $D$.
		Because there is an  unsatisfiable $\sig_1$-component and an unsatisfiable $\sig_2$-component in $D$, 
		at some point all partial homomorphisms of some vertex in $D$ get discarded.
		But that means that $k$-consistency rejects the input
		and so the cohomological $k$-consistency algorithm correctly rejects $\StructB$.
	\end{proof}		
	
	After this positive result about the power of cohomological $k$-consistency, we now turn to the NP-complete counterexample.
	
	\begin{theorem}
	\label{thm:cohomologyDoesNotSolveAllCSP}
	There is an NP-complete template structure $\StructA$ such that for every sublinear function in the instance size $k \in o(n)$, the cohomological $k$-consistency algorithm does not solve $\CSP{\StructA}$.
\end{theorem}
	\begin{proof}
		Let $p_1=2$, $p_2=3$ and let $\StructA_i =\CosetGrpTmplt{\ZZ_{p_i}}{3}$ be the template structure
		for ternary $\ZZ_{p_i}$-coset CSPs.
		Set $\StructA := \ORNPC{\StructA_1,\StructA_2}$. 
		We show that the cohomological $k$-consistency algorithm does not solve $\CSP{\StructA}$.
		Let $G = (V,E)$ be a sufficiently large $3$-regular $2$-connected expander graph and $H$ be an orientation of~$G$. Let $k \coloneqq k(|E|)$.
		Let $\lambda_i \colon V(G) \to \ZZ_2$ be zero apart from one arbitrary vertex that is mapped to $1$.
		Let $\StructB_i := \grpCSP{H}{\ZZ_{p_i}}{\lambda_i}$
		and $\StructB := \OR{\StructB_1,\StructB_2}$.
		Consider the following family of sets of partial homomorphisms:
		for each $X \in \tbinom{B}{\leq k}$ define $\Hh(X) \subseteq \Hom{\StructB[X]}{\StructA}$ as follows.
		For $i\in[2]$, let $X_i = X \cap B_i$.
		Let $\Hh_i(X_i)$ be the set of all robustly consistent homomorphisms $\StructB_i[X_i] \to \StructA_i$.
		All $f_1 \in \Hh_1(X_1)$ and $f_2 \in \Hh_2(X_2)$
		induce a partial homomorphism $f \in \Hom{\StructB[X]}{\StructA}$
		by Lemma~\ref{lem:hom-or-intractable-compose}.
		Let $\Hh(X)$ be the sets of all these homomorphisms.
		We show that this family of partial homomorphisms is stable
		under the cohomological $k$-consistency algorithm.
		By Lemma~\ref{lem:robustlyConsistentSurviveKconsistency},
		the robustly consistent partial homomorphisms are not ruled out by $k$-consistency for each $\StructB_i$
		and by Lemma~\ref{lem:hom-or-intractable-k-consistency},
		the induced ones for $\StructB$ are also not ruled out by $k$\nobreakdash-consistency
		on the intractable homomorphism-or construction.
		Let $X \in \tbinom{B}{\leq k}$
		and $f \in \Hh(X)$ be induced by $f_i \in \Hh_i(X_i)$, for an $i \in [2]$.
		By Corollary~\ref{cor:group-csp-p-solution-with-fixed-assignment},
		$\cspiso{k}{\StructA_i}{\StructB_i}$ has a solution $\Phi$
		such that $\Phi(x_{X_i,f_i}) = 1$.
		By Lemma~\ref{lem:hom-or-intractable-solution},
		$\cspiso{k}{\StructA}{\StructB}$ has a solution $\Psi$
		such that $\Psi(x_{X,f}) = \Phi(x_{X_i,f_i}) = 1$.
		Hence indeed, this family of partial homomorphisms is stable under the cohomological $k$-consistency algorithm.
		Since by Lemma~\ref{lem:fHatCanBeExtended}
		none of the sets $\Hh(X)$ is empty,
		$\StructB$ is accepted by the cohomological $k$-consistency algorithm.
		However, $\StructB \notin \CSP{\StructA}$ by Lemma~\ref{lem:hom-or-construction-correct} and
		because $\StructB_i \notin \CSP{\StructA_i}$ for both $i\in[2]$.
		
		To show that $\CSP{\StructA}$ is NP-complete,
		it suffices 
		by Lemmas~\ref{lem:intractable-or-NPC-basic} and~\ref{lem:intractable-or-NPC-reduce}
		that there is an inclusion-wise minimal no-instance $\StructC_i \notin \CSP{\StructA_i}$ of size $3$ for every $i\in [2]$.
		This is, e.g., achieved by the equations $x_1+x_2+x_3 = 1$
		and $x_1+x_2+x_3 = 0$,
		which over both $\ZZ_2$ and $\ZZ_3$ form an inclusion-wise minimal no-instance:
		deleting one variable removes all equations.
	\end{proof}

	\section{Affine Algorithms and Coset-CSPs}	
	\label{sec:groupStuff}
	The counterexample we have used so far is not a coset-CSP itself, but a combination of two Abelian coset-CSPs.
	We now set out to explore the power of the affine algorithms on coset-CSPs.
	Effectively, this is asking the question which coset-CSPs are reducible to a CSP over the infinite Abelian group $(\bbZ,+)$. Our answer is summarized in Theorem \ref{thm:mainPowerOnGroupCSPs}, whose three parts we now prove.
	
	\mainGroupCSPs*
	
		The three parts of this theorem are proved by Theorems \ref{thm:abelianSolvedByAIP}, \ref{thm:AIPsolvesOdd2nilpotent}, and \ref{thm:counterexampleSymmetricGroup} below.
		
	\paragraph*{Abelian Groups.}
	The coset CSPs of Abelian groups are solved by any of the affine algorithms, in fact already by the simplest possible one, which is AIP. This just checks for solvability of the basic affine integer relaxation (AIP) of a CSP. 
	This relaxation is the system $\aipk{k}{\StructA}{\StructB}$ for $k=1$ introduced in Section \ref{sec:BLP}, and 
	every algorithm we have studied clearly solves at least those CSPs that AIP can solve.
	It can be derived from the literature that already AIP suffices to solve all Abelian coset-CSPs:
	\begin{theorem}
		\label{thm:abelianSolvedByAIP}
		If $\Gamma$ is Abelian, then AIP solves $\CSP{\CosetGrpTmplt{\Gamma}{r}}$ for every $r \in \bbN$.
	\end{theorem}	
	\begin{proof}
		Theorem 7.19 in \cite{BartoBKO2021} characterizes the power of AIP in terms of the polymorphisms of the template~$\CosetGrpTmplt{\Gamma}{r}$:
		The CSP is solved correctly by AIP if and only if $\CosetGrpTmplt{\Gamma}{r}$ has alternating functions of all odd arities as polymorphisms. The exact definition of an alternating function is not needed here but it suffices to know that in any Abelian group, the $(2n+1)$-ary function $a(x_1,...,x_{2n+1}) = x_1 - x_2 +x_3 - \cdots + x_{2n+1}$ is alternating \cite[Example 7.17]{BartoBKO2021}.
		We know that $\CosetGrpTmplt{\Gamma}{r}$ has the polymorphism $f(x,y,z) = x-y+z$, so we can use $f$ to generate a $2n+1$-ary function $a$ as above for every $n \in \bbN$, e.g., $a(x_1,x_2,x_3,x_4,x_5) = f(x_1,x_2,f(x_3,x_4,x_5))$, and likewise for higher $n$.  
	\end{proof}

	\paragraph*{2-Nilpotent Groups.}
	Surprisingly, Abelian problems are not the demarcation line for the power of affine algorithms. The following result shows that there exist non-Abelian groups for which AIP still works; these are certain 2-nilpotent groups, so intuitively, they are as close to being Abelian as possible. Formally, a group $\Gamma$ is \emph{2-nilpotent} if its commutator subgroup is contained in its center, i.e.\ the \emph{commutator} $\inv{\alpha}\inv{\beta}\alpha\beta$ of any two $\alpha,\beta \in \Gamma$ commutes with all elements of $\Gamma$. 
	\begin{theorem}
		\label{thm:AIPsolvesOdd2nilpotent}
		If $\Gamma$ is 2-nilpotent and of odd order, then AIP solves $\CSP{\CosetGrpTmplt{\Gamma}{r}}$ for every $r \in \bbN$. 
	\end{theorem}
	\begin{proof}
		We begin with some background. Let $\Gamma=(G, \cdot)$ be a group. For group elements $\alpha, \beta \in \Gamma$, their commutator is defined as $[\alpha, \beta] = \inv{\alpha}\cdot \inv{\beta} \cdot \alpha \cdot \beta$. The commutator subgroup of $\Gamma$ is denoted $[\Gamma, \Gamma]$ and is defined as the group generated by all $[\alpha, \beta]$, for $\alpha, \beta \in \Gamma$. Let $m$ be the exponent of $[\Gamma,\Gamma]$, the least common multiple of the order of all  elements in $[\Gamma,\Gamma]$. This is odd because whenever $|\Gamma|$ and hence $|[\Gamma,\Gamma]|$ are odd.
		For 2-nilpotent groups of odd order, we can apply the so-called ``Baer trick'' \cite{kompatscher2024, Isaacs} to obtain an Abelian group reduct. Define $(G,+)$ as the group on the same universe as $\Gamma$ but with the operation defined as $x + y := xy[x,y]^{(m-1)/2}$. As shown in the proof of \cite[Corollary 5.2]{kompatscher2024}, $(G,+)$ is Abelian. 
		The goal is now to show that the operation $f(x,y,z) = x-y+z$ in $(G,+)$ is a polymorphism of $\CosetGrpTmplt{\Gamma}{r}$. Once we have that, we can obtain alternating operations of all odd arities exactly as in the proof of Theorem \ref{thm:abelianSolvedByAIP}.
		To start with, it is easy to check that the inverse $-x$ in $(G,+)$ is $\inv{x}$. Thus we have
		\begin{align*}
			x-y+z &= x\inv{y}[x,\inv{y}]^{(m-1)/2} \cdot z[x\inv{y}[x,\inv{y}]^{(m-1)/2},z]^{(m-1)/2}\\
			&= x\inv{y}z[x,\inv{y}]^{(m-1)/2} \cdot [x\inv{y},z]^{(m-1)/2} \cdot [[x,\inv{y}]^{(m-1)/2},z]^{(m-1)/2}\\
			&= x\inv{y}z[x,\inv{y}]^{(m-1)/2} \cdot [x\inv{y},z]^{(m-1)/2}.
		\end{align*}
		Here we used that commutators in 2-nilpotent groups are central, the commutator identity $[xy,z] = [x,z] \cdot [y,z]$ that holds in this form in 2-nilpotent groups, and the fact that a commutator that has another commutator as one of its arguments is the neutral element in any 2-nilpotent group.
		Let $d(x,y,z) := x\inv{y}z$, $s:= d(x,y,z)$, and $t := d(z,y,x)$.
		We show by induction on $c$ that the following identity holds:
		\begin{align*}
			d(\dots d(d(t,s,t),s,t)\dots,s,t) = x\inv{y}z[x,\inv{y}]^{c} \cdot [x\inv{y},z]^{c},
		\end{align*}
		where $d$ appears $c$ times in the equation. Then for $c=(m-1)/2$, this identity gives us a term for $x-y+z$ that just uses the Maltsev polymorphism of $\CosetGrpTmplt{\Gamma}{r}$. To prove the identity, consider first the case $c=1$, in which we get:
		\begin{align*}
			x\inv{y}z[x,\inv{y}] \cdot [x\inv{y},z] &= x\inv{y}z \cdot \inv{x}yx\inv{y} \cdot  y\inv{x}\inv{z}x\inv{y}z\\
			&= x\inv{y}z \cdot \inv{x}y\inv{z} \cdot x\inv{y}z = d(t,s,t).
		\end{align*}
		For the inductive step, we have
		\begin{align*}
			x\inv{y}z[x,\inv{y}]^{c+1} \cdot [x\inv{y},z]^{c+1} &= \underbrace{d(...d(d(t,s,t),s,t)...,s,t)}_{c \text{ occurrences of d}} \cdot [x,\inv{y}] \cdot [x\inv{y},z]\\
			&=  \underbrace{d(...d(d(t,s,t),s,t)...,s,t)}_{c \text{ occurrences of d}} \cdot \inv{x}y\inv{z} \cdot x\inv{y}z\\
			&= \underbrace{d(...d(d(t,s,t),s,t)...,s,t)}_{c+1 \text{ occurrences of d}} 
		\end{align*}
		This finishes the proof of the theorem.
	\end{proof}	
	 A known example of a 2-nilpotent group of odd order is the semi-direct product $\bbZ_9 \rtimes \bbZ_3$, which is of order $27$.

	\paragraph*{A Group Coset-CSP Counterexample via Graph Isomorphism.}
	Next, we show that the affine algorithms studied in Section~\ref{sec:power-of-affine} also fail on group-coset-CSPs.
	The key idea is to exploit that  group coset-CSPs are inter-reducible
	with the \emph{bounded color class size graph isomorphism problem}~\cite{BerkholzGrohe2017}.
	For every constant $d$, this is the task to decide whether two vertex-colored graphs, in which at most $d$ many vertices have the same color, are isomorphic.
	Instead of the homomorphism or-construction,
	we use an \emph{isomorphism or-construction}.
	We first reduce our aforementioned Tseitin instances of $\CSP{\CosetGrpTmplt{\ZZ_2}{3}}$ and
	$\CSP{\CosetGrpTmplt{\ZZ_3}{3}}$
	to bounded color class size graph isomorphism.
	Using the isomorphism or-construction, we combine these two isomorphism problems
	in the same fashion as we did with homomorphisms.
	Finally, the resulting bounded color class size graph isomorphism problem
	is translated back into a group coset-CSP over the symmetric group,
	which, on $d$ elements, is denoted by $\Sym{d}$.
	
	\begin{theorem}
		\label{thm:counterexampleSymmetricGroup}
		For every $d \geq 18$ and every constant or at most sublinearly growing $k$, neither the $\ZZ$-affine $k$-consistency relaxation, BA$^{k}$, nor CLAP solve $\CSP{\CosetGrpTmplt{\Sym{d}}{2}}$. 
	\end{theorem}	
	
	\noindent The proof of this theorem spans the rest of this section.

	\subsection{Bounded Color Class Structure Isomorphism and Group Coset-CSPs}
	
	A \defining{colored relational structure} is a pair $(\StructA,\chi_\StructA)$
	of a relational structure $\StructA$ and a function $\chi_\StructA \colon A \to
	\colors$, for some finite set of colors $\colors$,
	that assigns colors to the vertices of $\StructA$.
	A \defining{color class} of $\StructA$ is a maximal set $V \subseteq A$ of elements
	of the same color.
	The \defining{color class size} of $\StructA$ is the maximal size of the color classes of $\StructA$.
	For a set of colors $C \subseteq \colors$, we denote by $\StructA[C]$ the substructure of $\StructA$ induced by all vertices those color is in $C$.
	For s color $c \in \colors$, we write $\StructA[c]$ for $\StructA[\set{c}]$.
	An isomorphism between colored structures has to preserve colors,
	that is, it maps vertices of one color to vertices of the same color.
	For two possibly colored relational structures $\StructA$ and $\StructB$ we write $\isos{\StructA}{\StructB}$ for the set of \defining{isomorphisms} $\StructA \to \StructB$.
	For a number $d \in \bbN$,
	instances of the \defining{$d$-bounded color class size structure isomorphism problem} are pairs $(\StructA, \StructB)$ of relational structures
	of color class size at most $d$. The problem asks whether there is a color-preserving isomorphism from $\StructA$ to $\StructB$.
	Polynomial time reductions in both directions between this problem and group coset-CSPs~\cite{BerkholzGrohe2017} are presented in the following.
	
	\subparagraph*{Reducing Coset-CSPs to Bounded Color Class Size Isomorphism.}
	Let $\Gamma$ be a finite group and $\StructB$ be an $r$-ary $\Gamma$-coset-CSP instance.
	We encode $\StructB$ into a colored graph $\CFIA{\Gamma}{\StructB}$ as follows:
	For every variable $x$ of $\StructB$,
	we add a vertex $(x,\gamma)$ for every $\gamma \in \Gamma$.
	We call $x$ the origin of $(x,\gamma)$ and color
	all vertices with origin $x$ with a fresh color $c_x$.
	For every constraint $C\colon (x_1,\dots, x_r) \in \Delta\delta$,
	add a vertex $(C, \gamma_1,\dots, \gamma_r)$
	for all $(\gamma_1,\dots, \gamma_r) \in \Delta\delta$.
	We call $C$ the origin of these vertices and color
	all vertices with origin $C$ with a fresh color $c_C$.
	We then add edges $\set{(x_i,\gamma_i), (C, \gamma_1,\dots, \gamma_r)}$,
	which we color with fresh colors $c_i'$,
	for all $i \in [r]$ (which formally is encoded in a fresh binary relation symbol).
	Note that, since~$\CFIA{\Gamma}{\StructB}$ is a graph,
	its arity is always~$2$, independently of the arity of~$\StructB$.
	
	We now derive the homogeneous $\Gamma$-coset-CSP $\tilde{\StructB}$ from
	$\StructB$ as follows:
	we replace every constraint $C \colon (x_1, \dots, x_r) \in \Delta\delta$ of $\StructB$
	with the constraint $\tilde{C} \colon (x_1, \dots, x_r) \in \Delta$ in $\tilde{\StructB}$.
	For $\tilde{\StructB}$, we obtain the graph $\CFIB{\Gamma}{\StructB}$
	by the construction before,
	where we identify the colors~$c_{C}$ and~$c_{\tilde{C}}$ for every constraint~$C$.
	The graphs $\CFIA{\Gamma}{\StructB}$ and $\CFIB{\Gamma}{\StructB}$
	are the \defining{CFI graphs} over $\Gamma$ for $\StructB$.
	If~$\StructB$ is the Tseitin equation system over~$\bbZ_2$,
	the obtained CFI graphs correspond to the known CFI graphs
	introduced by Cai, Fürer, and Immerman~\cite{CaiFuererImmerman1992},
	which have found many applications in finite model theory and other areas since then.
	\begin{lemma}[\cite{BerkholzGrohe2015}]
		\label{lem:cfi-basics}
		Let $\Gamma$ be a finite group and $\StructB$ an $r$-ary $\Gamma$-coset-CSP instance.
		\begin{enumerate}
			\item $\CFIA{\Gamma}{\StructB}$ and $\CFIB{\Gamma}{\StructB}$ have color class size at most the maximum of $|\Gamma|$ and $|\Delta|$ over all subgroups $\Delta$ occurring in constraints of $\StructB$,
			which is in particular bounded by
			$|\Gamma|^r$.
			\item $\CFIA{\Gamma}{\StructB} \iso \CFIB{\Gamma}{\StructB}$ if and only if $\StructB \in \CSP{\CosetGrpTmplt{\Gamma}{r}}$.
		\end{enumerate}
	\end{lemma}

	\subparagraph*{Reducing Bounded Color Class Size Isomorphism to Coset-CSPs.}
	Let $(\StructA, \StructB)$ be an instance of the $d$-bounded color class structure isomorphism problem, where the arity of the structures is at most $r$.
	We encode isomorphisms between $\StructA$ and $\StructB$
	as solutions of the following $\Sym{d}$-coset-CSP.
	Denote the set of colors~$\StructA$ and~$\StructB$ by~$\colors$.
	We also assume that $\ell_c := |\StructA[c]| = |\StructB[c]|$ for each color $c \in \colors$.
	Otherwise,~$\StructA$ and~$\StructB$ are trivially non-isomorphic.
	For every $c\in \colors$, we introduce a variable $y_c$.
	First, we add constraints that ensure that $y_c$ is actually a variable over $\Sym{\ell_c}$:
	\[y_c \in \setcond*{\gamma \in \Sym{d}}{\gamma(j) = j \text{ for all } \ell_c \leq j \leq d}.\]
	It is clear that this set is a subgroup of $S_d$ and hence we indeed added $S_d$-constraints.
	Next, for every $c \in \colors$, let the vertices of
	$\StructA[c]$ be $u_{c,1},\dots, u_{c,\ell_c}$ and the ones of $\StructB[c]$
	be $v_{c,1},\dots, v_{c,\ell_c}$.
	We pick, for every set $C = \{c_1,...,c_{r'}\}$ of $r' \leq r$ color classes, an isomorphism $\phi_C \colon \StructA[C] \to \StructB[C]$ if it exists. We identify~$\phi_C$
	with a permutation in $\bigtimes_{i \in [r']} \Sym{\ell_{c_i}}$: The $i$-th component of this tuple of permutations maps
	$j$ to $k$ if $\phi_C(u_{c_i,j}) = v_{c_i,k}$.
	If for some $C$ such an isomorphism $\phi_C$ does not exist, then $\StructA \not\iso \StructB$ and we just add some unsatisfiable constraints and are done (e.g., use two cosets $\set{1}\gamma$, $\set{1}\delta$ for $\gamma \neq \delta$).
	Via these identifications, we add the $r'$-ary $\Sym{d}$-constraint
	\[(y_{c_1},\dots, y_{c_r'}) \in \autgrp{\StructA[C]}\phi_C.\]
	We denote the resulting $\Sym{d}$-coset-CSP by $\bcisosys{\StructA}{\StructB}$.
	For a set $C$ of colors of $\StructA$ and $\StructB$,
	we denote by $\bcisosys{\StructA}{\StructB}[C]$
	the subsystem induced by all variables $y_c$ of colors $c \in C$.

	\begin{lemma}[\cite{BerkholzGrohe2015,KlinLOT2014}]
		\label{lem:iso-system-correct}
		For all $r$-ary colored structures $\StructA$ and $\StructB$ of color class size $d$,
		the structure $\bcisosys{\StructA}{\StructB}$ is an instance of an $r$-ary $S_{d}$-coset-CSP
		such that $\bcisosys{\StructA}{\StructB} \in \CSP{\CosetGrpTmplt{\Sym{d}}{r}}$ if and only if
		$\StructA \iso \StructB$.
	\end{lemma}

	\subsection{Isomorphism OR-Construction on Structures}
	
	We now present the isomorphism-analogue of our previous homomorphism or-construction. It realizes the disjunction of two (or more) structure isomorphism instances again as an instance of structure isomorphism.
	
	A sequence of colored structures $\StructB_1,\dots, \StructB_j$
	is encoded by a colored structure $\langle \StructB_1, \dots, \StructB_j \rangle$
	that is defined as follows:
	Assume $\StructB_i$ uses $\colors_i$ as set of colors
	and, up to renaming colors, assume that all color sets $\colors_i$ are
	pairwise disjoint\footnote{The renaming has to be done canonically, for example, rename color $c \in \colors_i$ to the new color $(c,i)$.
	In this way, the colors of difference sequences get renamed in the same way.}.
	Next, we extend each~$\StructB_i$ by a new binary relation
	that is interpreted as~$B_i^2$.
	We now start with the disjoint union of all~$\StructB_i$,
	where we call vertices of $\StructB_i$ \defining{entry-$i$ vertices}.
	We add a new binary relation symbol
	such that for all $i <j$ we add an edge
	between all entry-$i$ and entry-$j$ vertices to this relation.

	Now let $\StructB_1^0, \dots, \StructB_j^0$ and $\StructB_1^1,\dots, \StructB_j^1$
	be two sequences of colored structures.
	We define a pair of structures $(\StructC^0, \StructC^1) = \ORISO{i\in[j]} {\StructB_i^0,\StructB_i^1}$ as follows.
	For each $k \in \set{0,1}$, define
	\[\StructC^k := \bigdisunion \setcond*{\langle \StructB_1^{a_1}, \dots, \StructB_j^{a_j} \rangle}{a_1+\dots+a_j \equiv k \mod 2},\]
	where we call the $\langle \StructB_1^{a_1}, \dots, \StructB_j^{a_j} \rangle$ \defining{components}.
	\begin{lemma}
		\label{lem:or-construction-basics}
		Let $\StructB_1^0, \dots, \StructB_j^0$ and $\StructB_1^1,\dots, \StructB_j^1$
		be two sequences of colored structures of color class size at most $d$.
		Then for $(\StructC^0, \StructC^1) = \ORISO{i\in[j]} {\StructB_i^0,\StructB_i^1}$ we have
		\begin{enumerate}
			\item $\StructC^0 \iso \StructC^1$ if and only if there exists an $i \in [j]$ such that $\StructB_i^0 \iso \StructB_i^1$, and
			\item $\StructC^0$ and $\StructC^1$ have color class size at most $2^{j-1}d$.
		\end{enumerate}
	\end{lemma}
	\begin{proof}
		The first claim was (for a slightly different encoding of sequences of graphs) shown in~\cite{BerkholzGrohe2017}.
		For the second claim, we note that the encoding of a sequence
		does not increase the color class size
		and that there are $2^{j-1}$ such sequences in the disjoint union.
	\end{proof}

	\subsection{Instances of the Counterexample}
	\label{sec:group-csp-counterexamle-instances}
	
	To obtain instances of $\CSP{\CosetGrpTmplt{\Sym{d}}{2}}$
	that are hard for the affine algorithms,
	we start with Tseitin systems over $\ZZ_2$ and $\ZZ_3$
	and then chain together the former constructions.
	From now on, fix a positive integer $k$.
	As in the proofs in Section~\ref{sec:power-of-affine},
	let $G = (V,E)$ be a $3$\nobreakdash-regular $2$\nobreakdash-connected expander graph sufficiently larger than the width parameter $k$.
	Let $H$ be an arbitrary orientation of $G$.
	Let $p_1 := 2$ and $p_2 := 3$. For $i \in [2]$,
	let $\lambda_i \colon V \to \ZZ_{p_i}$ be defined to be $0$ everywhere except at one arbitrarily chosen vertex $v^* \in V$, where we set $\lambda_i(v^*) := 1$.
	For each $i \in [2]$, we consider the $3$-ary $\ZZ_{p_i}$-coset-CSPs $\StructB_i := \grpCSP{H}{\ZZ_{p_i}}{\lambda_i}$. 
	We apply the reduction to graph isomorphism (see Lemma \ref{lem:cfi-basics}) to obtain for each $i \in [2]$ a pair of colored graphs
	$(\CFIA{\bbZ_{p_i}}{\StructB_i}, \CFIB{\bbZ_{p_i}}{\StructB_i})$ such that $\CFIA{\bbZ_{p_i}}{\StructB_i} \cong \CFIB{\bbZ_{p_i}}{\StructB_i}$ if and only if $\StructB_i \in \CSP{\CosetGrpTmplt{\ZZ_{p_i}}{3}}$.
	By construction, $\StructB_i \notin \CSP{\CosetGrpTmplt{\ZZ_{p_i}}{3}}$, so the corresponding graphs are non-isomorphic. 
	Now apply the graph isomorphism or-construction $(\StructC^0, \StructC^1) =
	\ORISO{i \in [2]}{\CFIA{\bbZ_{p_i}}{\StructB_i},\CFIB{\bbZ_{p_i}}{\StructB_i}}$
	so that $\StructC^0 \cong \StructC^1$ if and only if $\CFIA{\bbZ_{p_1}}{\StructB_1} \cong \CFIB{\bbZ_{p_1}}{\StructB_1}$ or $\CFIA{\bbZ_{p_2}}{\StructB_2} \cong \CFIB{\bbZ_{p_2}}{\StructB_2}$. Since neither of these are isomorphic, we have $\StructC^0 \not\cong \StructC^1$. 
	The two graphs $\StructC^0$ and $\StructC^1$ have bounded color class size and it can in fact be checked that this size is $18$:
	The color class size of $\CFIB{\bbZ_{p_2}}{\StructB_2}$ is upper bounded by~$9$ because there exist~$9$ triples in~$\bbZ_{3}$ whose sum in~$\bbZ_3$ is~$0$. The color class size of $\CFIB{\bbZ_{p_1}}{\StructB_1}$ is smaller. The isomorphism or-construction applied to two graphs doubles the color class size.  
	So with the reduction of bounded color class size isomorphism to a coset-CSP as described above (see Lemma \ref{lem:iso-system-correct}),
	the problem ``$\StructC^0 \cong \StructC^1$?'' is turned into the instance $\bcisosys{\StructC^0}{\StructC^1}$ of $\CSP{\CosetGrpTmplt{\Sym{d}}{2}}$ for every $d \geq 18$. This instance does not admit a solution because $\StructC^0 \not\cong \StructC^1$.
	However, we can show that $\cspiso{k}{\CosetGrpTmplt{\Sym{d}}{2}}{\bcisosys{\StructC^0}{\StructC^1}}$ has an integral solution.
	
	To do so, we pull the notion of a robustly consistent partial homomorphism
	of the Tseitin-systems from Section \ref{sec:tseitin} through  all the constructions, so through the translation of group coset-CSP into bounded color class isomorphism, through the isomorphism or-construction,
	and through the reverse translation of bounded color class isomorphism 
	to group coset-CSPs over symmetric groups.
	\begin{itemize}
		\item Partial homomorphisms of the Tseitin system induce partial isomorphisms of the graph encoding.
		\item Partial isomorphisms of the graph encoding
		induce partial isomorphisms in the isomorphism or-construction.
		\item Finally, partial isomorphisms of the isomorphism or-construction
		induce partial homomorphisms of the encoding as a group coset-CSP over $\Sym{d}$.
	\end{itemize}
	The reverse direction is not always true.
	But for the partial isomorphisms or homomorphisms
	for which this is true,
	we can transfer the notion of robust consistency:
	A partial homomorphism
	$\bcisosys{\StructC^0}{\StructC^1} \to \CosetGrpTmplt{\Sym{d}}{2}$ is
	robustly consistent
	if it is induced
	by a robustly consistent partial homomorphism 
	$\grpCSP{H}{\ZZ_{p_i}}{\lambda_i} \to \CosetGrpTmplt{\ZZ_{p_i}}{r}$
	(we will make this notion precise in the following).
	We show that the properties of robustly consistent homomorphisms from Section~\ref{sec:tseitin} transfer to the group coset-CSP setting in the end:
	\begin{itemize}
		\item Robustly consistent partial solutions of the resulting $\SymStruct{d}{2}$-coset-CSP are also
		not ruled out by $k$\nobreakdash-consistency.
		\item A $p_i$-solution to the width-$k$ affine relaxation of the Tseitin
		system over $\ZZ_{p_i}$ translates to a $p_i$\nobreakdash-solution to the width-$k$ affine relaxation for the resulting $\Sym{d}$-coset-CSP.
		In particular, only variables for robustly consistent partial homomorphisms are non-zero in the solution.
		\item Thus, the width-$k$ affine relaxation of the  $\Sym{d}$-coset-CSP also has
		an integral solution.
	\end{itemize}
	So essentially, all the proofs in Section~\ref{sec:power-of-affine}
	translate to the $\Sym{d}$-coset-CSP.
	These arguments are the technically tedious part of the proof of Theorem \ref{thm:counterexampleSymmetricGroup}.
	We prove this in detail in the following subsection
	but the key source of hardness is the same as in Section \ref{sec:tseitin}. 
	
	\subsection{Proof of Theorem~\ref{thm:counterexampleSymmetricGroup}}
	First of all, we show that $p$-solutions to the width-$k$ affine relaxation for any $\Gamma$-coset-CSP translate to $p$-solutions of the width-$k$ affine relaxation for the $\CSP{\Sym{d}}$-formulation of the corresponding graph isomorphism instance.

	\begin{lemma}
		\label{lem:cfi-p-solution}
		Let $k \in \nat$, $\Gamma$ be a finite group, $\StructB$ an $r$-ary $\Gamma$-coset-CSP,
		and $d$ be the maximum color class size of $\CFIA{\Gamma}{\StructB}$.
		If $\cspiso{kr}{\CosetGrpTmplt{\Gamma}{r}}{\StructB}$ has a $p$-solution,
		then $\cspiso{k}{\SymStruct{d}{2}}{\bcisosys{\CFIA{\Gamma}{\StructB}}{\CFIB{\Gamma}{\StructB}}}$ has a $p$-solution.
	\end{lemma}
	\begin{proof}
		Let $\StructL = \bcisosys{\CFIA{\Gamma}{\StructB}}{\CFIB{\Gamma}{\StructB}}$ and
		let $\Phi$ be a $p$-solution of $\cspiso{kr}{\CosetGrpTmplt{\Gamma}{r}}{\StructB}$.
		We define a $p$-solution $\Psi$ for $\cspiso{k}{\SymStruct{d}{2}}{\StructL}$ as follows.
		Let $\colors$ be the colors of $\CFIA{\Gamma}{\StructB}$.
		Associate with a set of colors $Y \subseteq \colors$
		the set~$\hat{Y}$ of the corresponding elements of~$\StructB$:
		If~$Y$ contains the color of a variable vertex $(x,\gamma)$,
		add~$x$ to~$\hat{Y}$.
		If~$Y$ contains the color of a constraint vertex $(C,(\gamma_1,\dots,\gamma_{r'}))$ for $C\colon (x_1,\dots,x_{r'}) \in \Delta\delta$, add $x_1, \dots, x_{r'}$ to $\hat{Y}$.
		Note that $|\hat{Y}| \leq r |Y|$ because variable-vertices for different variables, and constraint-vertices for different constraints, have different colors, respectively.
		
		Let $Y \in \tbinom{\colors}{\leq k}$ and 
		$f\in \Hom{\StructB[\hat{Y}]}{\CosetGrpTmplt{\Gamma}{r}}$.
		We define a bijection $f' \colon V(\CFIA{\Gamma}{\StructB}[Y]) \to V(\CFIB{\Gamma}{\StructB}[Y])$ via
		\begin{align*}
			f'((x,\gamma)) &:=\left(x, \gamma\inv{f(x)} \right) & \text{ for all } (x,\gamma) \in V(\CFIA{\Gamma}{\StructB}[Y]),\\
			f'((C, \gamma_1,\dots,\gamma_{r
'})) & := \left(C, \gamma_1\inv{f(x_1)} , \dots,  \gamma_{r'}\inv{f(x_{r'})}\right) & \text{for all }(C, \gamma_1,\dots,\gamma_{r'}) \in V(\CFIA{\Gamma}{\StructB}[Y]),
		\end{align*}
		where $C$ is the constraint $C \colon (x_1, \dots, x_{r'}) \in \Delta \delta$.
		Since $f$ is a partial homomorphism, the map $f'$ indeed maps to vertices of $\CFIB{\Gamma}{\StructB}[Y]$ and moreover is a partial isomorphism: 
		If $f$ satisfies a constraint $C$ of $\StructB$, then $(f(x_1),...,f(x_{r'})) \in \Delta\delta$. So for all $(\gamma_1,...,\gamma_{r'}) \in \Delta\delta$, we have $(\gamma_1\inv{f(x_1)},...,\gamma_{r'}\inv{f(x_{r'})}) \in \Delta$. This is exactly the homogeneous version of the constraint, which occurs in $\CFIB{\Gamma}{\StructB}$. Thus, $f' \in \isos{\CFIA{\Gamma}{\StructB}[Y]}{\CFIB{\Gamma}{\StructB}[Y]}$.
		Hence, $f'$ induces a partial homomorphism $\hat{f} \in \Hom{\StructL[Y]}{\SymStruct{d}{2}}$.
		For all $g \in \Hom{\StructL[Y]}{\SymStruct{d}{2}}$, define
		\begin{align*}
			\Psi(x_{Y,g}) &:= \begin{cases}
				\Phi(x_{\hat{Y},f}) & \text{if } g = \hat{f} \text{ for some } f\in \Hom{\StructB[\hat{Y}]}{\CosetGrpTmplt{\Gamma}{r}},\\
				0 & \text{otherwise.}
			\end{cases}
		\end{align*}
		We say that the \defining{partial homomorphism $g$
			corresponds to $f$} in the equation above.
		Likewise, the \defining{variable $x_{Y,g}$ of $\cspiso{k}{\SymStruct{d}{2}}{\bcisosys{\CFIA{\Gamma}{\StructB}}{\CFIB{\Gamma}{\StructB}}}$ 
			corresponds} to the variable $x_{\hat{Y},f}$ of $\cspiso{kr}{\CosetGrpTmplt{\Gamma}{r}}{\StructB}$.
		
		We show that $\Psi$ is a solution to 
		$\cspiso{k}{\SymStruct{d}{2}}{\StructL}$,
		which then is obviously a $p$-solution.
		We first consider the equations of Type~\ref{eqn:csp-iso-agree}.
		Let $Y \in \tbinom{\colors}{\leq k}$,
		$c \in Y$, and $g' \in \Hom{\StructL[Y \setminus \set{c}]}{\SymStruct{d}{2}}$.
		First assume that there is a $g \in \Hom{\StructB[\widehat{Y\setminus \set{c}}]}{\CosetGrpTmplt{\Gamma}{r}}$ such that $g' = \hat{g}$.
		Then, exploiting Lemma~\ref{lem:csp-iso-subsets},
		\begin{align*}
			\sum_{\substack{f' \in \Hom {\StructL[Y]}{\SymStruct{d}{2}},\\ \restrict{f'}{Y\setminus\set{c}}= \hat{g}}} \Psi(x_{Y,f'})  &=
			\sum_{\substack{f \in \Hom{\StructB[\hat{Y}]}{\CosetGrpTmplt{\Gamma}{r}},\\ \restrict{f}{\widehat{Y \setminus \set{c}}} = g}} \Phi(x_{\hat{Y},f})
			= \Phi(x_{\hat{Y},g}) = \Psi(x_{Y, \hat{g}}).
		\end{align*}
		Second assume that there is no $g \in \Hom{\StructB[\widehat{Y\setminus \set{c}]}}{\CosetGrpTmplt{\Gamma}{r}}$ such that $g' = \hat{g}$.
		Then for every partial homomorphism $ f'\in \Hom {\StructL[Y]}{\SymStruct{d}{2}}$ (via the identification of permutation on each color class with the $S_d$\nobreakdash-variables) such that $\restrict{f'}{Y\setminus\set{c}} = g'$,
		there is also no $f \in \Hom{\StructB[\hat{Y}]}{\CosetGrpTmplt{\Gamma}{r}}$ such that $f' = \hat{f}$.
		Hence both sides of Equation~\ref{eqn:csp-iso-agree} are $0$.
		
		Finally, consider Equation~\ref{eqn:csp-iso-empty}:
		we have $\Psi(x_{\emptyset,\emptyset}) = \Psi(_{\emptyset,\hat{\emptyset}}) = \Phi(x_{\emptyset,\emptyset}) = 1$
		because the empty homomorphism $\emptyset \colon \StructB[\emptyset] \to \CosetGrpTmplt{\Gamma}{r}$
		induces the empty homomorphism $\hat{\emptyset} \colon \StructL[\emptyset] \to \SymStruct{d}{2}$.
	\end{proof}
	
	\noindent In the setting of the previous proof,
	we show that if a partial homomorphism of $\StructB$
	is not discarded by the $k$-consistency algorithm,
	then the corresponding one of $\StructL$ is not discarded, either.
	
	\begin{lemma}
		\label{lem:cfi-k-consistency}
		Let $k \in \nat$, let $\Gamma$ be a finite group and $\StructB$ an $r$-ary $\Gamma$-coset-CSP,
		and let $\StructL = \bcisosys{\CFIA{\Gamma}{\StructB}}{\CFIB{\Gamma}{\StructB}}$.
		Let $\colors$ be the set of colors of $\CFIA{\Gamma}{\StructB}$.
		For all $Z \in \tbinom{B}{rk}$, $Z' \in \tbinom{\colors}{k}$,
		$f \in \Hom{\StructB[Z]}{\CosetGrpTmplt{\Gamma}{r}}$,
		and $g \in \Hom{\StructL[Z']}{\CosetGrpTmplt{\Gamma}{r}}$
		such that $f$ corresponds to $g$
		(in the sense of the proof of Lemma~\ref{lem:cfi-p-solution}),
		if $f \in \kcol{rk}{\CosetGrpTmplt{\Gamma}{r}}{\StructB}(Z)$
		then $g \in \kcol{k}{\SymStruct{d}{2}}{\StructL}(Z')$.	
	\end{lemma}
	\begin{proof}
		Recall that if $f$ corresponds to $g$,
		we have $g = \hat{f}$ and $Z = \widehat{Z'}$.
		We consider the sets of partial homomorphisms
		$\setcond{\hat{f}}{f \in \kcol{rk}{\CosetGrpTmplt{\Gamma}{t}}{\StructB}(\hat{X})}$
		for every $X \in \tbinom{\colors}{\leq k}$
		and show that this collection satisfies the down-closure and forth-condition.
		
		Let $Y \subseteq X \in \tbinom{\colors}{\leq k}$.
		To show the down-closure,
		let $f \in \kcol{rk}{\CosetGrpTmplt{\Gamma}{t}}{\StructB}(\hat{X})$.
		Because $Y \subseteq X$, we have $\hat{Y} \subseteq \hat{X}$.
		From the down-closure of $\kcol{rk}{\CosetGrpTmplt{\Gamma}{r}}{\StructB}$ it
		follows that $\restrict{f}{\hat{Y}} \in \kcol{rk}{\CosetGrpTmplt{\Gamma}{t}}{\StructB}(\hat{Y})$.
		By the construction of the corresponding homomorphisms,
		we have $\restrict{\hat{f}}{Y} = \widehat{\restrict{f}{\hat{Y}}}$.
		
		The forth-condition is similarly inherited from $\kcol{rk}{\CosetGrpTmplt{\Gamma}{r}}{\StructB}$:
		Let $g \in \kcol{rk}{\CosetGrpTmplt{\Gamma}{t}}{\StructB}(\hat{Y})$.
		Then $g$ extends to some $f \in \kcol{rk}{\CosetGrpTmplt{\Gamma}{t}}{\StructB}(\hat{X})$
		by the forth-condition of $ \kcol{rk}{\CosetGrpTmplt{\Gamma}{r}}{\StructB}$.
		Again by the construction of corresponding homomorphisms,
		we have that $\hat{f}$ extends $\hat{g}$.
	\end{proof}
	
	The previous lemmas establish the link between coset-CSPs and their isomorphism formulation. The next step is to deal with the isomorphism or-construction.
	We extend the notion of entry-$\ell$ vertices from the encoding of sequences
	to the isomorphism or-construction:
	For $\ell \in [j]$, we call a vertex of $\StructC^0$ or $\StructC^1$
	an entry-$\ell$ vertex
	if it is a an entry-$\ell$ vertex of some component $\langle \StructB_1^{a_1}, \dots, \StructB_j^{a_j} \rangle$.
	
	For the following, fix an $i \in [j]$. We now describe how partial isomorphisms between~$\StructB_i^0$ and~$\StructB_i^1$ can be extended to partial isomorphisms of~$\StructC_0$ and~$\StructC_1$.
	We fix a bijection $b$ between the components of~$\StructC^0$ and~$\StructC^1$, that is, between the structures $\langle \StructB_1^{a_1}, \dots, \StructB_j^{a_j} \rangle$ with even and odd sum of the $a_\ell$,
	such that identified components only differ in entry $i$:
	\[b(\langle \StructB_1^{a_1}, \dots, \StructB_i^{a_i}, \dots, \StructB_j^{a_j} \rangle ) = \langle \StructB_1^{a_1}, \dots,\StructB_i^{1-a_i}, \dots, \StructB_j^{a_j} \rangle.\]
	Using the identity map on $\StructB_\ell^0$ and $\StructB_\ell^1$ for all $\ell \neq i$,
	the bijection~$b$ induces a bijection $\hat{b}$ between
	the vertices of these components apart from the entry-$i$ vertices.
	
	Let~$X$ be a set of colors of~$\StructC^0$ and~$\StructC^1$
	and denote by $\restrict{X}{i}$ the set of colors of~$\StructB_i^0$ and~$\StructB_i^1$ that occur (after the possible renaming to encode sequences)
	in~$X$.
	We define the function $\iota_i^X\colon \isos{\StructB_i^0[\restrict{X}{i}]}{\StructB_i^1[\restrict{X}{i}]} \to
	\isos{\StructC^0[X]}{\StructC^1[X]}$ for every set of colors $X$. It essentially defines an extension for each partial isomorphism with domain $X|_i$ to the whole color class $X$.
	For a partial isomorphism $f \in \isos{\StructB_i^0[\restrict{X}{i}]}{\StructB_i^1[\restrict{X}{i}]}$,
	the function $\iota_i^X(f)$ is defined as follows:
	\begin{itemize}
		\item Let $v$ be an entry-$i$ vertex of a component $D=\langle \StructB_1^{a_1}, \dots, \StructB_i^{a_j}\rangle$.
		If $a_i = 0$, then $\iota_i^X(f)$ maps $v$ to an entry-$i$ vertex of $b(D)$ according to $f$
		(when seeing $v$ as a vertex of $\StructB_i^0$).
		If $a_i = 1$, then we proceed as in the previous case using $\inv{f}$ instead of $f$.
		\item Otherwise, $\iota_i^X(f)$ maps $v$ to $b(v)$.
	\end{itemize}
	Intuitively, $\iota_i^X(f)$ maps all components in $\StructC^0[X]$ to the corresponding ones in $\StructC^1[X]$ according to $b$
	and uses $f$ or $\inv{f}$, respectively, for the $i$-th entry.
	
	\begin{lemma}
		\label{lem:or-construction-p-solution}
		Fix $k \in \nat$, let $\StructB_1^0, \dots, \StructB_j^0$ and $\StructB_1^1,\dots, \StructB_j^1$
		be two sequences of colored structures of arity at most $r$ and color class size at most $d$,
		and let $(\StructC^0, \StructC^1) = \ORISO{i\in [j]} {\StructB_i^0,\StructB_i^1}$.
		Assume~$\colors$ is the set of colors of~$\StructC^0$ and~$\StructC^1$,
		$\StructL = \bcisosys{\StructC^0}{\StructC^1}$, and
		$\StructL_i = \bcisosys{\StructB_i^0}{\StructB_i^1}$ for all $i \in [j]$.
		If, for some $i \in [j]$, 
		the equation system $\cspiso{k}{\SymStruct{d}{r}}{\StructL_i}$
		has a $p$-solution $\Phi$,
		then the equation system 
		$\cspiso{k}{\SymStruct{d}{r}}{\StructL}$
		has the  $p$-solution $\Psi$
		defined, for all $X \in \tbinom{\colors}{\leq k}$ and $g \in \Hom{\StructL[X]}{\SymStruct{d}{r}}$,
		via
		\[
		\Psi(x_{X,g}) :=  \begin{cases}
			\Phi(x_{\restrict{X}{i},f}) & \text{if } \iota_i^X(f(X|_i)) = g(X) \text{ for some } f \in \Hom{\StructL_i[\restrict{X}{i}]}{\SymStruct{d}{r}},\\
			0 & \text{otherwise.}
		\end{cases}
		\]
		We say that the \defining{partial homomorphism $g$
			corresponds} to $f$ in the equation above
		and that the \defining{variable $x_{X,g}$ of $\cspiso{k}{\SymStruct{d}{r}}{\StructL}$
			corresponds} to the variable $x_{\restrict{X}{i},f}$
		of $\cspiso{k}{\SymStruct{d}{r}}{\StructL_i}$.
	\end{lemma}
	\begin{proof}
		First consider the equations of Type~\ref{eqn:csp-iso-agree}:
		Recall that $\StructL$
		has a variable for every color class.
		Let $X \in \tbinom{\colors}{\leq k}$ be a set of at most $k$ colors,
		let $c \in X$ be a color, and 
		$g \in \Hom{\StructL[X \setminus \set{c}]}{\SymStruct{d}{r}}$.
		First assume that there is an $f \in \Hom{\StructL_i[\restrict{X}{i}]}{\SymStruct{d}{r}}$ such that  $\iota_i^X(f) = g$. Then
		\begin{align*}
			\sum_{\substack{h\in \Hom{\StructL[X]}{\SymStruct{d}{r}},\\ \restrict{h}{X\setminus \set{c}}= g}} \Psi(x_{X,h})
			&= \sum_{\substack{h\in \Hom{\StructL_i[\restrict{X}{i}]}{\SymStruct{d}{r}},\\ \restrict{\iota_i^X(h)}{X \setminus \set{c}} = g}} \Phi(x_{\restrict{X}{i},h})\\
			&= \sum_{\substack{h\in \Hom{\StructL_i[\restrict{X}{i}]}{\SymStruct{d}{r}},\\ \restrict{h}{X\setminus \set{c}}= f}} \Phi(x_{\restrict{X}{i},h})\\
			&=\Phi(x_{\restrict{X}{i}, f}) = \Psi(x_{X,g}).
		\end{align*}
		Assume otherwise that there is no such $f$.
		Then $\Psi(x_{X, g}) = 0$.
		But in this case, every partial homomorphism $h\in \Hom{\StructL[X]}{\SymStruct{d}{r}}$
		is not in the image of $\iota_i$, which means that
		both sides of Equation~\ref{eqn:csp-iso-agree} are zero.
		It remains to check Equation~\ref{eqn:csp-iso-empty}.
		Since the empty homomorphism is the image of the empty homomorphism under $\iota^i_k$,
		Equation~\ref{eqn:csp-iso-empty} for $\cspiso{k}{\SymStruct{d}{r}}{\StructL}$
		follows from Equation~\ref{eqn:csp-iso-empty}
		for $\cspiso{k}{\SymStruct{d}{r}}{\StructL_i}$.
		It is clear that the solution is a $p$-solution.
	\end{proof}
	
	\noindent The previous lemma shows that $p$-solutions for one entry translate to a $p$-solution of the whole or-construction.
	We now show a similar statement for the $k$-consistency algorithm.
	\begin{lemma}
		\label{lem:or-construction-k-consistency}
		Fix $k \in \nat$, let $\StructB_1^0, \dots, \StructB_j^0$ and $\StructB_1^1,\dots, \StructB_j^1$
		be two sequences of colored structures of arity at most $r$ and color class size at most $d$,
		and let $(\StructC^0, \StructC^1) = \ORISO{i\in [j]} {\StructB_i^0,\StructB_i^1}$.
		Let $\colors$ be the set of colors of $\StructC^0$ and $\StructC^1$,
		and let $\StructL = \bcisosys{\StructC^0}{\StructC^1}$.
		For every $i \in [j]$,
		let $\colors_i$ be the set of colors of $\StructB^0_i$ and $\StructB^1_i$,
		and let $\StructL_i = \bcisosys{\StructB_i^0}{\StructB_i^1}$.
		Let $i \in [j]$,  $X \in \tbinom{\colors}{\leq k}$, and $f \in \Hom{\StructL_i[\restrict{X}{i}]}{\SymStruct{d}{r}}$.
		If $f \in \kcol{k}{\SymStruct{d}{r}}{\StructL_i}(\restrict{X}{i})$,
		then for every  $g \in \Hom{\StructL[X]}{\SymStruct{d}{r}}$
		that corresponds to $f$
		we have that $g \in \kcol{k}{\SymStruct{d}{r}}{\StructL}(X)$.
	\end{lemma}
	\begin{proof}
		We show that the collection of partial homomorphisms
		$\iota_i^X(\kcol{k}{\SymStruct{d}{r}}{\StructL_i}(\restrict{X}{i}))$ for all $X \in \tbinom{\colors}{\leq k}$
		satisfies the down-closure and forth-condition property.
		Then they have to be included in the greatest fixed-point computed
		by the $k$-consistency algorithm.
		
		The down-closure is inherited from $\kcol{k}{\SymStruct{d}{r}}{\StructL_i}$:
		Let  $Y \subset X \in \tbinom{\colors}{\leq k}$
		and $f \in \iota_i^X(\kcol{k}{\SymStruct{d}{r}}{\StructL_i}(\restrict{X}{i}))$.
		Then there is a $g \in \kcol{k}{\SymStruct{d}{r}}{\StructL_i}(\restrict{X}{i})$
		such that $\iota_i^X(g) = f$.
		By the down-closure of $\kcol{k}{\SymStruct{d}{r}}{\StructL_i}$,
		we have that $\restrict{g}{\restrict{Y}{i}} \in 
		\kcol{k}{\SymStruct{d}{r}}{\StructL_i}(\restrict{Y}{i})$.
		Hence $h \coloneqq \iota_i^Y(\restrict{g}{\restrict{Y}{i}}) \in
		\iota_i^Y(\kcol{k}{\SymStruct{d}{r}}{\StructL_i}(\restrict{Y}{i}))$.
		Because $\hat{b}$ is a bijection
		and by the definition of $\iota_i^X$,
		it follows that $\restrict{f}{Y} = h$.
		
		To show the forth-condition,
		let  $Y \subset X \in \tbinom{\colors}{\leq k}$
		and $f \in \iota_i^Y(\kcol{k}{\SymStruct{d}{r}}{\StructL_i}(\restrict{Y}{i}))$.
		Again, there is a $g \in\kcol{k}{\SymStruct{d}{r}}{\StructL_i}(\restrict{Y}{i})$
		such that $\iota_i^Y(g) = f$.
		If $\restrict{Y}{i} = \restrict{X}{i}$, 
		then we extend $f$ to $X$ using $b$ yielding $g$.
		By the definition of $\iota_i$, we have that
		$g \in \iota_i^X(\kcol{k}{\SymStruct{d}{r}}{\StructL_i}(\restrict{X}{i}))$.
		Otherwise, $\restrict{Y}{i} \subset \restrict{X}{i}$.
		By the forth-condition for $\kcol{k}{\SymStruct{d}{r}}{\StructL_i}$,
		there is an $h \in \kcol{k}{\SymStruct{d}{r}}{\StructL_i}(\restrict{X}{i})$
		that extends~$g$.
		But then $\iota_i^X(h) \in \iota_i^X(\kcol{k}{\SymStruct{d}{r}}{\StructL_i}(\restrict{X}{i}))$
		and $\iota_i^X(h)$ extends $f$
		by the definition of $\iota_i^X$.
	\end{proof}
	
	Finally, we are ready to prove Theorem~\ref{thm:counterexampleSymmetricGroup}.
	\begin{proof}[Proof of Theorem~\ref{thm:counterexampleSymmetricGroup}]
		Fix a $k \in \nat$ and $d\geq 18$.
		We construct unsatisfiable $\SymStruct{d}{2}$-instances as described in
		Section~\ref{sec:group-csp-counterexamle-instances}.
		Set $p_1 = 2$ and $p_1=3$.
		Robustly consistent partial homomorphism
		of the Tseitin system $\StructB_i:= \grpCSP{H}{\ZZ_{p_i}}{\lambda_i}$
		are not ruled out by $k$-consistency by Lemma~\ref{lem:robustlyConsistentSurviveKconsistency}.
		Lemma~\ref{lem:cfi-k-consistency}
		shows that the corresponding partial homomorphisms of
		$\StructL_i := \bcisosys{\CFIA{\ZZ_{p_i}}{\StructB_i}}{\CFIB{\ZZ_{p_i}}{\StructB_i}}$
		for both $i\in [2]$
		also survive $k$\nobreakdash-consistency.
		Let  $(\StructC^0, \StructC^1) =
		\ORISO{i \in [2]}{\CFIA{\bbZ_{p_i}}{\StructB_i},\CFIB{\bbZ_{p_i}}{\StructB_i}}$.
		Lemma~\ref{lem:or-construction-k-consistency}
		shows that also the corresponding partial homomorphisms of $\bcisosys{\StructC^0}{\StructC^1}$
		are not ruled out by the $k$-consistency algorithm.
		
		Now consider solutions of the width-$k$ affine relaxation.
		By Lemma~\ref{lem:group-csp-p-solution},
		there is a $p_i$-solution for $\cspiso{k}{\CosetGrpTmplt{\ZZ_{p_i}}{3}}{\StructB_i}$
		for both $i \in[2]$
		in which only variables for robustly consistent partial homomorphisms
		are non-zero.
		Lemma~\ref{lem:cfi-p-solution} shows that such $p_i$-solutions
		also exist for $\cspiso{k}{\SymStruct{q_i}{2}}{\StructL_i}$ , where $q_i = p_i^2$,
		for both $i\in [2]$.
		These solutions are non-zero only for variables
		of partial homomorphisms that correspond to robustly consistent ones of the Tseitin systems.
		The domain size $p_i^2$ comes from Lemma~\ref{lem:cfi-basics}
		and the fact that the coset-constraints of the ternary Tseitin systems use subgroups of order $p_i^2$.
		Lemma~\ref{lem:or-construction-p-solution}
		provides $p_i$-solutions for 
		$\cspiso{k}{\SymStruct{18}{2}}{\bcisosys{\StructC^0}{\StructC^1}}$ for both $i\in[2]$,
		for which again only variables are set to a non-zero
		value for partial homomorphisms corresponding to robustly consistent ones.
		Here the domain size is $2\max\set{p_1^2,p_2^2}=18$
		and comes from Lemma~\ref{lem:or-construction-basics}.
		
		Now for the $\ZZ$-affine $k$-consistency relaxation,
		the proof proceeds exactly as the one of Theorem~\ref{thm:z-affine-does-not-solve-bounded-color-class}.
		For BA$^k$,
		we proceed as in the proof of Theorem~\ref{thm:BLP-does-not-solve-bounded-color-class},
		where we again note that the $p_i$-solutions
		exactly set the variables for partial homomorphisms corresponding to robustly consistent ones of the Tseitin system to a non-zero value.
		And finally for CLAP,
		we proceed as in the proof of Theorem~\ref{thm:clap-does-not-solve-all}.
		Here we again use Lemma~\ref{cor:group-csp-p-solution-with-fixed-assignment}
		to show that we can set the variable of
		a single robustly consistent solution in 
		$\cspiso{k}{\CosetGrpTmplt{\ZZ_{p_i}}{3}}{\StructB_i}$ to $1$,
		which then travels through Lemmas~\ref{lem:cfi-p-solution}
		and~\ref{lem:or-construction-p-solution}
		to the corresponding variable of 
		$\cspiso{k}{\SymStruct{18}{2}}{\bcisosys{\StructC^0}{\StructC^1}}$.
	\end{proof}

	\section{Conclusion}
	Regarding the question for a universal polynomial-time CSP algorithm, we conclude that most of the affine algorithms from recent years are not powerful enough to accomplish this, not even on Maltsev templates. 
	
	The remaining candidates are essentially all affine algorithms that set local solutions to $1$ when solving the affine relaxation. The one we have focused on is \emph{cohomological $k$\nobreakdash-consistency} but there are others with this feature, for example \emph{C(BLP+AIP)}, a variation of CLAP mentioned in \cite{CiardoZivny2023CLAP} and defined more explicitly in \cite{ChanNg}. We expect that it solves our counterexample, too. In \cite{ChanNg}, it is shown that C(BLP+AIP) fails on certain intractable templates, though again, a tractable counterexample is not known. In \cite{zhuk2025singletonalgorithms}, the technique of fixing local solutions in the affine relaxation is called \emph{Singleton-AIP}, and singleton variants of the other algorithms, such as BLP+AIP, are considered, too.
	In fact, already Singleton-AIP \cite{zhuk2025singletonalgorithms}, which is subsumed by cohomological $k$-consistency, is a candidate for a universal algorithm that has not been disproved yet.
	Cohomological $k$-consistency can be seen as a hierarchy parameterized by $k$ that is built on top of Singleton-AIP, combined with $k$-consistency and iteration. Such hierarchies of singleton algorithms are not considered in \cite{zhuk2025singletonalgorithms}, and gaining a deeper understanding of these seems like the most important next step to advance this research direction. Possibly, the minion-theoretic methods that are beginning to emerge in \cite{zhuk2025singletonalgorithms} will be vital for this. 
	
	Another question that we have not addressed in this article concerns the relationship between the different algorithms. It is obvious from the definitions that cohomological $k$-consistency subsumes $\bbZ$-affine $k$-consistency, and that CLAP subsumes BLP+AIP. How the $k$\nobreakdash-consistency based methods compare to the BLP-based ones remains unanswered; it may be that they are incomparable. In particular, we would like to know if cohomological $k$-consistency strictly subsumes all the other algorithms. In light of our results, this seems likely, but since the cohomological algorithm does not use the BLP, it is not obvious how it compares to, say, BA$^{k}$.

	\bibliographystyle{plainurl}
	\bibliography{journalVersion.bib}

\end{document}